\documentclass[11pt,a4paper]{article}
\usepackage[round]{natbib}
\usepackage{fancyhdr}
\usepackage{mathrsfs}
\usepackage{amsmath,amsfonts,amssymb,amsthm}
\usepackage{dsfont}
\usepackage{physics}
\usepackage{xcolor}
\usepackage{algorithm}
\usepackage{indentfirst}
\usepackage{booktabs}
\usepackage{graphicx}
\usepackage{subcaption}
\usepackage{url}
\usepackage{hyperref}

\newtheorem{theorem}{Theorem}[section]

\newtheorem{lemma}[theorem]{Lemma}
\newtheorem{prop}{Proposition}[section]
\theoremstyle{definition}
\newtheorem{definition}{Definition}[section]

\usepackage{authblk}
\usepackage{parskip}
\usepackage{xurl}

\usepackage[bottom=2.5cm,top=2.5cm,left=2.5cm,right=2.5cm]{geometry}

\DeclareMathOperator{\argmax}{argmax}
\usepackage[utf8]{inputenc}  
\usepackage[T1]{fontenc}

\begin{document}

\providecommand{\keywords}[1]{\textbf{\textit{Keywords --}} #1}

\title{High-Frequency Market Manipulation Detection with a Markov-modulated Hawkes process}
\author[1,2]{Timothée Fabre}
\author[1]{Ioane Muni Toke}
\affil[1]{Université Paris-Saclay, CentraleSupélec, Laboratoire MICS, France.}
\affil[2]{SUN ZU Lab, Paris, France.}
\date{\today}

\maketitle

\begin{abstract}
\noindent This work focuses on a self-exciting point process defined by a Hawkes-like intensity and a switching mechanism based on a hidden Markov chain. Previous works in such a setting assume constant intensities between consecutive events. We extend the model to general Hawkes excitation kernels that are piecewise constant between events. We develop an expectation-maximization algorithm for the statistical inference of the Hawkes intensities parameters as well as the state transition probabilities. The numerical convergence of the estimators is extensively tested on simulated data. Using high-frequency cryptocurrency data on a top centralized exchange, we apply the model to the detection of anomalous bursts of trades. We benchmark the goodness-of-fit of the model with the Markov-modulated Poisson process and demonstrate the relevance of the model in detecting suspicious activities.

\vspace{\baselineskip}

\noindent \keywords{Hawkes process, Regime switching, Cryptocurrency, Wash trading, Ramping, Price manipulation}
\end{abstract}
\vspace{\baselineskip}

\hrule
\setcounter{tocdepth}{2}
\tableofcontents
\vspace{0.5\baselineskip}
\hrule
\vspace{0.5\baselineskip}

\section{Introduction}

Due to their ability to model events that cluster over time, Hawkes processes have been extensively studied in recent years. Their interpretability has been highlighted in many areas of application, for example in finance. Mathematically, a univariate linear Hawkes process is a counting process $(N_t)_{t\geq 0}$ with an intensity process $(\lambda_t)_{t\geq 0}$ that depends on the history of the process $N$. The intensity is usually defined as
\begin{equation}
    \lambda_t=\mu+\int_{(-\infty,t]}\phi(t-s)\mathrm{d}N_s,
\end{equation}
where $\mu$ is the baseline intensity and $\phi$ is the kernel, a non-negative function that quantifies the excitation effect of past events.

In many real-world applications, the parameters of the Hawkes process are not static and can change over time in response to exogenous factors or regimes. For instance, in electronic financial markets, the intensity of trades can vary with respect to microstructure variables such as liquidity and order flow imbalances. To capture such dynamics, we propose a new point process with a Hawkes-like intensity whose parameters vary according to a regime-switching mechanism based on a continuous time Markov chain (CTMC).

\subsection{Literature review}

The use of Hawkes processes has been explored extensively in the microstructure literature as they are natural candidates for limit order book models. Non-linear extensions can capture inhibitory behaviors and lead to a better reproduction of order flow dynamics \citep{lu2018high}. Although significant progress has been made towards the discovery of the true kernel shape \citep{bacry2016first, fosset2022non, fabre2024neural}, careful empirical studies suggest that the intensity itself depends on exogenous factors such as bid-ask spread and liquidity imbalance, \textit{e.g.} see \cite{muni2017modelling} and \cite{muni2020analyzing}. For example, extreme bid-ask spreads are often correlated with higher market uncertainty, which potentially increases market activity.

State-dependent Hawkes processes are designed to infer the relationship between intensities and market features. The first \citep{morariu2022state} postulates that the parameters of the kernel depend on the state of the market that is to be defined. Each pre-defined state is associated with a corresponding set of unknown parameters. Although the framework provides great flexibility, it is highly sensitive to the market characterization provided as input. The other approach specifies the intensity process as a multiplicative model, composed of a linear Hawkes process intensity and a state-dependent factor \citep{sfendourakis2023lob,wu2022queue}. In this case, the parameters of the Hawkes intensity do not depend on the market state, but are rather multiplied by a factor that is a function of a features vector.

We decide to take a different point of view and do not explicitly characterize the market state with observable quantities. Instead, we consider that the parameters of the Hawkes process depend on a continuous time Markov chain (CTMC) that can explore a finite number of states that are not observed. In fact, each state may encode a specific configuration of market features that can be characterized once the parameters of the model are estimated. The idea is to use an unsupervised technique to identify states of the market that are associated with different levels of endogeneity and exogeneity. The Markov-modulated Poisson process has been applied many times to real-world problems, see \cite{meier1984statistical}, \cite{kawashima1990teletraffic}, \cite{davison1993stochastic} and more recently in \cite{bergault2024price}. Nevertheless, the Markov-modulated Hawkes process (MMHP) has received much less attention due to the multiple challenges that arise in their estimation. A workaround was proposed in \citet{wang2010statistical, wang2012markov} by considering a stepwise decay instead of a continuous one, which allowed the authors to design an expectation-maximization (EM) algorithm for parameter estimation. In a more recent work, \cite{wu2022markov} propose to use Bayesian inference for the estimation of a MMHP with one Poisson regime and one Hawkes regime on small data sets. 

As an application, we are interested in detecting suspicious trading activity on illiquid cryptocurrencies through the unsupervised identification of extreme bursts of trades. In particular, we focus on a particular wash trading tactic, which involves the execution of a large number of fake orders within a short period of time. The goal is to generate a significant increase in both traded volume and trade intensity in order to mislead other agents. Wash trading activity has been studied on both traditional and cryptocurrency markets, see e.g., \cite{imisiker2018wash, nekrasova2023does, aloosh2024direct}. When the agents identification numbers are available, graph theory can be a powerful tool for the detection of suspicious transactions \citep{cao2015detecting, victor2021detecting, chen2023dark, tovsic2023beyond, jayant2023economics}. Nevertheless, traders' IDs -- even anonymous -- are not provided by centralized venues, which makes wash trading detection extremely challenging. Many indirect detection methodologies have been proposed -- see, e.g. \cite{cong2023crypto} -- and rely on anomaly identification in either first digit distribution of trade sizes or price impact of time-bucketed signed volumes. The latter technique is intuitive and convenient: set a time window $W$, \textit{e.g.} $W=10$ minutes, discretize the observation time frame in buckets of size $W$ and for each bucket, compute the mid-price move and the signed traded volume. If wash trading activity occurred in a bucket, then it is likely that a large signed traded volume together with a small price move has been observed. Many detection methods can follow from this simple framework, like fitting a linear regression model and identifying outliers corresponding to small price moves. The problem of such indirect methodologies is that it is extremely sensitive to arbitrary choices such as the chosen time window $W$, the price move model, the characterization of a suspicious bucket, and it is not suitable for live detection. Furthermore, in a recent work, \cite{falk2023crypto} demonstrate how current indirect techniques can fail substantially for fake volume identification. To the best of our knowledge, our work is the first to propose an application of Markov-modulated point processes to the detection of suspicious events in a trading environment.

\subsection{Contribution and outline of the paper}

Our main contributions are summarized below.
\begin{itemize}
    \item We introduce a novel self-exciting point process whose intensity is piecewise constant between event times, with a discretization step parameter $\delta$.
    \item Inspired by \cite{wang2010statistical}, we develop an estimation procedure for such processes with a general non-negative, decreasing kernel function. We provide a Viterbi algorithm for state estimation in both historical and online settings, and proceed to numerical experiments on simulations to validate the estimation method. We provide elements about the numerical convergence of the model to its continuous version.
    \item Finally, we characterize two types of suspicious trading behaviors and apply the model to cryptocurrency trade data. We show that it is able to identify extreme bursts of events that do not generate price changes. We delve deeper into the analysis of the characterized states and demonstrate the robustness of our model as it is able to detect suspicious trading activity even weeks after the in-sample data.
\end{itemize}

The paper is organised as follows. Section \ref{section:mmhp_framework} introduces the Markov-modulated Hawkes process (MMHP) with $\delta$-piecewise constant intensity and its continuous version. We discuss how the piecewise constant intensity enables us to solve the ordinary differential equations (ODEs) that are required for the computation of the likelihood. An EM algorithm is developed in Section \ref{section:estimation_procedure}, and a careful validation experiment is conducted on simulations. In Section \ref{section:application} we focus on the estimation of the model on cryptocurrency market data, and we provide an extensive analysis of the states that are characterized.

\section{The Markov-Modulated Hawkes Process}
\label{section:mmhp_framework}

\subsection{Mathematical framework}

Let $(\Omega, \mathscr{F}, \mathbb{P}, (\mathcal{F}_t)_{t\geq0})$ be a filtered probability space. Consider a point process $(N_t)_{t\geq0}$ with occurrence times $(t_n)_{n\in\{1,\dots,K\}}$, with convention $t_0=0$ and parameters varying according to an unobservable M-state irreducible and aperiodic continuous-time Markov chain $(S_t)_{t\geq0}$. Let $\mathcal{F}_t:=\sigma(N_s, 0\leq s\leq t)$, for $t\geq0$.

We denote the infinitesimal generator of the underlying Markov process by $Q:=(q_{ij})_{1\leq i,j \leq M}$ such that $q_i:=-q_{ii}=\sum_{j\neq i}q_{ij},\hspace{0.1cm}1\leq i \leq M$. Let $\xi_0$ be the initial distribution vector, with elements
\begin{equation}
    \xi_0^i=\mathbb{P}(S_0=i),\hspace{0.3cm}1\leq i\leq M,
\end{equation}
and we denote by $\pi$ the stationary distribution of $S$.

The observed Hawkes process is assumed to be characterized by the following conditional intensity
\begin{align}
    \lambda_t=\lambda(t|\mathcal{F}_t)&:=\sum_{i=1}^M\mathds{1}_{\{S_t=i\}}\lambda_t^i\\
    &:=\sum_{i=1}^M\mathds{1}_{\{S_t=i\}}\left(\mu^i\,+\,\int_{-\infty}^t\phi^i(t-s)\,\mathrm{d}N_s\right)\\
    &=\sum_{i=1}^M\mathds{1}_{\{S_t=i\}}\left(\mu^i\,+\,\sum_{t_k<t}\phi^i(t-t_k)\right),
\end{align}
where $\mu^i$ and $\phi^i$ are respectively the baseline and the kernel associated to state $i$.

Let $\Lambda_t$ be the diagonal matrix with coefficients $(\lambda_t^i)_{1\leq i\leq M}$ and denote by $\mathbb{I}_M$ the identity matrix of size $M$.
Let $x_n:=t_n-t_{n-1}$ be the duration between the $(n-1)$-th and $n$-th event of the process $N$.

\begin{definition}[Forward transition matrix]
    \label{def:forward_transition_matrix} Let $t_{n-1}$ and $t_n$ be two successive jump times of $N$, and let
    \begin{equation}
        H_{ij}^{(n)}(u):=\mathbb{P}\left(S_{t_{n-1}+u}=j,\, N_{t_{n-1}+u}-N_{t_{n-1}}=0\,|\,S_{t_{n-1}}=i, \mathcal{F}_{t_{n-1}}\right),\hspace{0.3cm} 1\leq i,j\leq M,\hspace{0.1cm} 0\leq u\leq x_n.
    \end{equation}
    The forward transition matrix from $t_{n-1}$ to $t_{n-1}+u$ is denoted by $H^{(n)}(u)::=\left(H_{ij}^{(n)}(u)\right)_{1\leq i, j\leq M}$.
\end{definition}

\begin{prop}[Forward ODEs]
    \label{prop:forward_ode}
    The forward transition matrix is the solution of the system of ordinary differential equations
    \begin{equation}
        \frac{\mathrm{d}H^{(n)}}{\mathrm{d}u}(u)=H^{(n)}(u)\left(Q-\Lambda_{t_{n-1}+u}\right),\hspace{0.3cm} u\in]0, x_n],\label{eq:ode_forward}
    \end{equation}
    with initial condition $H^{(n)}(0)=\mathbb{I}_M$.
\end{prop}

For the sake of readability, the proof is shown in appendix, in Section \ref{section:proofs}.

Equation \eqref{eq:ode_forward} is a system of ordinary differential equations that does not have a trivial analytic solution. It would be the case if the coefficient matrix $A_u:=Q-\Lambda_{t_{n-1}+u}$ was not time-dependent (we would get the exponential of a matrix as a solution of the ODE). But here, since the coefficients are time-dependent, the commutativity property does not necessarily hold ($\forall\,(s,t)$, $A_s A_t\neq A_t A_s$).

\begin{definition}[Backward transition matrix]
    \label{def:backward_transition_matrix} Let $t_{n-1}$ and $t_n$ be two successive jump times of $N$ and let
    \begin{equation}
        G_{ij}^{(n)}(u):=\mathbb{P}\left(S_{t_n}=j,\, N_{t_n}-N_{t_n-u}=0\,|\,S_{t_n-u}=i, \mathcal{F}_{t_n-u}\right),\hspace{0.3cm} 1\leq i,j\leq M,\hspace{0.1cm} 0\leq u\leq x_n.
    \end{equation}
    The backward transition matrix from $t_n-u$ to $t_n$ is denoted by $G^{(n)}(u):=\left(G_{ij}^{(n)}(u)\right)_{1\leq i, j\leq M}$.
\end{definition}

\begin{lemma}\label{lemma:forward_backward_equation}
    The forward and backward transition matrices follow the equality
    \begin{equation}
        H^{(n)}(t_n-t_{n-1})=H^{(n)}(t-t_{n-1})G^{(n)}(t_n-t), \hspace{0.3cm}t_{n-1}\leq t\leq t_n.
    \end{equation}
\end{lemma}

\begin{prop}[Backward ODEs]
    \label{prop:backward_ode}
    The backward transition matrix is the solution of the system of ordinary differential equations
    \begin{equation}
        \frac{\mathrm{d}G^{(n)}}{\mathrm{d}u}(u)=\left(Q-\Lambda_{t_n-u}\right)G^{(n)}(u),\hspace{0.3cm}u\in]0, x_n],\label{eq:ode_backward}
    \end{equation}
    with initial condition $G^{(n)}(0)=\mathbb{I}_M$.
\end{prop}

Finally, let
\begin{equation}
    F_{ij}^{(n)}(x):=\mathbb{P}\left(S_{t_{n-1}+X_n}=j,\, X_n\leq x\,|\,S_{t_{n-1}}=i, \mathcal{F}_{t_{n-1}}\right),\hspace{0.3cm} 1\leq i,j\leq M,\hspace{0.1cm} x\geq 0,
\end{equation}
where $X_n$ denotes the random duration after $t_{n-1}$. Let$f_{ij}^{(n)}(x)$, $x\geq 0$, be the associated density and $f^{(n)}(x):=\left(f_{ij}^{(n)}(x)\right)_{1\leq i, j\leq M}$ its matrix form.
Following \citet{meier1987fitting} and \citet{ryden1996algorithm}, we have

\begin{equation}\label{eq:transition_density_matrix}
    f^{(n)}(x)=H^{(n)}(x)\Lambda_{t_{n-1}+x},\quad x\geq 0.
\end{equation}

\subsection{A piecewise constant intensity function}

The ODEs of Proposition \ref{prop:forward_ode} and Proposition \ref{prop:backward_ode} do not have a simple analytic solution. Thus, we propose to use a piecewise constant kernel function in the spirit of the Markov-modulated Hawkes process with stepwise decay (MMHPSD) of \citet{wang2010statistical}. Our approach naturally extends the MMHPSD by allowing the intensity to decrease between event times. 

Denote by $\delta>0$ the time step that characterizes the frequency at which the intensity varies between events, and define
\begin{equation}
    \omega(t):=\sup\{u\leq t,\, N_u-N_{u-}=1\},
\end{equation}

the last arrival time of an event before $t$ and

\begin{equation}
    \ell(t):=\max\{l\in\mathbb{N},\, \omega(t)+l\,\delta\leq t\},
\end{equation}

the number of disjoint sub-intervals of length $\delta$ that fit into the interval $[\omega(t),t]$.

For convenience, we will use the notation $l_n:=\ell(t_n^-)$, and we define $\Delta_n$ such that: $t_n-t_{n-1}=l_n\delta+\Delta_n$.

\begin{definition}
    Let $N$ be a Markov-modulated Hawkes process as defined in Section \ref{section:mmhp_framework}. The conditional intensity with $\delta$-piecewise constant kernel is defined as
    \begin{align}\label{eq:piecewise_constant_intensity}
        \lambda_t&=\sum_{i=1}^M\mathds{1}_{\{S_t=i\}}\lambda_t^i\\
        &=\sum_{i=1}^M\mathds{1}_{\{S_t=i\}}\left\{\mu^i+\int_{-\infty}^t\phi^i\left(\omega(t)+\ell(t)\delta-s\right)\mathrm{d}N_s\right\}\\
        &=\sum_{i=1}^M\mathds{1}_{\{S_t=i\}}\left\{\mu^i+\sum_{t_k<t}\phi^i\left(\omega(t)+\ell(t)\delta-t_k\right)\right\}.
    \end{align}
    Recall that $\Lambda_t$ is the diagonal matrix with coefficients $(\lambda_t^i)_{1\leq i\leq M}$. The MMHP model with $\delta$-piecewise constant kernel is hereafter called the MMHP-$\delta$ model.
\end{definition}

A sample trajectory of the MMHP-$\delta$ intensity is provided in Figure \ref{fig:example_mmhp_intensity_and_kernel}. We show an example of how the intensity evolves with respect to $\delta$ and how it compares to its continuous counterpart. The behavior of the $\delta$-stepwise decay of an exponential kernel is displayed, showing how the decay adapts to a new event time in order to preserve the piecewise constant property with step $\delta$ of the overall intensity: the remaining time to next decay $\tau\leq \delta$ is reinitialized with value $\delta$ at each event time. 

\begin{figure}
    \centering
    \subfloat[A trajectory sample of MMHPs intensities with varying $\delta$ together with their continuous counterpart.]{%
        \includegraphics[width=0.5\linewidth]{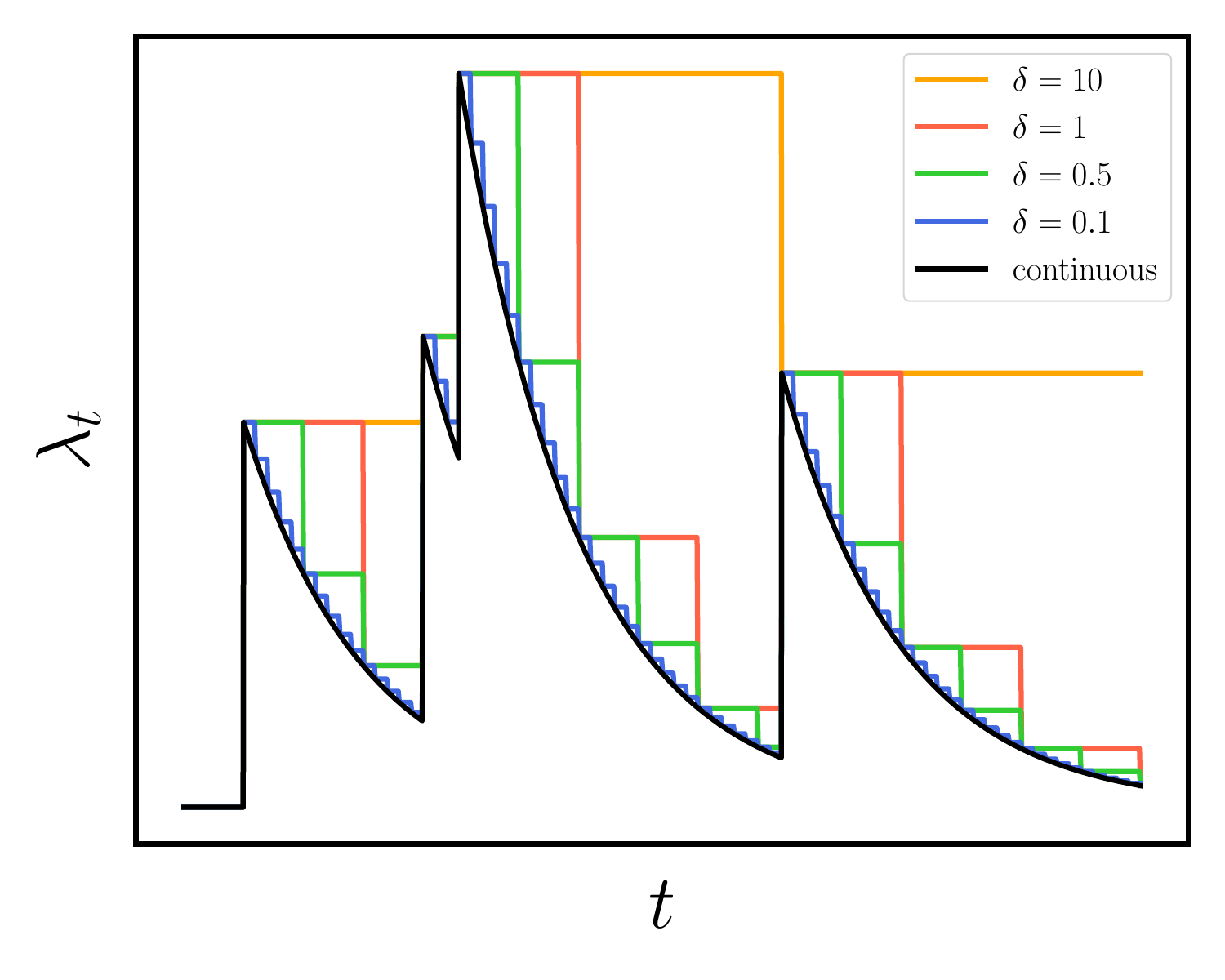}%
    }
    \subfloat[The evolution of the MMHP-$\delta$ kernel when a new event arrives.]{%
        \includegraphics[width=0.5\linewidth]{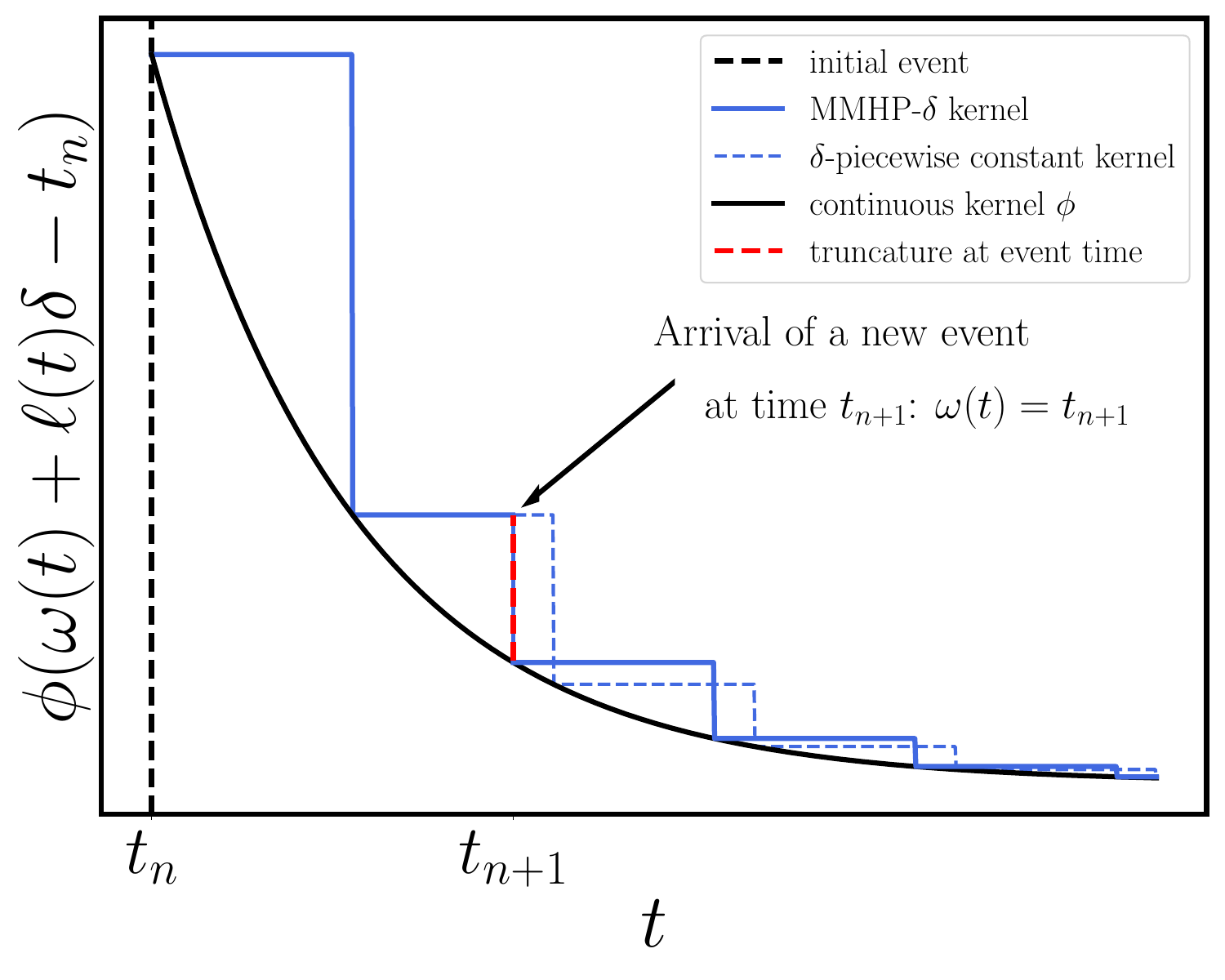}%
    }
    \caption{An illustration of the dynamics of MMHP-$\delta$ intensity.}
    \label{fig:example_mmhp_intensity_and_kernel}
\end{figure}

Under the $\delta$-piecewise constant kernel framework, the forward and backward ODEs of Equations \eqref{eq:ode_forward} and \eqref{eq:ode_backward} can be explicitly solved.

\begin{prop}\label{prop:forward_backward_formulas}
    If the intensity is given by Equation \eqref{eq:piecewise_constant_intensity}, the forward transition matrix is
    \begin{equation}\label{eq:formula_H}
        H^{(n)}(u)=\Xi^{(n)}_{\ell(t_{n-1}+u)}e^{(Q-\Lambda_{t_{n-1}+\ell(t_{n-1}+u)\delta})(u-\ell(t_{n-1}+u)\delta)},\hspace{0.3cm} 0\leq u\leq x_n,
    \end{equation}
    for which we define the following quantity
    \begin{equation}
        \Xi^{(n)}_k:=\prod_{r=1}^{k}e^{((Q-\Lambda_{t_{n-1}+(r-1)\delta})\delta)}, \hspace{0.3cm} 1\leq n\leq K,\hspace{0.1cm} 1\leq k \leq \ell_n.
    \end{equation}
    
    Similarly, the backward transition matrix is written as, for $u\in[0, x_n]$
    \begin{equation}\label{eq:formula_G}
        \begin{cases} G^{(n)}(u)=e^{(Q-\Lambda_{t_{n-1}+\ell_n\delta})u}\hspace{0.2cm}\text{if}\;u\leq\Delta_n,\\\\  G^{(n)}(u)=e^{(Q-\Lambda_{t_{n-1}+\ell(t_n-u)\delta})(u-((\ell_n-\ell(t_n-u)-1)\delta+\Delta_n))}\Psi^{(n)}_{\ell_n-\ell(t_n-u)-1}e^{(Q-\Lambda_{t_{n-1}+\ell_n\delta})\Delta_n}\hspace{0.2cm}\text{else},
    \end{cases}
    \end{equation}
    for which we define the following quantity
    \begin{equation}
        \Psi^{(n)}_k:=\prod_{r=1}^{k}e^{(Q-\Lambda_{t_{n-1}+(\ell_n-(k-r)-1)\delta})\delta}, \hspace{0.3cm} 1\leq n\leq K,\hspace{0.1cm} 1\leq k \leq \ell_n.
    \end{equation}
    Here, the notation $e^A$ stands for the matrix exponential of $A$ and we use the convention $\prod_{r=1}^{0}A_r=\mathbb{I}_M$.
\end{prop}

\section{Estimation procedure}
\label{section:estimation_procedure}

\subsection{The likelihood function}

Let $K$ be the number of observations, and let $\mathcal{T}:=\{t_n\}_{1\leq n\leq K}$ be the ordered set of observed event times and $\mathcal{X}:=\{x_n=t_n-t_{n-1}\}_{1\leq n\leq K}$ be the set of observed durations, with the convention $t_0=0$. Note that we do not consider that the process $N$ jumps at time $0$. For notational simplicity, we assume that the observation time frame $[0,T]$ is such that $t_K=T$. Denote by $\mathcal{L}(\Theta,\mathcal{T})$ the likelihood of the model with the set of parameters $\Theta:=\Theta^M\cup\Theta^H$, where $\Theta^M:=\{(\xi_0^i)_{1\leq i\leq M},(q_{ij})_{1\leq i,j\leq M}\}$ is the set of Markov chain parameters and $\Theta^H$ is the set of point process parameters.

In \cite{wang2010statistical}, the authors give a simple form of the complete likelihood $\mathcal{L}$. After simplification, the resulting log-likelihood is separated into two parts: the Markov chain part which characterizes past state transitions and the Hawkes part which is expressed in terms of past event occurrences characterized by the underlying Hawkes process. We use the same idea and write the complete likelihood for the MMHP in the proposition below.

\begin{prop}\label{prop:complete_likelihood}
    The log-likelihood of an MMHP with parameters $\Theta:=\Theta^M\cup\Theta^H$ over the observations $\mathcal{T}$ can be written as
    \begin{equation}
        \mathcal{L}(\Theta,\mathcal{T})=\log \mathcal{L}^M(\Theta^M,\mathcal{T})+\log \mathcal{L}^H(\Theta^H,\mathcal{T}),
    \end{equation}
    where
    \begin{equation}
        \log \mathcal{L}^M(\Theta^M,\mathcal{T})=\sum_{i=1}^M\left\{\mathds{1}_{\{S_0=i\}}\log \xi_0^i-D_iq_i+\sum_{j=1,j\neq i}^Mw_{ij}\log q_{ij}\right\},
    \end{equation}
    and
    \begin{equation}
        \log \mathcal{L}^H(\Theta^H,\mathcal{T})=\sum_{i=1}^M\left\{\sum_{n=1}^K\mathds{1}_{\{S_{t_n}=i\}}\log \lambda_{t_n}^i(\Theta^H)\,-\,\int_0^T\mathds{1}_{\{S_t=i\}}\lambda_t^i(\Theta^H)\,\mathrm{d}t\right\},
    \end{equation}
    for which we define the time $S$ spends in state $i$ over $[0,T]$
    \begin{equation}\label{eq:D_i}
        D_i:=\int_0^T\mathds{1}_{\{S_t=i\}}\mathrm{d}t=\sum_{p=1}^{m+1} \mathds{1}_{\{s_p=i\}}(u_p-u_{p-1}),
    \end{equation}
    and the number of times $S$ jumps from state $i$ to state $j$ over $[0,T]$
    \begin{equation}\label{eq:w_ij}
        w_{ij}:=\int_0^T\underset{\Delta\to0^+}{\lim}\frac{1}{\Delta}\mathds{1}_{\{S_{t-\Delta}=i,\, S_t=j\}}\mathrm{d}t=\sum_{p=1}^m \mathds{1}_{\{s_p=i, s_{p+1}=j\}}.
    \end{equation}
\end{prop}

The proof is based on \cite{wang2010statistical} and is given in Appendix, in Section \ref{section:proofs}.

In the rest of the paper, we sometimes write $\lambda_t(\Theta)$ instead of $\lambda_t(\Theta^H)$ or simply $\lambda_t$ for the sake of readability when the context allows us to do so.

\subsection{An Expectation-Maximization algorithm}

In this section, we design an Expectation-Maximization (EM) algorithm that can be used to estimate the parameters. To this extent, we first develop the computations that are needed for the E-step. Finally, we describe the M-step and wrap up the estimation procedure by providing a detailed algorithm.

\subsubsection{Expectation step}

\begin{definition}
    Let $t\geq0$, and $n$ such that $N_t=n$. The forward likelihood of being in state $i$ at time $t$ is
    \begin{equation}
        \alpha_t^i:=p\left(T_1=t_1,\,\dots,\,T_n=t_n,\, t_n\leq t < t_{n+1},\, S_t=i\right),\hspace{0.3cm}1\leq i\leq M,
    \end{equation}
    where $p$ is a convenient generic notation for the density.
    
    The forward likelihood vector at time $t$ is $\boldsymbol{\alpha}_t:=(\alpha_t^i)_{1\leq i\leq M}$.

    The backward likelihood conditionally on being in state $i$ at time $t$ is
    \begin{equation}
        \beta_t^i:=p\left(T_{n+1}=t_{n+1},\,\dots,\,T_K=t_K,\, t_n< t\leq t_{n+1}\,\big|\, S_t=i\right),\hspace{0.3cm}1\leq i\leq M.
    \end{equation}
    The backward likelihood vector at time $t$ is $\boldsymbol{\beta}_t:=(\beta^i)_{1\leq i\leq M}$.
\end{definition}

\begin{prop}
    Let $t\geq0$, and $n$ such that $N_t=n$. The forward likelihood of being in state $i$ at time $t$ and the backward likelihood conditionally on being in state $i$ at time $t$ satisfy
    \begin{equation}
        \alpha_t^i=\xi_0\,\prod_{k=1}^nf^{(k)}(x_k)\,H^{(n+1)}\left(t-t_n\right)\,\mathds{1}_i, \hspace{0.3cm} 1\leq i\leq M,
    \end{equation}
    \begin{equation}
        \beta_t^i=\mathds{1}_i'\,G^{(n+1)}\left(t_{n+1}-t\right)\,\Lambda_{t_{n+1}^-}\,\prod_{k=n+2}^Kf^{(k)}(x_k)\,\mathds{1}, \hspace{0.3cm} 1\leq i\leq M,
    \end{equation}
    where $\mathds{1}_i$ is the $\mathbb{R}^M$ vector the elements of which are zeros except the $i$th entry which is 1, and $\mathds{1}$ is the $\mathbb{R}^M$ vector of ones.
\end{prop}

The first step of the EM algorithm consists in taking the expectation of the log-likelihood conditionally on the entire history $\mathcal{F}_T$, over the state variable $S$ with the current parameter estimate $\widehat{\Theta}$. By doing so, we can infer the expected path of the unobserved process $S$ using the information we have about the observed process $N$.

In the rest of the paper, we will use the following definitions of filtered and smoothed probabilities.
\begin{definition}
    The filtered probability of being in state $i$ at time $t$ is
    \begin{equation}
        \xi_{t|t}^i:=\mathbb{P}(S_t=i|\mathcal{F}_t),\hspace{0.3cm}1\leq i\leq M.
    \end{equation}
    The smoothed probability of being in state $i$ at time $t$ is
    \begin{equation}
        \xi_{t|T}^i:=\mathbb{P}(S_t=i|\mathcal{F}_T),\hspace{0.3cm}1\leq i\leq M.
    \end{equation}
\end{definition}

We will also use the same notations as in \cite{roberts2006ryde},
\begin{align}
    c_n&:=L(n-1)f^{(n)}(x_n)\mathds{1},\\
    L(n)&:=L(0)\prod_{k=1}^n\frac{f^{(k)}(x_k)}{c_k},\\
    R(n)&:=\prod_{k=n}^K\frac{f^{(k)}(x_k)}{c_k}R(K+1),\hspace{0.3cm}1\leq n\leq K,
\end{align}
with conditions $L(0)=\xi_0$, $R(K+1)=\mathds{1}$.

Note that the likelihood can be expressed as
\begin{equation}
    \mathcal{L}(\Theta,\mathcal{T})=\prod_{n=1}^Kc_n=\boldsymbol{\alpha}_t'\boldsymbol{\beta}_t,\hspace{0.3cm}0\leq t\leq T,
\end{equation}
and we have the following expressions of the filtered and smoothed probabilities
\begin{align}
    \label{eq:filtered_prob_vs_forward_backward}
    \xi_{t|t}^i&=\frac{\alpha_t^i}{\boldsymbol{\alpha}_t'\mathds{1}},\\
    \label{eq:smoothed_prob_vs_forward_backward}
    \xi_{t|T}^i&=\frac{\alpha_t^i\beta_t^i}{\mathcal{L}(\Theta,\mathcal{T})},\hspace{0.3cm}1\leq i\leq M,\hspace{0.1cm}0\leq t\leq T.
\end{align}

We first compute the expected log-likelihood conditionally on the history of events and an estimation $\widehat{\Theta}$ of the parameters. We will show that it is natural to define the following quantities

\begin{equation}
    P^{(n)}:=\Lambda_{t_n^-}R(n+1)L(n-1),
\end{equation}

\begin{equation}
    \Omega^{(n)}_k:=\Psi^{(n)}_{\ell_n-k-1}e^{(Q-\Lambda_{t_{n-1}+\ell_n\delta})\Delta_n}P^{(n)}\Xi^{(n)}_k,
\end{equation}

\begin{equation}
    C^{(n)}_k\,:=\,\begin{bmatrix}
        Q-\Lambda_{t_{n-1}+k\delta} & \Omega^{(n)}_k \\\\
        \mathbb{O}_M & Q-\Lambda_{t_{n-1}+k\delta}
    \end{bmatrix}_{2M\times 2M},
\end{equation}

\begin{equation}
    D^{(n)}\,:=\,\begin{bmatrix}
        Q-\Lambda_{t_{n-1}+\ell_n\delta} & P^{(n)}\,\Xi^{(n)}_{\ell_n} \\\\
        \mathbb{O}_M & Q-\Lambda_{t_{n-1}+\ell_n\delta}
    \end{bmatrix}_{2M\times 2M},
\end{equation}

where we denoted by $\mathbb{O}_M$ the $M\times M$ matrix filled with zeros. Moreover, we will denote by $[A]^{{M\boxtimes M}}$ the upper-right block matrix of size $M$ of $A$, $A_{ij}$ the element of A at row $i$ and column $j$, and $\odot$ the Hadamard product, \textit{i.e.} the element-wise multiplication.

An analytic expression of the expected log-likelihood needed during the E-step is provided by the following proposition.

\begin{prop}\label{prop:e_step_likelihood}
    Given an estimated set of parameters $\widehat{\Theta}$, the expected log-likelihood of the E-step for the MMHP can be written as
    \begin{equation}
        \mathbb{E}\left(\log \mathcal{L}(\Theta,\mathcal{T})|\mathcal{F}_T,\widehat{\Theta}\right)=\mathbb{E}\left(\log \mathcal{L}^M(\Theta,\mathcal{T})\Big|\mathcal{F}_T,\widehat{\Theta}\right)+\mathbb{E}\left(\log \mathcal{L}^H(\Theta,\mathcal{T})\Big|\mathcal{F}_T,\widehat{\Theta}\right),
    \end{equation}
    with
    \begin{equation}\label{eq:e_step_likelihood_markov}
        \mathbb{E}\left(\log \mathcal{L}^M(\Theta,\mathcal{T})\Big|\mathcal{F}_T,\widehat{\Theta}\right)=\sum_{i=1}^M\left\{\xi_{0|T}^i(\widehat{\Theta})\log \xi_0^i-\mathbb{E}(D_i|\mathcal{F}_T,\widehat{\Theta})q_i+\sum_{j=1,j\neq i}^M\mathbb{E}(w_{ij}|\mathcal{F}_T,\widehat{\Theta})\log q_{ij}\right\},
    \end{equation}
    and
    \begin{equation}\label{eq:e_step_likelihood_hawkes}
        \mathbb{E}\left(\log \mathcal{L}^H(\Theta,\mathcal{T})\Big|\mathcal{F}_T,\widehat{\Theta}\right)=\sum_{i=1}^M\left\{\sum_{n=1}^K\xi_{t_n|T}^i(\widehat{\Theta})\log \lambda_{t_n}^i(\Theta)-\int_0^T\xi_{t|T}^i(\widehat{\Theta})\lambda_t^i(\Theta)\mathrm{d}t\right\}.
    \end{equation}

    Furthermore, in the case of the MMHP-$\delta$, we have the analytic expressions

    \begin{equation}
        \mathbb{E}\left(D_i|\mathcal{F}_T,\widehat{\Theta}\right)=\sum_{n=1}^K\frac{1}{c_n}\left[\sum_{k=0}^{\ell_n-1}\left[e^{C^{(n)}_k\delta}\right]^{M\boxtimes M}+\left[e^{D^{(n)}\Delta_n}\right]^{M\boxtimes M}\right]_{ii},
    \end{equation}

    \begin{equation}
        \mathbb{E}\left(w_{ij}|\mathcal{F}_T,\widehat{\Theta}\right)=\widehat{q}_{ij}\sum_{n=1}^K\frac{1}{c_n}\left[\sum_{k=0}^{\ell_n-1}\left[e^{C^{(n)}_k\delta}\right]^{M\boxtimes M}+\left[e^{D^{(n)}\Delta_n}\right]^{M\boxtimes M}\right]_{ji},
    \end{equation}

    \begin{equation}\label{eq:e_step_integral_smoothed_prob_intensity}
        \int_0^T\xi_{t|T}^i(\widehat{\Theta})\lambda_t^i\mathrm{d}t=\sum_{n=1}^K\frac{1}{c_n}\left[\sum_{k=0}^{\ell_n-1}\Lambda_{t_{n-1}+k\delta}\,\odot\,\left[e^{C^{(n)}_k\delta}\right]^{M\boxtimes M}+\Lambda_{t_{n-1}+\ell_n\delta}\,\odot\,\left[e^{D^{(n)}\Delta_n}\right]^{M\boxtimes M}\right]_{ii},
    \end{equation}

    and the smoothed probabilities $\xi_{t_n|T}^i$ are given by Equation \eqref{eq:smoothed_prob_vs_forward_backward}.
\end{prop}

The proof is given in Appendix, in Section \ref{section:proofs}.

\subsubsection{Maximization step}

The second step of the EM algorithm consists in maximizing the expected likelihood obtained at the E-step. Thus, our goal is to solve the following optimization problem

\begin{equation}\label{eq:optimization_argmax_m_step}
    \Theta^*=\underset{\Theta}{\argmax}\hspace{0.1cm}\mathbb{E}\left(\log \mathcal{L}(\Theta,\mathcal{T})|\mathcal{F}_T,\widehat{\Theta}\right).
\end{equation}

We can maximize the two likelihood terms of Proposition \ref{prop:e_step_likelihood} separately to get the new optimal set of Markov chain parameters $\Theta^{M*}$ and the new optimal set of Hawkes parameters $\Theta^{H*}$. Thus the above problem may be divided into two independent optimization problems

\begin{equation}
    \Theta^{M*}=\underset{\Theta^M}{\argmax}\hspace{0.1cm}\mathbb{E}\left(\log \mathcal{L}^M(\Theta^M,\mathcal{T})|\mathcal{F}_T,\widehat{\Theta}\right),
\end{equation}

\begin{equation}
    \Theta^{H*}=\underset{\Theta^H}{\argmax}\hspace{0.1cm}\mathbb{E}\left(\log \mathcal{L}^H(\Theta^H,\mathcal{T})|\mathcal{F}_T,\widehat{\Theta}\right).
\end{equation}

The proposition below provides the parameter update equations for the M-step.

\begin{prop}\label{prop:m_step}
    The update equations of the Markov chain parameters $(Q, \xi_0)$ are
    \begin{equation}
        q_{ij}^*=\frac{\mathbb{E}(w_{ij}|\mathcal{F}_T,\widehat{\Theta})}{\mathbb{E}(D_i|\mathcal{F}_T,\widehat{\Theta})},\hspace{0.3cm}1\leq i,j\leq M,\hspace{0.1cm} i\neq j,
    \end{equation}
    \begin{equation}
        q_{ii}^*=-\sum_{j=1,j\neq i}^Mq_{ij}^*,\hspace{0.3cm}1\leq i\leq M.
    \end{equation}
    \begin{equation}
        {\xi_0^i}^*=\widehat{\xi_0^i}\,\mathds{1}_i\,R(1),\hspace{0.3cm}1\leq i\leq M.
    \end{equation}
    In the MMHP-$\delta$ model, the updates of the Hawkes parameters $(\mu, \alpha, \beta)$ are obtained by solving the following optimization problem
    \begin{align}
        \Theta^{H*}&=\underset{\Theta^H}{\argmax}\hspace{0.1cm}\sum_{i=1}^M\sum_{n=1}^KL(n)\mathds{1}_i\mathds{1}_i^\intercal R(n+1)\log \lambda_{t_n}^i\left(\Theta^H\right)\notag\\
        &\hspace{0.5cm}-\sum_{n=1}^K\frac{1}{c_n}\Tr\left(\sum_{k=0}^{\ell_n-1}\Lambda_{t_{n-1}+k\delta}\left(\Theta^H\right)\odot\left[e^{C^{(n)}_k\delta}\right]^{M\boxtimes M}+\Lambda_{t_{n-1}+\ell_n\delta}\left(\Theta^H\right)\odot \left[e^{D^{(n)}\Delta_n}\right]^{M\boxtimes M}\right).
    \end{align}
\end{prop}

\begin{proof}
    Use Equation \eqref{eq:e_step_likelihood_markov} and write the first order condition of the optimization problem. After rearranging terms, the update equations of the Markov chain parameters are obtained. Concerning the update of the Hawkes parameters, use Equations \eqref{eq:e_step_likelihood_hawkes} and \eqref{eq:e_step_integral_smoothed_prob_intensity} to get the objective function.
\end{proof}

The complete EM estimation procedure is described in Algorithm \ref{algo:em_estimation}.

\begin{algorithm}[htb]
    \caption{EM algorithm for the MMHP-$\delta$}
    \label{algo:em_estimation}
    \begin{enumerate}
        \item Set a maximum number of steps $E$ to take and a log-likelihood improvement threshold $\varepsilon$ as a stopping criterion.
        \item Initialize the estimated parameters $\widehat{\Theta}$ and the current log-likelihood as $\log\mathcal{L}=-\infty$. Compute the sequence $(\ell_n)_{1\leq n\leq K}$.
        \item For $r = 1$ to $E$:
        \begin{enumerate}
            \item \textbf{E-step}:
            \begin{enumerate}
                \item Compute the diagonal intensity matrices $\Lambda_{t_{n-1}+k\delta}$, $1\leq n \leq K, 0\leq k \leq \ell_n$, the forward and backward transition matrices $H^{(n)}(x_n)$ and $G^{(n)}(x_n)$, $1\leq n \leq K$, and deduce recursively $L(n)$, $c_n$, $R(n)$, for $1\leq n \leq K$.
                \item Compute the matrix exponentials $e^{C^{(n)}_k\delta}$, $1\leq n \leq K, 0\leq k \leq \ell_n - 1$, and $e^{D^{(n)}\Delta_n}$,  $1\leq n \leq K$. Compute their upper-right block of size $M\times M$.
                \item Compute $\mathbb{E}(w_{ij}|\mathcal{F}_T,\widehat{\Theta})$, $1\leq i,j\leq M$, and $\mathbb{E}(D_i|\mathcal{F}_T,\widehat{\Theta})$, $1\leq i \leq M$.
            \end{enumerate}
            \item Check the log-likelihood stopping criterion: if $\sum_{n=1}^K\log c_n - \log\mathcal{L} < \varepsilon$, leave the loop. Set $\log\mathcal{L} = \sum_{n=1}^K\log c_n$.
            \item \textbf{M-step}:
            \begin{enumerate}
                \item Use the update equations of Proposition \ref{prop:m_step} for the Markov chain parameters and use a non-linear optimization algorithm to maximize the expected log-likelihood with respect to the Hawkes parameters.
                \item Update the estimated parameters $\widehat{\Theta}$ with the new estimates.
            \end{enumerate}
        \end{enumerate}
        \item \textbf{Output:} An estimation $\widehat{\Theta}$ of the optimal parameters.
    \end{enumerate}
\end{algorithm}

\subsection{Estimation of past states}

Suppose that we estimated the model parameters $\Theta^*$. One important feature of hidden Markov chains is the estimation of the sequence of past states, knowing the history of events and the estimated parameters. We propose here a Viterbi algorithm to handle this task.

\subsubsection{At event times}

Denote by $S^*:=(s_n^*)_{1\leq n \leq K}$ the optimal sequence of states at event times such that:
\begin{equation}
    S^*=\underset{S=(s_n)_{1\leq n \leq K}}{\argmax}\hspace{0.1cm}p\left(S,\mathcal{F}_T|\Theta^*\right).
\end{equation}

We get the following optimization problem
\begin{equation}
    S^*=\underset{S=(s_n)_{1\leq n \leq K}}{\argmax}\hspace{0.1cm}\prod_{n=1}^Kp\left(S_{t_n}=s_n,N_{t_n^-}-N_{t_{n-1}}=0,N_{t_n}-N_{t_n^-}=1|S_{t_{n-1}}=s_{n-1},\mathcal{F}_{t_{n-1}}\right).
\end{equation}
Using previous notations and Equation \eqref{eq:transition_density_matrix},
\begin{equation}\label{eq:optimal_viterbi_states}
    S^*=\underset{ S:=(s_n)_{n\in\{0,\dots,K\}}}{\argmax}\hspace{0.1cm}\prod_{n=1}^KH_{s_{n-1}s_n}(x_n)\lambda_{s_n}(t_n).
\end{equation}
Solving this problem with brute force is not computationally feasible for a history of thousands of events. The Viterbi algorithm enables us to handle this issue. To this extent, we define
\begin{equation}
    \eta_n^i:=\underset{(s_1,s_2,\dots,s_{n-1})}{\max}\hspace{0.1cm} p\left(S_{t_1}=s_1,\dots,S_{t_{n-1}}=s_{n-1},S_{t_n}=i,\mathcal{F}_{t_n}|\Theta^*\right),\hspace{0.3cm}1\leq i \leq M,\hspace{0.1cm} 2\leq n\leq K,
\end{equation}
with $\eta_1^i=p(S_0=i,S_{t_1}=i,N_{t_1^-}=0,N_{t_1}=1)$.

The sequence $(\eta_n)_n$ is computed using the dynamic programming equation
\begin{equation}
    \eta_n^j=\underset{1\leq i \leq M}{\max}\hspace{0.1cm}\eta_{n-1}^iH_{ij}^{(n)}(x_n)\lambda^j(t_n),\hspace{0.3cm}1\leq j \leq M,\hspace{0.1cm} 2\leq n\leq K.
\end{equation}
We introduce the notation
\begin{equation}
    \psi_n^j:=\underset{1\leq i \leq M}{\argmax}\hspace{0.1cm}\eta_{n-1}^iH_{ij}^{(n)}(x_n),\hspace{0.3cm}1\leq j \leq M,\hspace{0.1cm} 2\leq n\leq K,
\end{equation}
with $\psi_1^j=0$. Thus,

\begin{equation}
    \eta_n^j=\eta_{n-1}^{\psi_n^j}H_{\psi_n^j j}^{(n)}(x_n)\lambda^j(t_n),\hspace{0.3cm}1\leq j \leq M,\hspace{0.1cm} 2\leq n\leq K.
\end{equation}

The optimal sequence of states is obtained by operating the backward recursion
\begin{equation}
    s_n^*=\psi_{n+1}^{s_{n+1}^*},\hspace{0.3cm}1\leq n\leq K-1,
\end{equation}
from the terminal condition $s_K^*=\underset{1\leq i \leq M}{\argmax}\hspace{0.1cm}\eta_{K}^i$.

It is noteworthy that in practice, float computation errors can arise for very small values of transition matrices, leading to $\eta_n^j=0$, for $n\geq n_0$, for a particular $n_0$, and $1\leq j\leq M$. To avoid these numerical instabilities, we recommend working with the logarithm of densities. The algorithm is equivalent to the computation of the recursion $\eta_n^j=\underset{1\leq i \leq M}{\max}\hspace{0.1cm}\log{\eta_{n-1}^iH_{ij}^{(n)}(x_n)}+\log{\lambda^j(t_n)}$ for all $(j, n)$, $1\leq j \leq M,\hspace{0.1cm} 2\leq n\leq K$.

We describe the estimation procedure in Algorithm \ref{algo:viterbi}.

\begin{algorithm}[htb]
    \caption{Viterbi algorithm for the MMHP-$\delta$}
    \label{algo:viterbi}
    \begin{enumerate}
        \item Initialize $\eta_1^i=H_{ii}^{(1)}(x_1)\mu^i$, $\psi_1^j=0$.
        \item For $n = 1$ to $K$:
        \begin{itemize}
            \item Compute 
            \begin{equation}
                \eta_n^j=\underset{1\leq i \leq M}{\max}\hspace{0.1cm}\eta_{n-1}^iH_{ij}^{(n)}(x_n)\lambda^j(t_n),\hspace{0.3cm}1\leq j \leq M,
            \end{equation}
            and
            \begin{equation}
                \psi_n^j:=\underset{1\leq i \leq M}{\argmax}\hspace{0.1cm}\eta_{n-1}^iH_{ij}^{(n)}(x_n),\hspace{0.3cm}1\leq j \leq M.
            \end{equation}
        \end{itemize}
        \item Initialize $s_K^*=\underset{1\leq i \leq M}{\argmax}\hspace{0.1cm}\eta_{K}^i$.
        \item For $n = 1$ to $K$:
        \begin{itemize}
            \item Compute 
            \begin{equation}
                s_n^*=\psi_{n+1}^{s_{n+1}^*},\hspace{0.3cm}1\leq n\leq K-1,
            \end{equation}
        \end{itemize}
        \item \textbf{Output:} The optimal sequence of states $S^*$ as defined in Equation \eqref{eq:optimal_viterbi_states}.
    \end{enumerate}
\end{algorithm}

\subsubsection{Online estimation}

In order to compute the sequence of states in a live environment, we need to infer the optimal state at times between events. In this case, the recursion is similar but involves the computation of a specific transition matrix that we denote by $R$. It is defined such that, for $ 1\leq i,j\leq M$,

\begin{equation}
    R_{ij}^{(n)}(u,u'):=\mathbb{P}\left(S_{t_{n-1}+u'}=j,N_{t_{n-1}+u'}-N_{t_{n-1}+u}=0\,|\,S_{t_{n-1}+u}=i, \mathcal{F}_{t_{n-1}+u}\right),\hspace{0.3cm}0\leq u<u'\leq x_n.
\end{equation}

The transition matrix $R$ can be written in terms of the forward transition matrix and its inverse. The formula is given by the following lemma.

\begin{prop}\label{prop:forward_k_equation}
    The transition matrix $R$ is given by
    \begin{equation}
        R^{(n)}(u,u')=\left[H^{(n)}(u)\right]^{-1}H^{(n)}(u'), \hspace{0.3cm}0\leq u\leq u'\leq x_n.
    \end{equation}
\end{prop}

\begin{proof}
    The result comes from Definition \ref{def:forward_transition_matrix}
\end{proof}

The online Viterbi algorithm is described in Algorithm \ref{algo:online_viterbi}. Note that there is no backward recursion as for the historical version since we only need the maximum value of the transition density at the current time.

\begin{algorithm}[htb]
    \caption{Online Viterbi algorithm for the MMHP}
    \label{algo:online_viterbi}
    \begin{enumerate}
        \item Set a history of past observations such that $t_K=T$ is the starting point of the online algorithm, and set $\Delta$ as the discretization time step for the live algorithm. The sequence of states will be continuously updated.
        \item Apply Algorithm \ref{algo:viterbi} to compute $\eta_K^j$, $1\leq j\leq M$.
        \item Set $t=T$.
        \item While True: Denote by $t'$ the current time.
        \begin{enumerate}
            \item If an event has occurred:
            \begin{enumerate}
                \item $K\leftarrow K+1$
                \item $t_K\leftarrow t'$
                \item Compute
                \begin{equation}
                    \eta_K^j=\underset{1\leq i \leq M}{\max}\hspace{0.1cm}\eta_{K-1}^iH_{ij}^{(K)}(t_K-t)\lambda^j(t_K),\hspace{0.3cm}1\leq j \leq M.
                \end{equation}
                \item $t\leftarrow t_K$
            \end{enumerate}
            \item Else:
            \begin{enumerate}
                \item $\eta_K^j\leftarrow \underset{1\leq i \leq M}{\max}\hspace{0.1cm}\eta_{K}^iR_{ij}^{(K+1)}(t,t'),\hspace{0.3cm}1\leq j \leq M.$
                \item $t\leftarrow t'$
            \end{enumerate}
            \item Set $s_t^*=\underset{1\leq i \leq M}{\argmax}\hspace{0.1cm}\eta_{K}^i$.
        \end{enumerate}
        \item \textbf{Output:} A live feed of optimal states $s_t^*$.
    \end{enumerate}
\end{algorithm}

\subsection{Goodness of fit}

We proceed to a residual analysis to assess whether the estimated model correctly approximates the real data generation process. Since the state process $S$ is unobservable, we estimate the intensity $\lambda_t$ with $\widehat{\lambda}_t$ using

\begin{equation}
    \widehat{\lambda}_t:=\mathbb{E}\left(\lambda_t\,\big|\,\mathcal{F}_{T}\right)=\sum_{i=1}^M\xi_{t|T}^i\lambda_t^i.
\end{equation}

Define the residual transformed duration
\begin{equation}
    \tau_n:=\int_{t_{n-1}}^{t_n}\widehat{\lambda}_t\mathrm{d}t,\hspace{0.3cm}1\leq n \leq K.
\end{equation}
Then, a classical result in point process theory \citep{daley2003introduction} is that if the sequence of event times $(t_n)_{1\leq n\leq K}$ was generated by a point process with intensity process $(\widehat{\lambda}_t)_{t\geq0}$, then $\tau_n$ must follow an exponential distribution with unit parameter. Once the sequence of transformed inter-arrival times $(\tau_n)_{1\leq n\leq K}$ is computed, a Q-Q plot can be performed to visually assess the goodness of fit with respect to the exponential distribution. The following proposition provides an analytic formula for the transformed duration of an MMHP-$\delta$.

\begin{prop}
    Under the $\delta$-piecewise constant intensity assumption, we have
    \begin{equation}\label{eq:gof_compensator}
        \tau_n=\frac{1}{c_n}\Tr\left(\sum_{k=0}^{\ell_n-1}\Lambda_{t_{n-1}+k\delta}\odot\left[e^{C^{(n)}_k\delta}\right]^{M\boxtimes M}+\Lambda_{t_{n-1}+\ell_n\delta}\odot \left[e^{D^{(n)}\Delta_n}\right]^{M\boxtimes M}\right),\hspace{0.3cm}1\leq n \leq K.
    \end{equation}
\end{prop}

\begin{proof}
    For all $n$, $1\leq n \leq K$, write
    \begin{equation}
        \tau_n=\int_{t_{n-1}}^{t_n}\widehat{\lambda}_t\mathrm{d}t=\sum_{i=1}^M\int_{t_{n-1}}^{t_n}\xi_{t|T}^i\lambda_t^i\mathrm{d}t.\notag
    \end{equation}
    Use Equation \eqref{eq:e_step_integral_smoothed_prob_intensity} to deduce the result.
\end{proof}

\subsection{Numerical validation}

In this section, we validate the estimation procedure by conducting numerical experiments on simulated data. We show that the EM algorithm proposed in Section \ref{section:estimation_procedure} provides satisfactory results in terms of convergence with respect to the number of events and the number of EM steps. We repeatedly run the experiment and show the stability of the results across simulations using box plots.

To preserve the Markov property and minimize computation time, we will use exponential kernels in the rest of the paper,

\begin{equation}
    \phi^i(t):=\alpha^i e^{-\beta^it},\quad 1\leq i \leq M.
\end{equation}

However, note that the estimation procedure holds for any general positive kernel.

The parameters used in the experiment are

$$ M=2, \quad \mu=(1.0, 1.0)', \quad \alpha=(1.0, 4.0)', \quad \beta=(2.0, 10.0)',$$
$$ Q = \left(\begin{matrix}-1.0 & 1.0 \\ 1.0 & -1.0\end{matrix}\right), \quad \xi_0=(0.5, 0.5).$$

\paragraph{Convergence with respect to the number of EM steps} We simulate 100 samples of 6,400 events of the MMHP with $\delta$-piecewise constant kernel with $\delta=0.1$. The box plots of the EM estimators as functions of the number of steps are shown in Figure \ref{fig:box_plot_parameters_nem}. The convergence is quickly obtained for all Hawkes parameters as less than 500 steps are required. Concerning the Markov chain parameters, up to 1000 steps are needed to reach convergence. Other parameter settings led to similar results, showing that in empirical applications, the estimated Hawkes parameters might be more reliable than the Markov chain ones if the number of EM steps is small.

\begin{figure}
    \centering
    \subfloat[$\mu_1$]{%
        \includegraphics[width=0.25\linewidth]{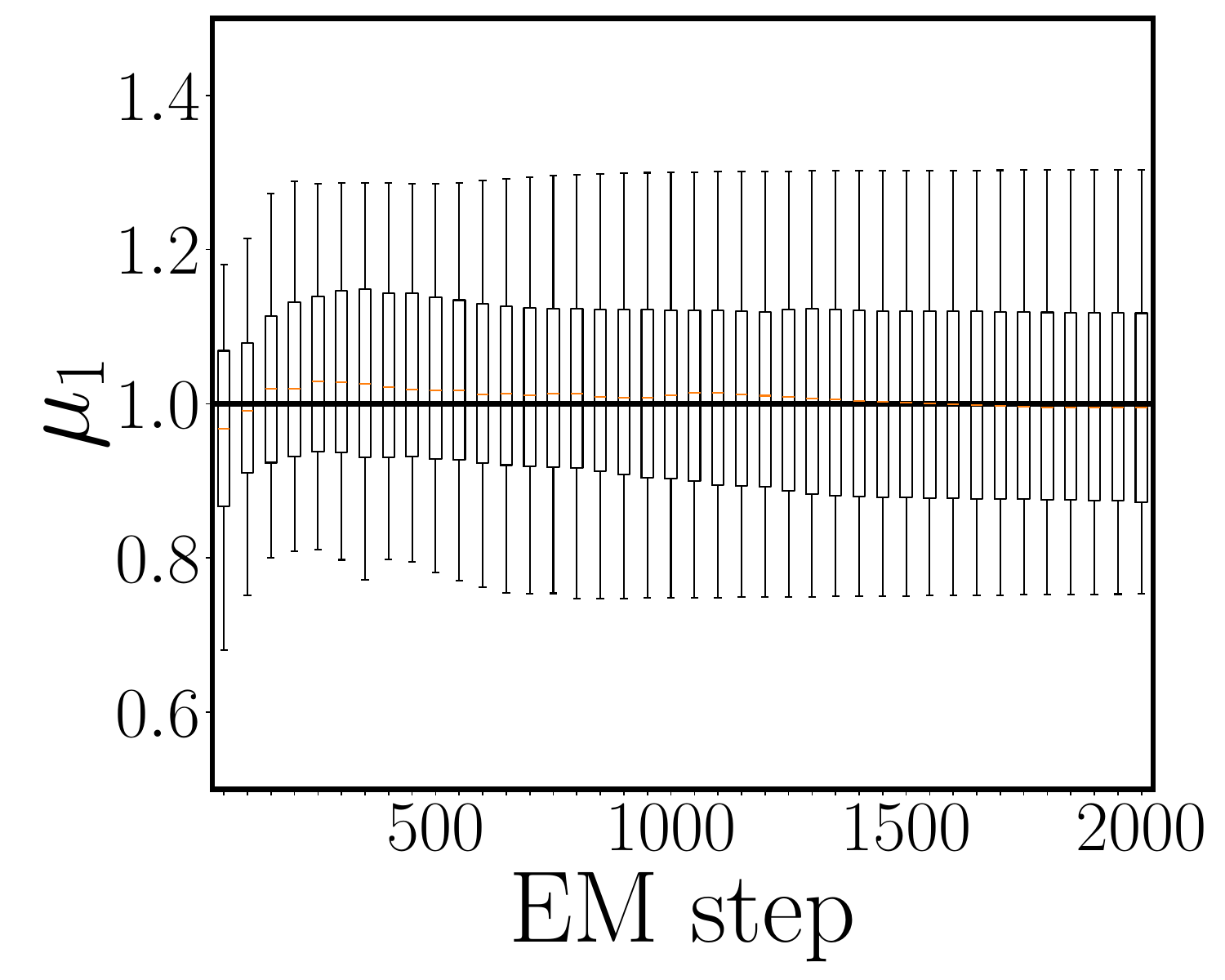}%
    }
    \subfloat[$\mu_2$]{%
        \includegraphics[width=0.25\linewidth]{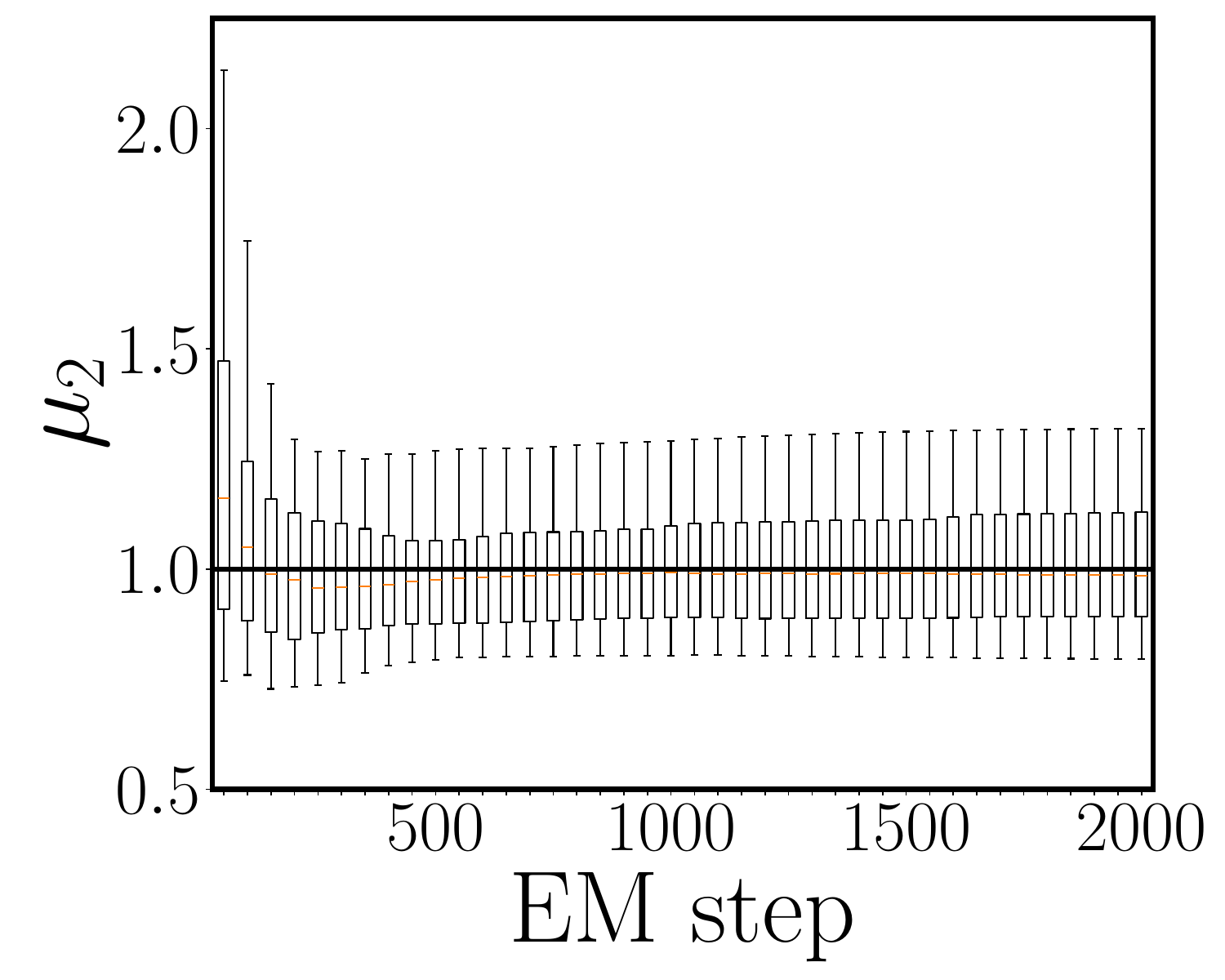}%
    }
    \subfloat[$\alpha_1$]{%
        \includegraphics[width=0.25\linewidth]{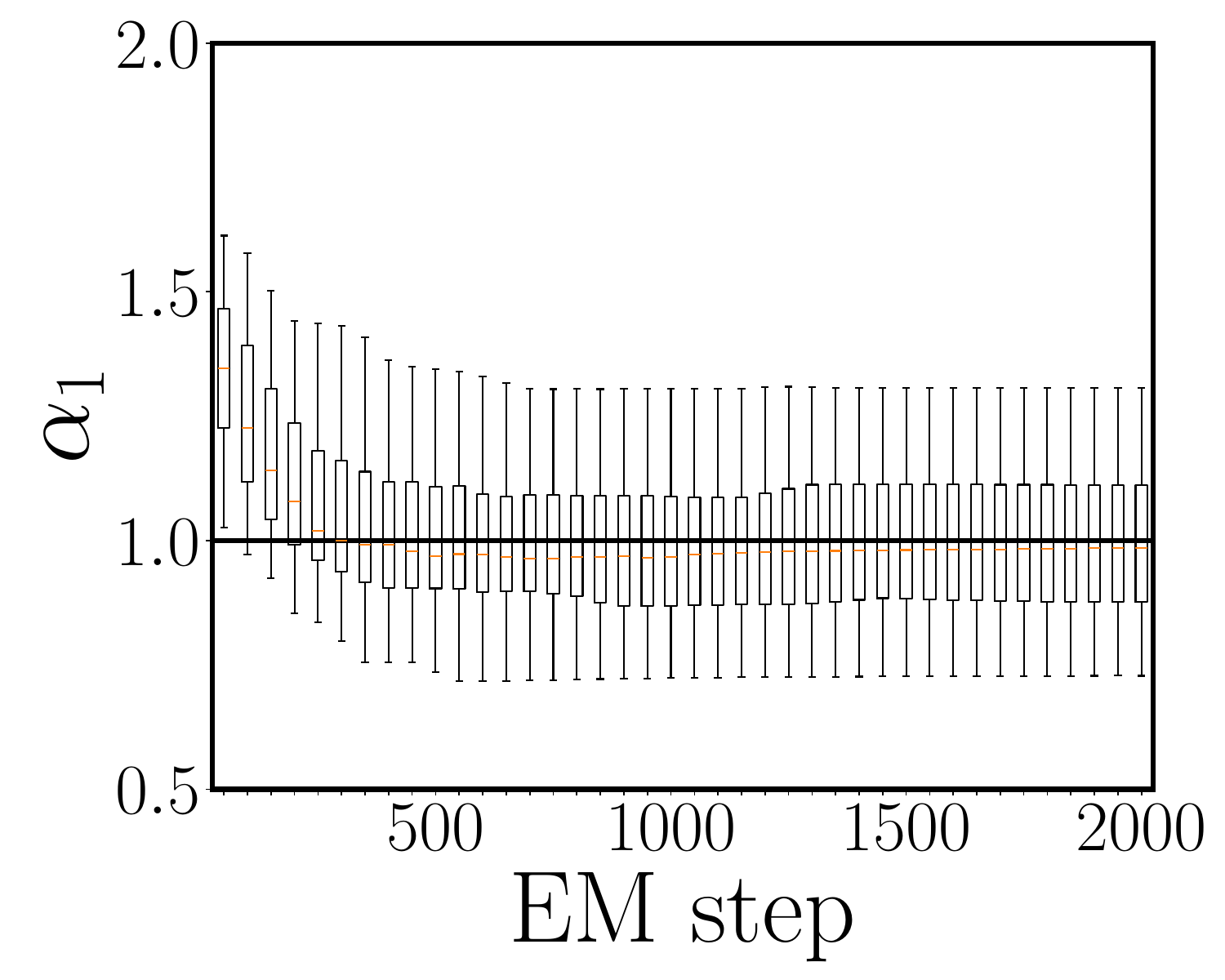}%
    }
    \subfloat[$\alpha_2$]{%
        \includegraphics[width=0.25\linewidth]{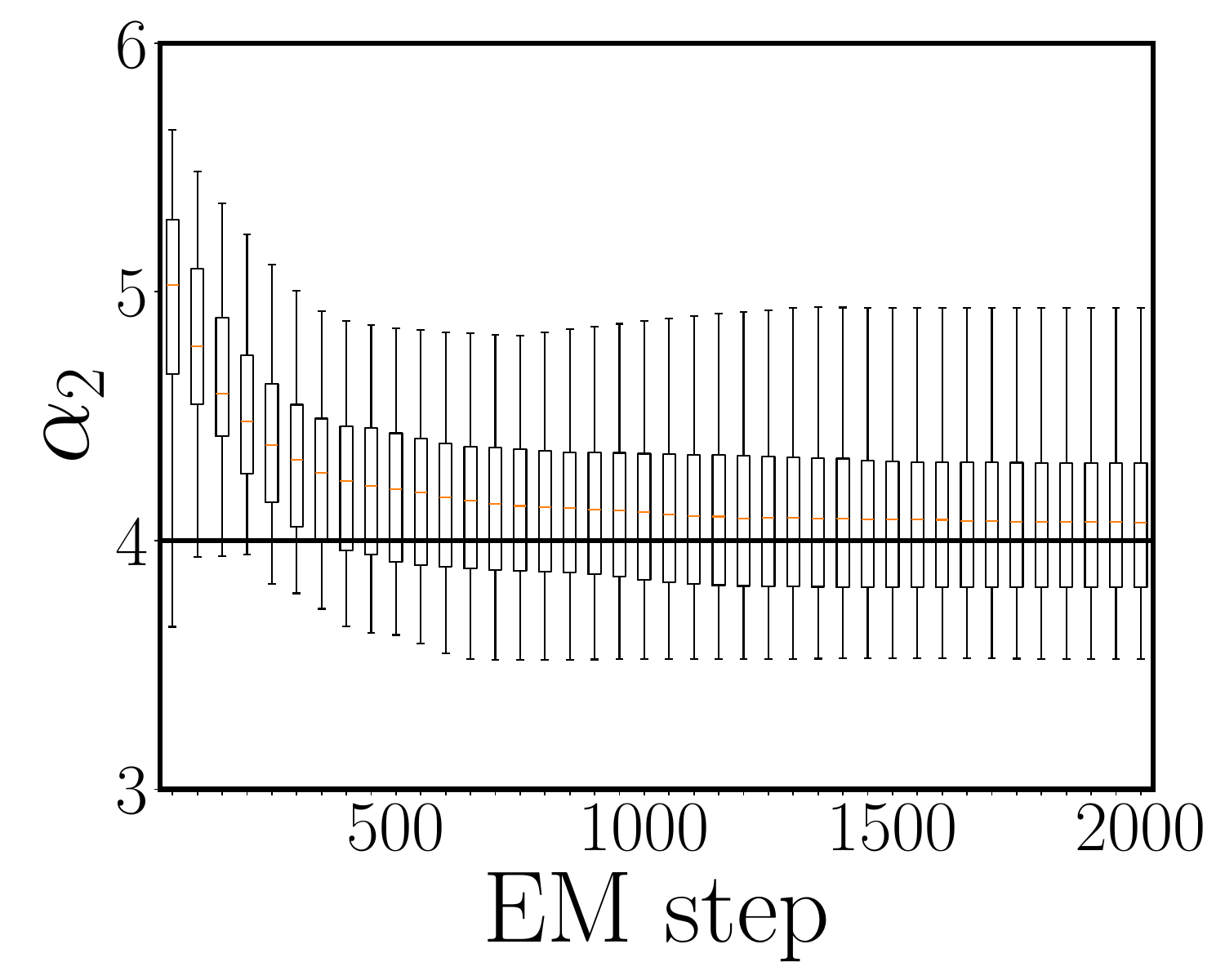}%
    }\\
    \subfloat[$\beta_1$]{%
        \includegraphics[width=0.25\linewidth]{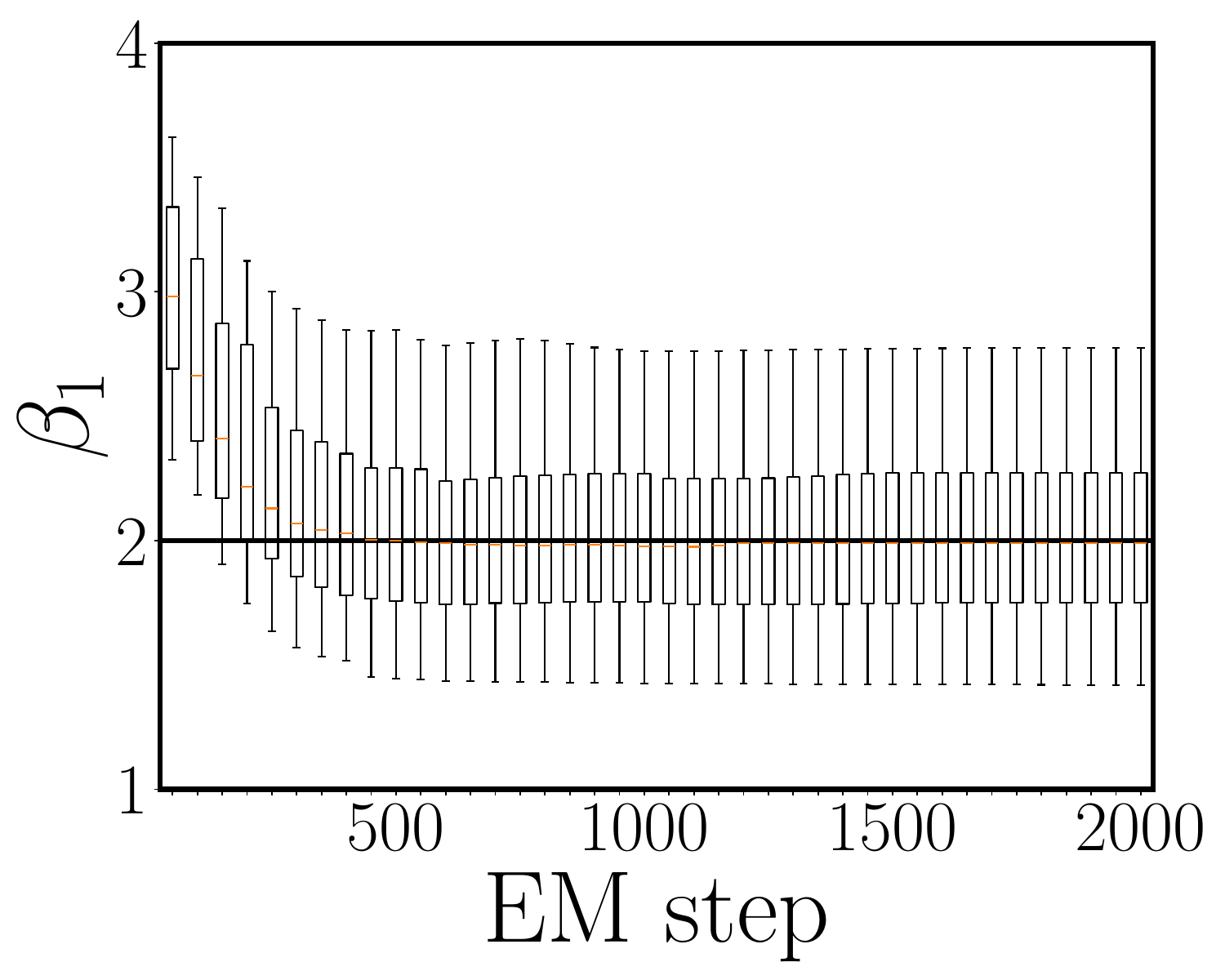}%
    }
    \subfloat[$\beta_2$]{%
        \includegraphics[width=0.25\linewidth]{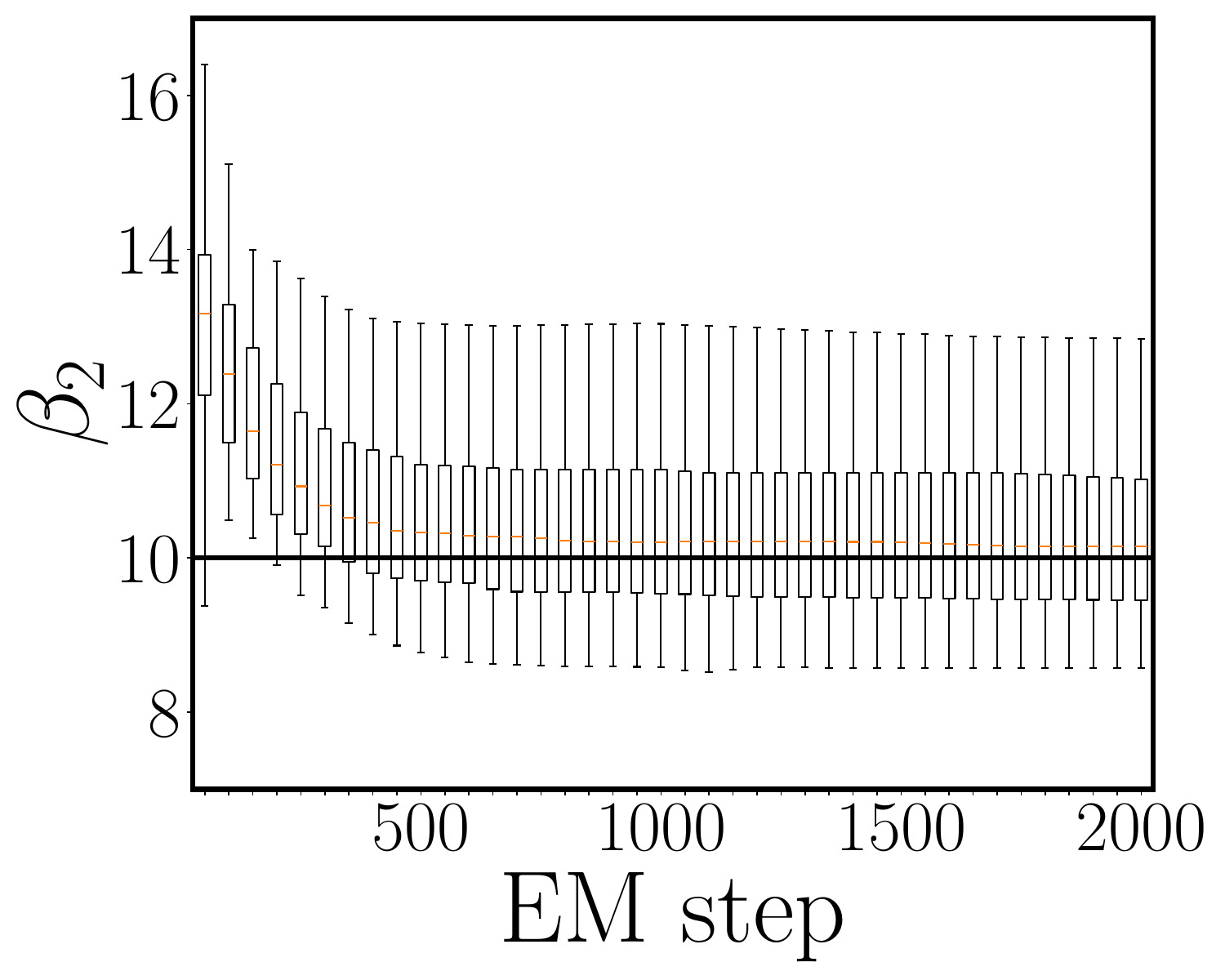}%
    }
    \centering
    \subfloat[$q_1$]{%
        \includegraphics[width=0.25\linewidth]{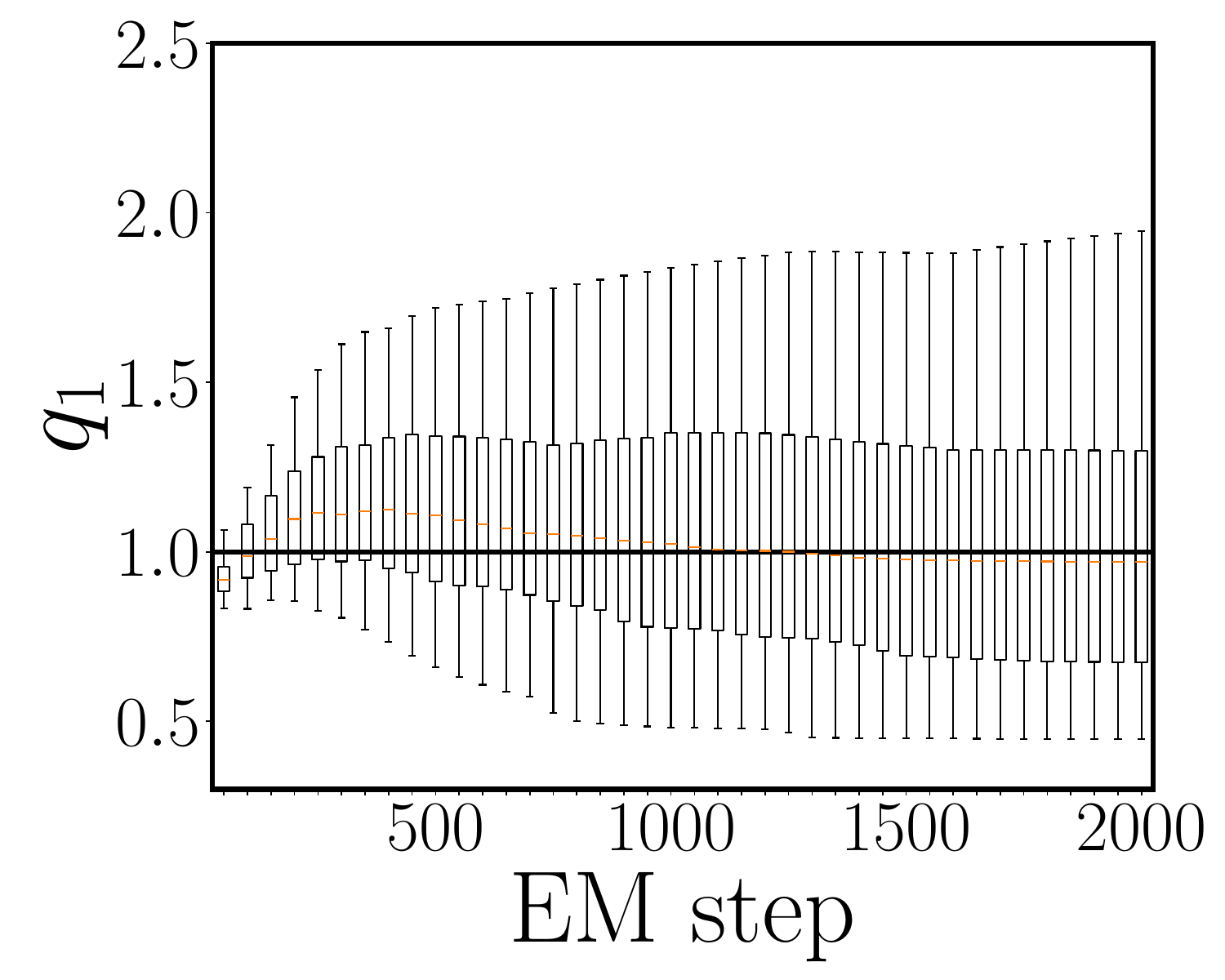}%
    }
    \subfloat[$q_2$]{%
        \includegraphics[width=0.25\linewidth]{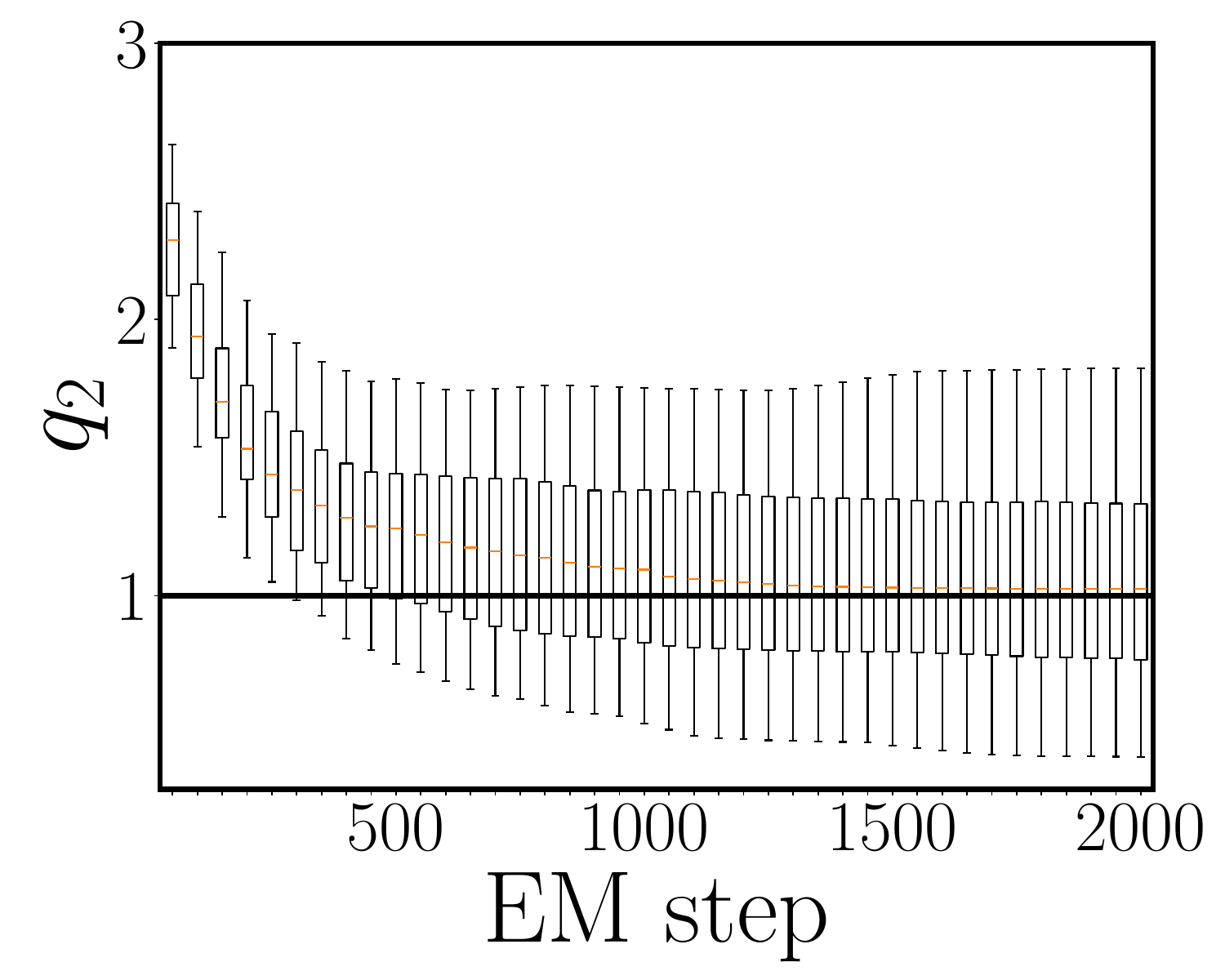}%
    }
    \caption{\textit{Convergence with respect to the number of EM steps} --- Box plots of the parameters with respect to the number of EM steps. Whiskers represent 5\% and 95\% quantiles. The plain line represents the true parameter.}
    \label{fig:box_plot_parameters_nem}
\end{figure}

\paragraph{Convergence with respect to the number of events} We simulate 100 samples of the MMHP with $\delta$-piecewise constant kernel with $\delta=0.1$ and study the convergence of the EM estimators towards the true parameters with respect to the number of events. The number of EM steps is set to 2,000. The boxplots are displayed in Figure \ref{fig:box_plot_parameters_nevents}. We get the desired convergence for all parameters, satisfactory results being obtained for samples of sizes greater than 1,600 events. We also display the convergence of the standard deviations of the estimated parameters amongst the 100 simulations on Figure \ref{fig:box_plot_parameters_rate_nevents}. Slopes of order $K^{-\gamma}$ with $\gamma\in\{0.5,1\}$ are plotted for visual reference.

\begin{figure}
    \centering
    \subfloat[$\mu_1$]{%
        \includegraphics[width=0.25\linewidth]{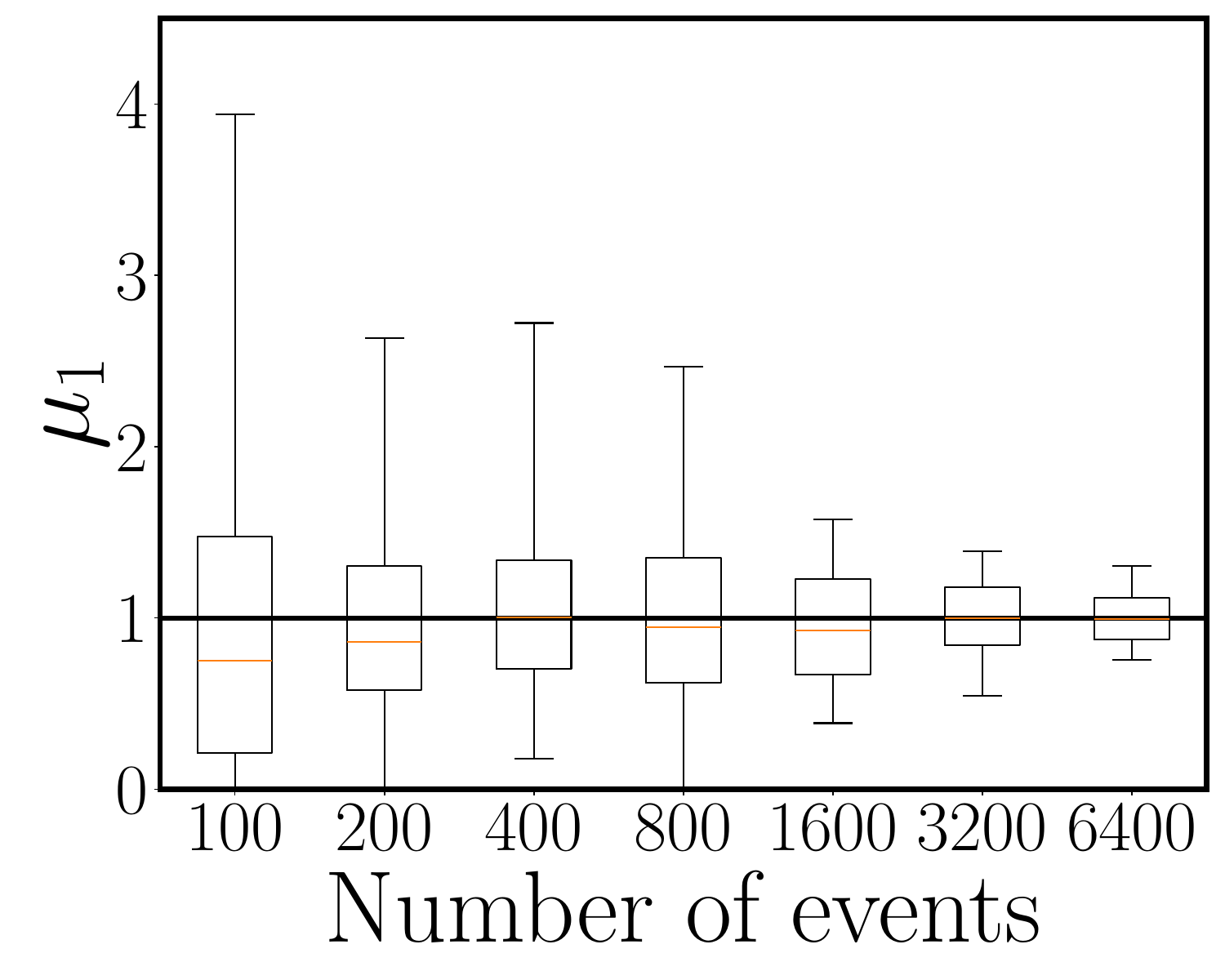}%
    }
    \subfloat[$\mu_2$]{%
        \includegraphics[width=0.25\linewidth]{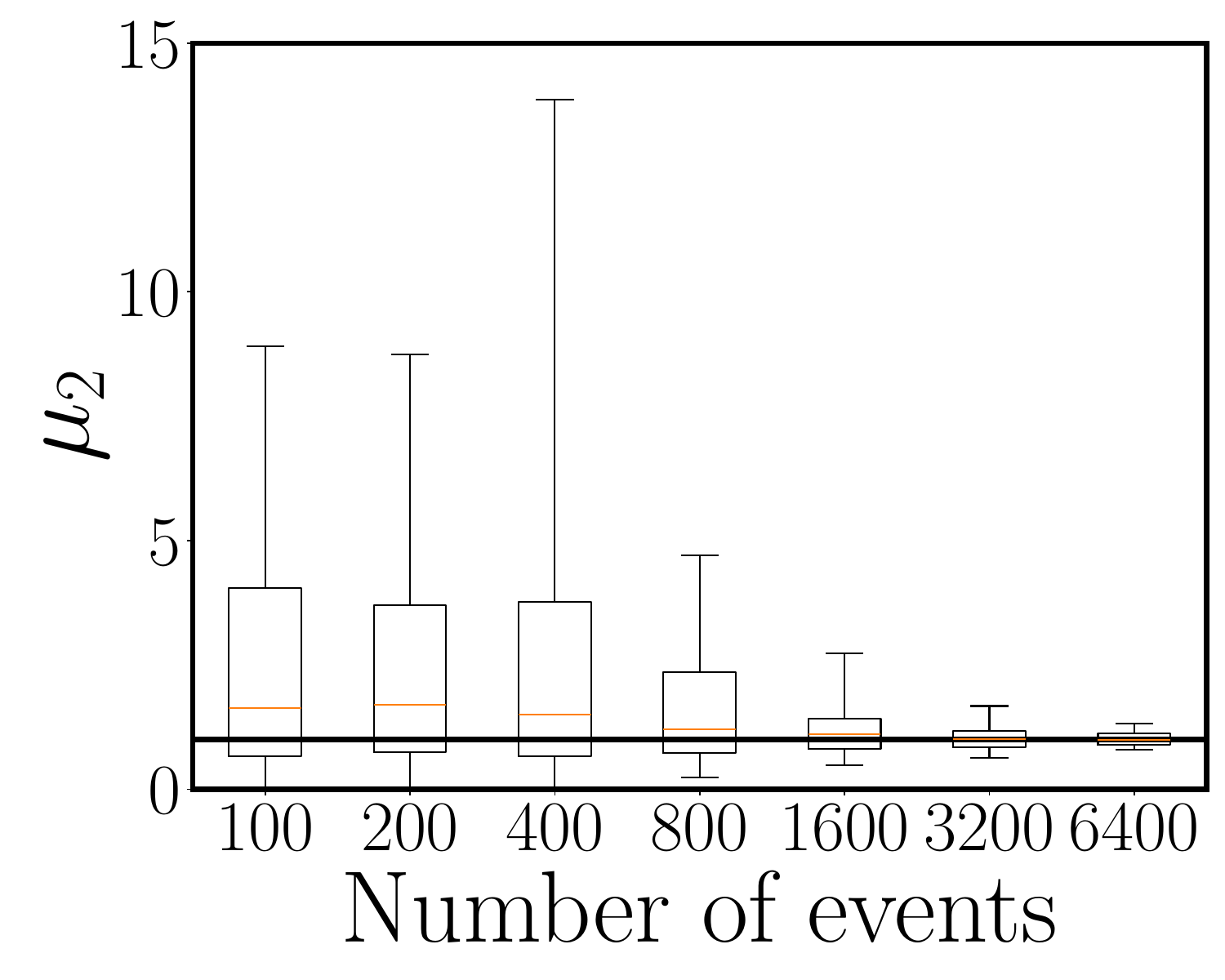}%
    }
    \subfloat[$\alpha_1$]{%
        \includegraphics[width=0.25\linewidth]{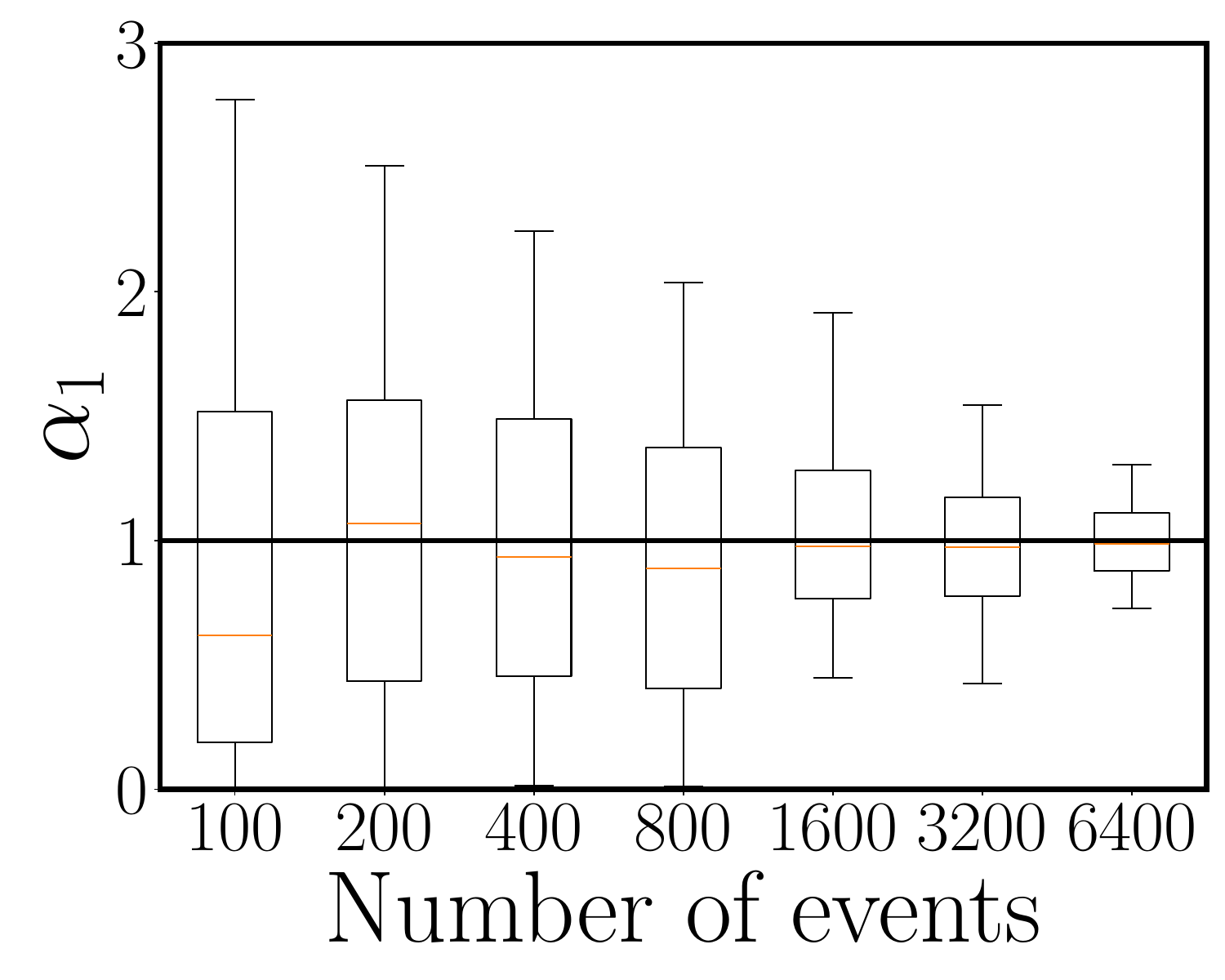}%
    }
    \subfloat[$\alpha_2$]{%
        \includegraphics[width=0.25\linewidth]{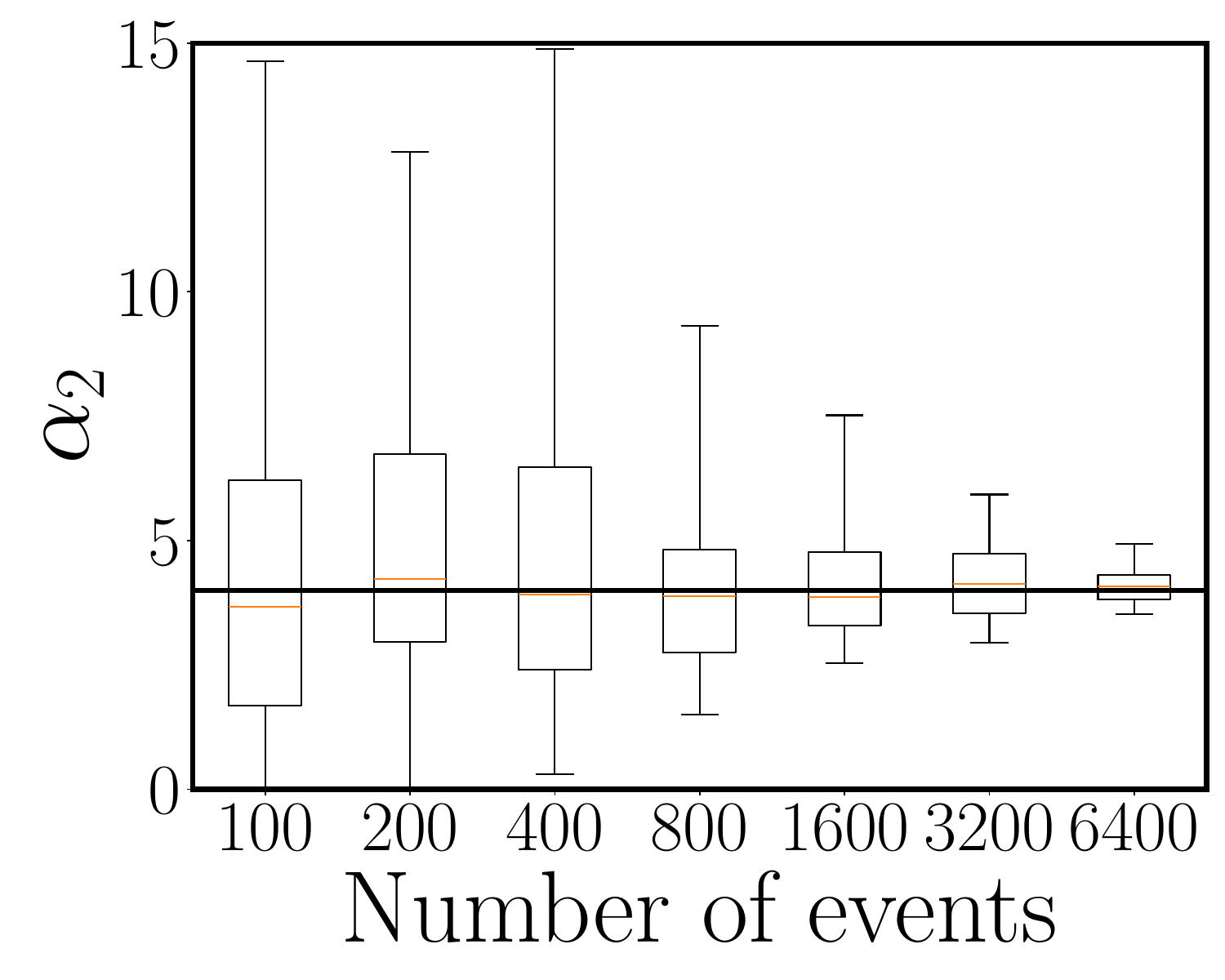}%
    }\\
    \subfloat[$\beta_1$]{%
        \includegraphics[width=0.25\linewidth]{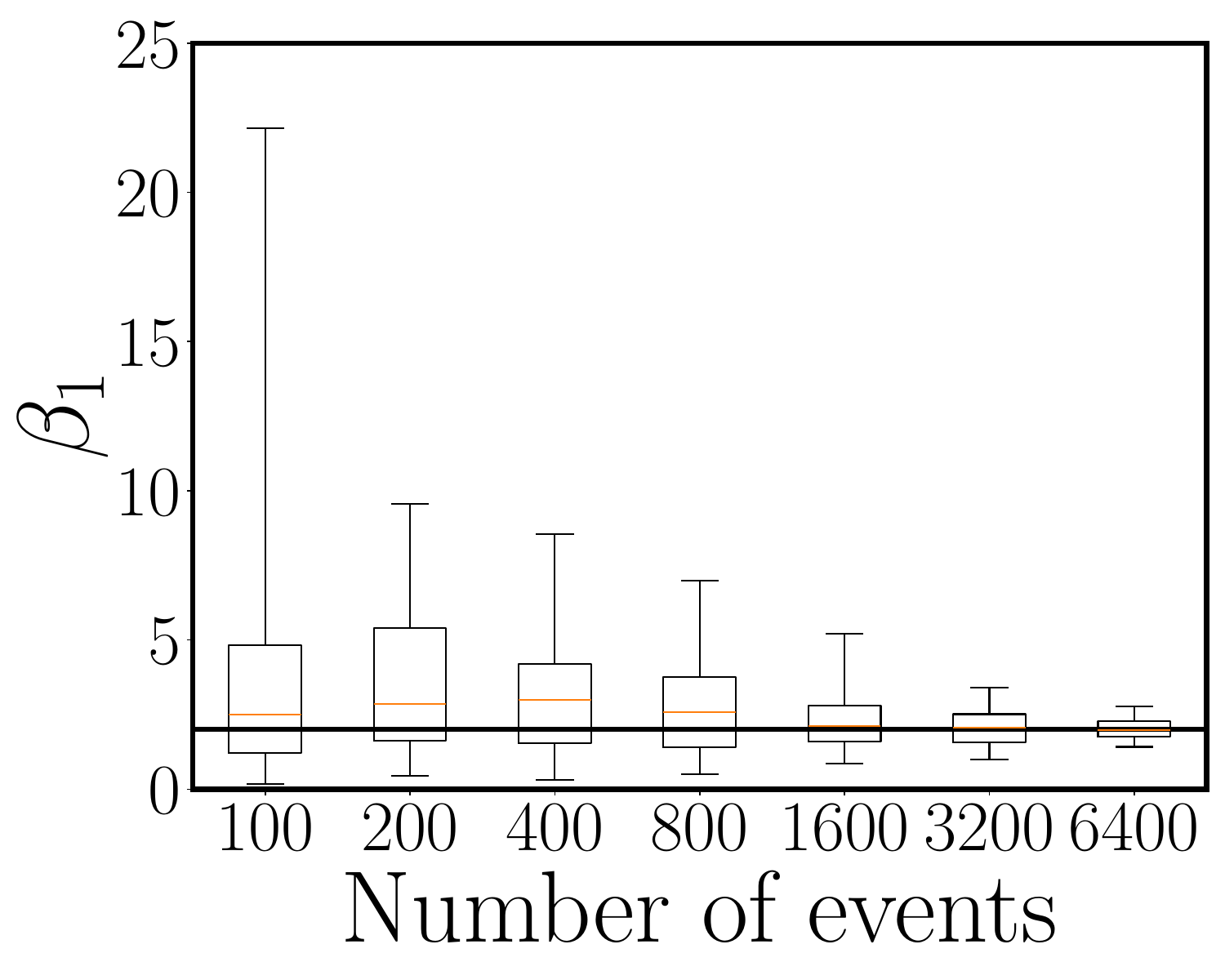}%
    }
    \subfloat[$\beta_2$]{%
        \includegraphics[width=0.25\linewidth]{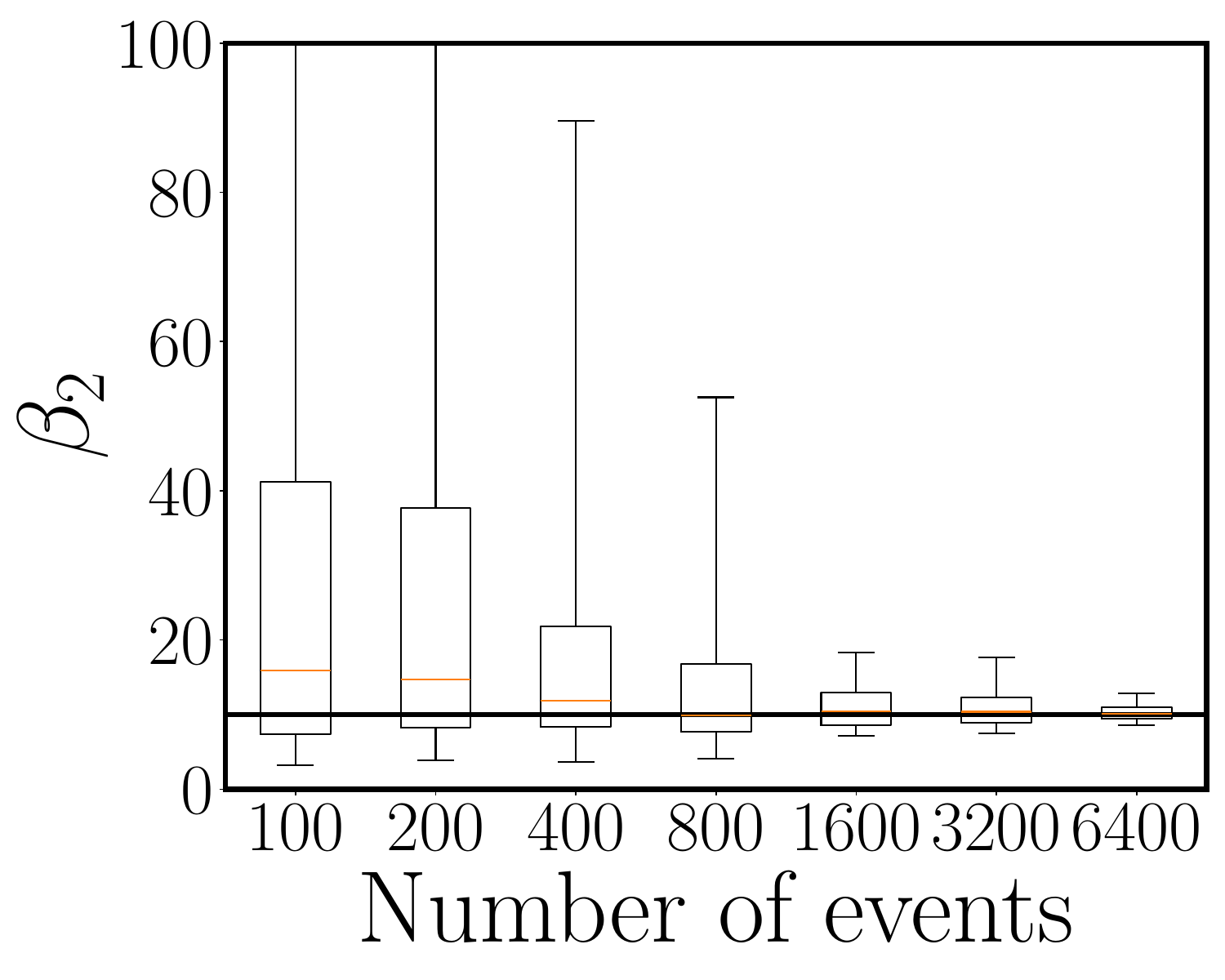}%
    }
    \centering
    \subfloat[$q_1$]{%
        \includegraphics[width=0.25\linewidth]{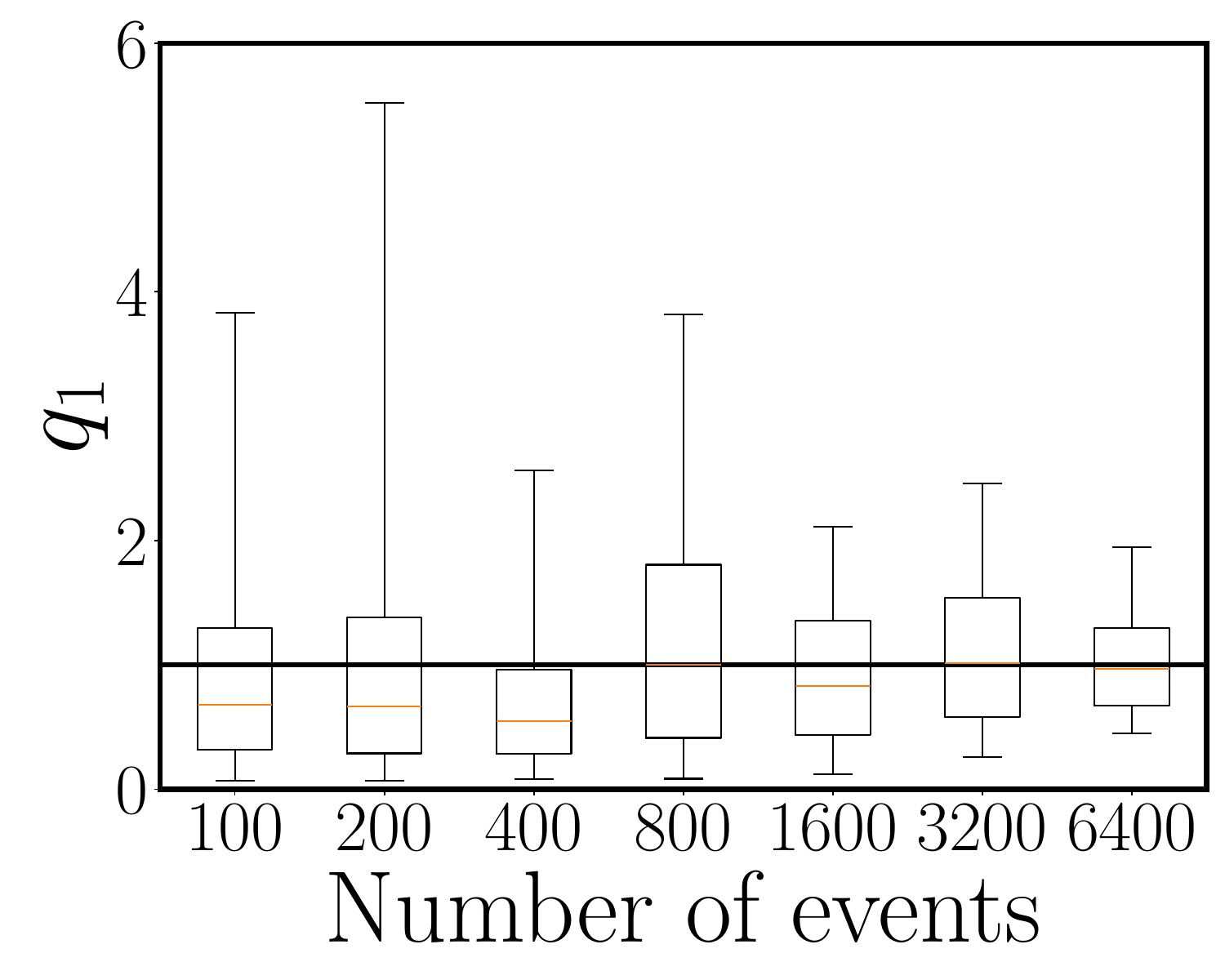}%
    }
    \subfloat[$q_2$]{%
        \includegraphics[width=0.25\linewidth]{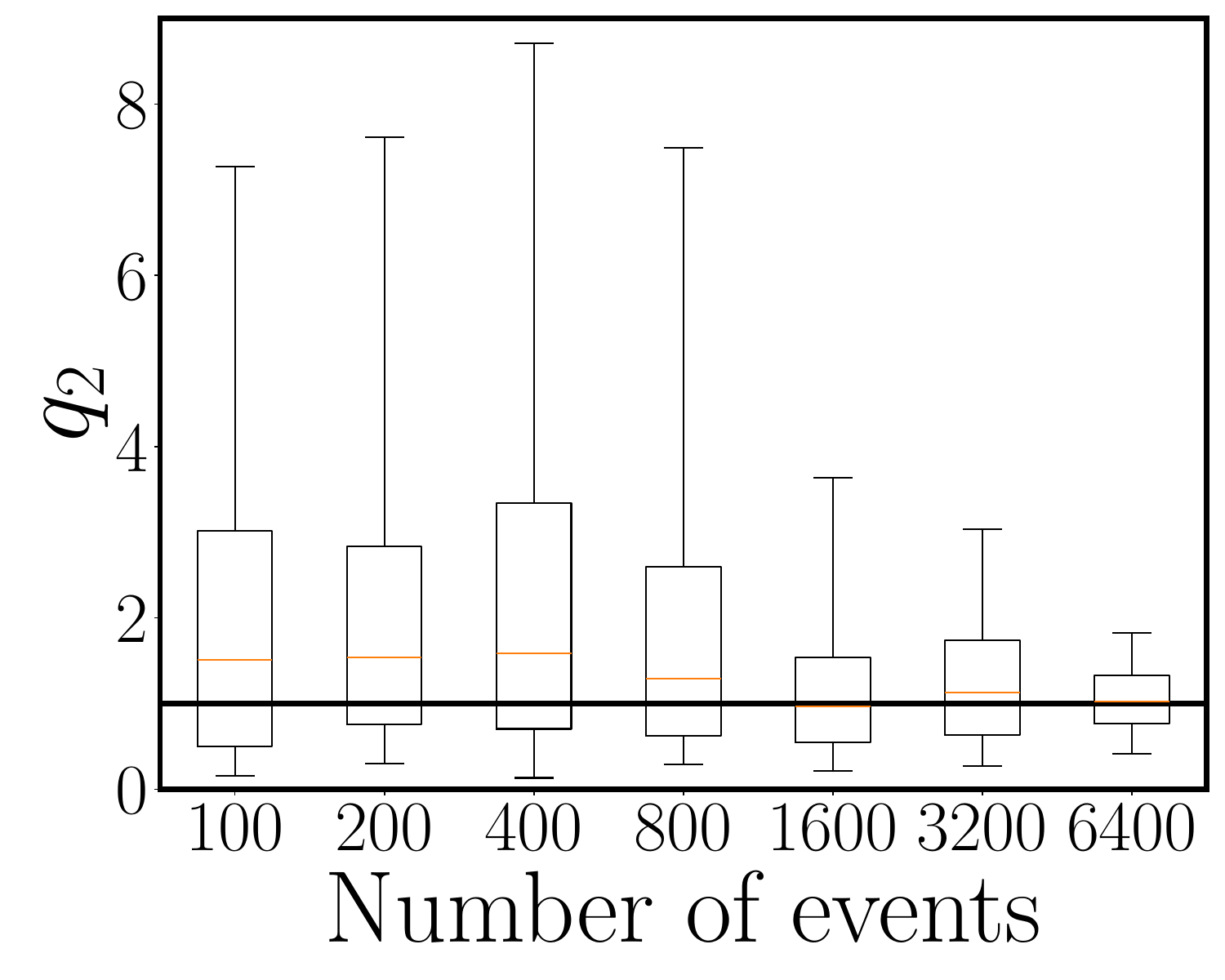}%
    }
    \caption{\textit{Convergence with respect to the number of events} --- Box plots of the parameters with respect to the number of events. Whiskers represent 5\% and 95\% quantiles. The plain line represents the true parameter.}
    \label{fig:box_plot_parameters_nevents}
\end{figure}

\begin{figure}
    \centering
    \subfloat[$\mu_1$]{%
        \includegraphics[width=0.25\linewidth]{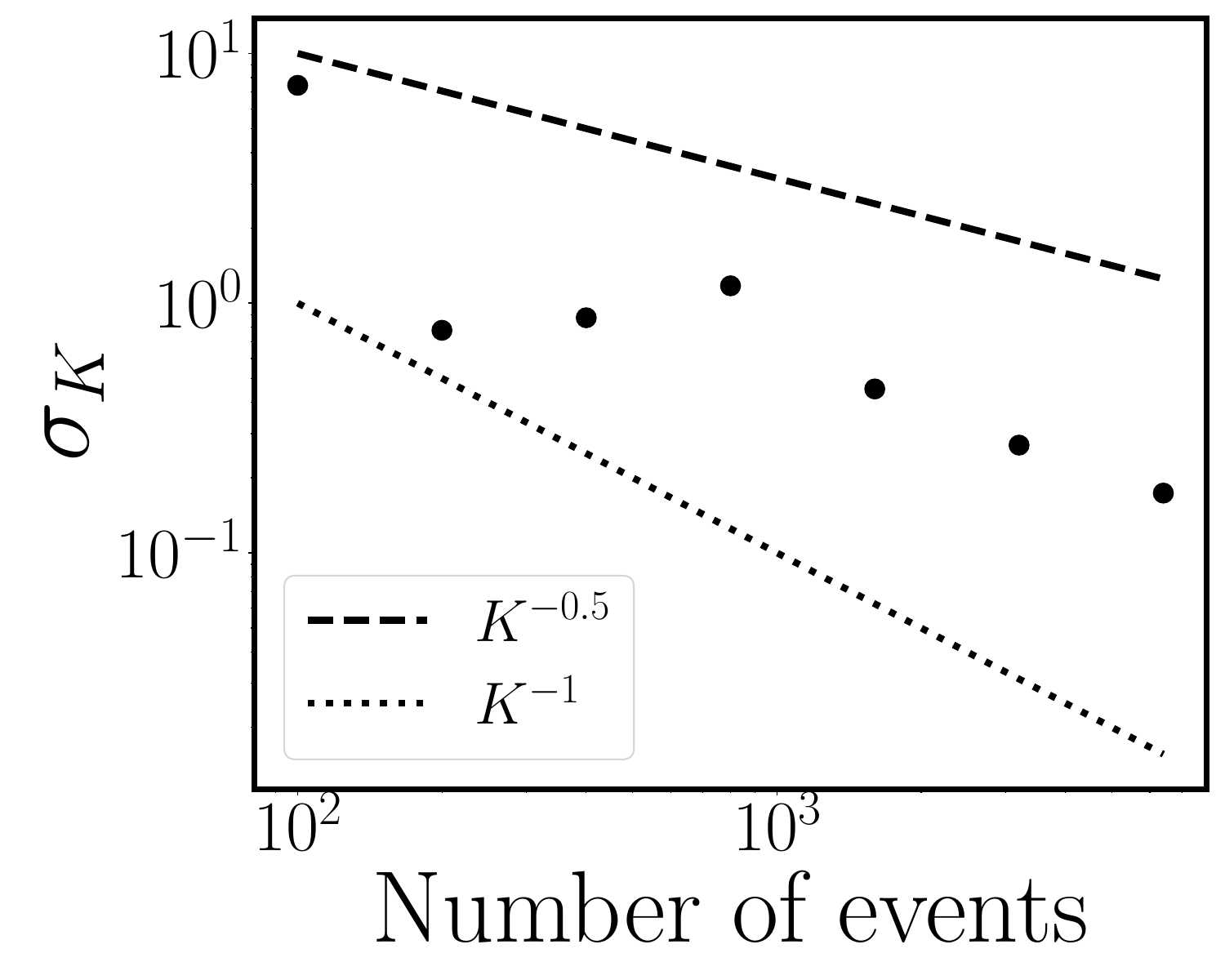}%
    }
    \subfloat[$\mu_2$]{%
        \includegraphics[width=0.25\linewidth]{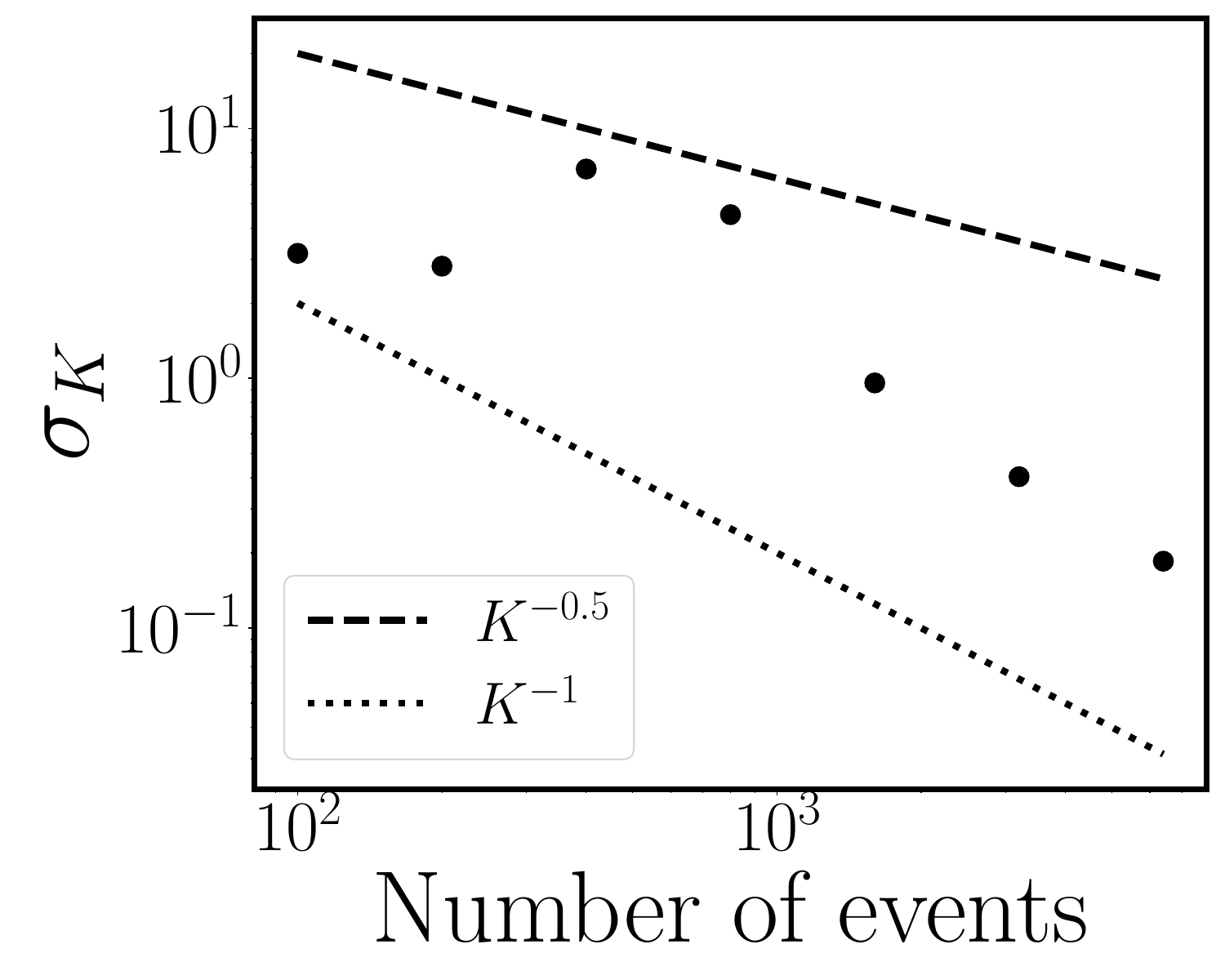}%
    }
    \subfloat[$\alpha_1$]{%
        \includegraphics[width=0.25\linewidth]{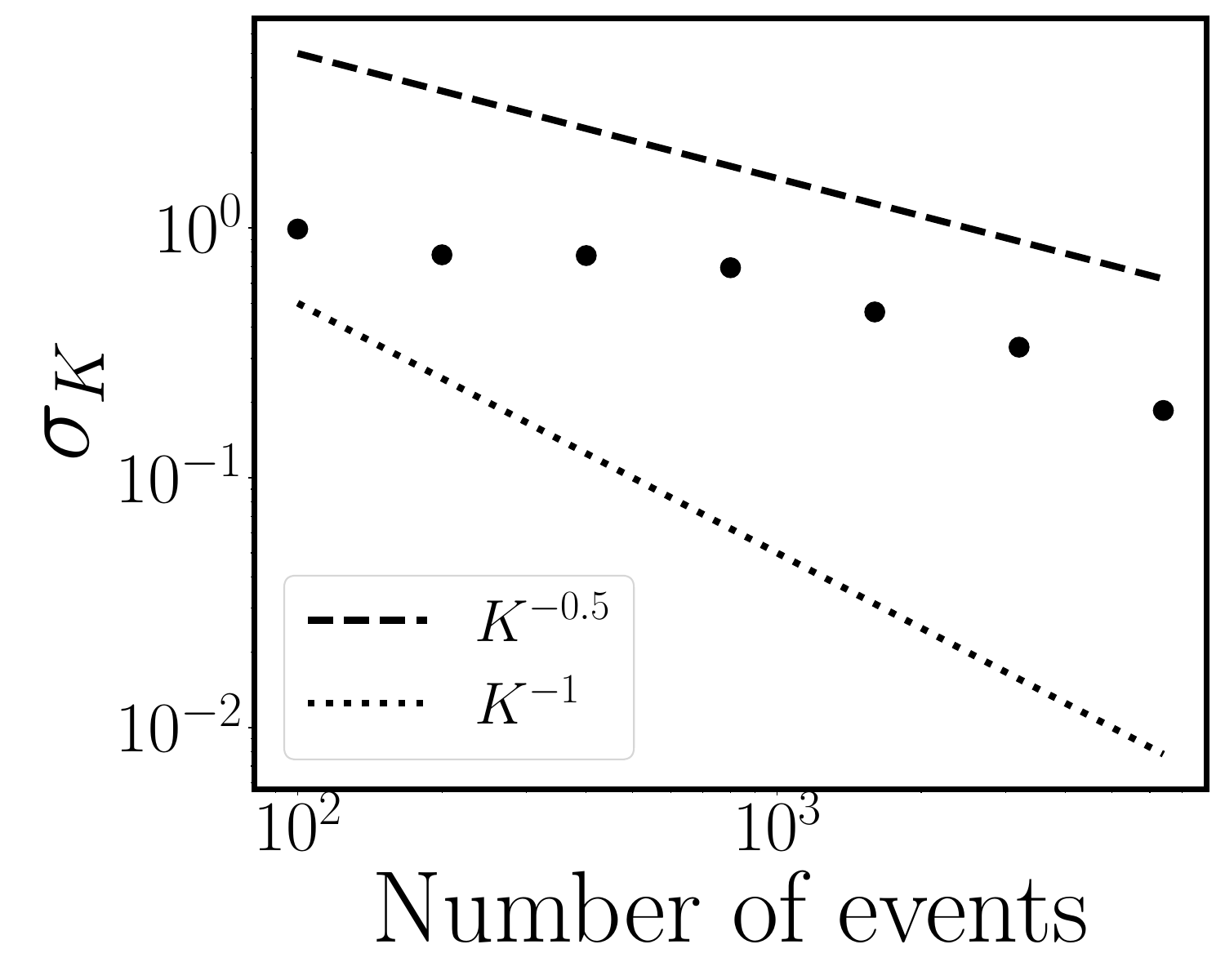}%
    }
    \subfloat[$\alpha_2$]{%
        \includegraphics[width=0.25\linewidth]{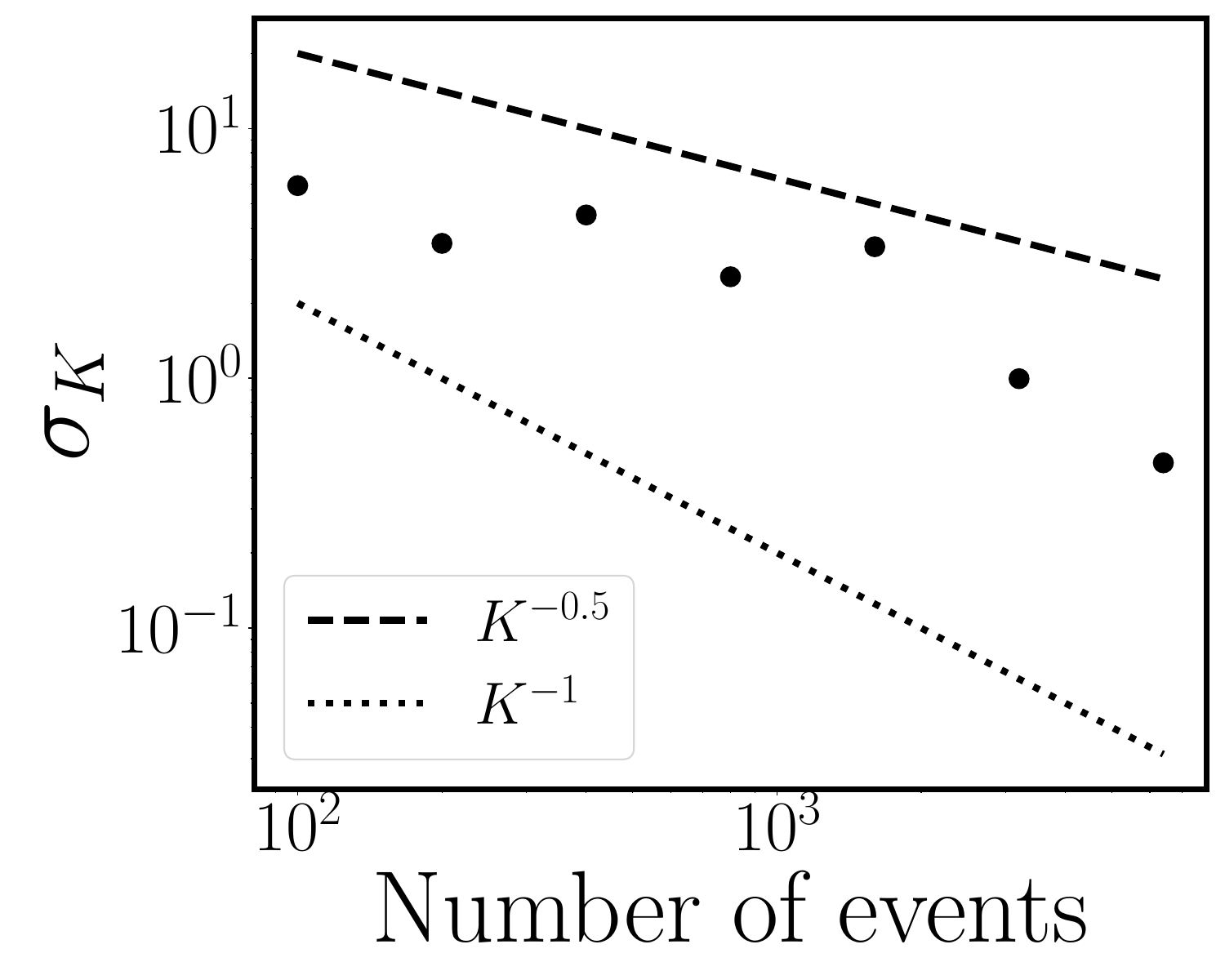}%
    }\\
    \subfloat[$\beta_1$]{%
        \includegraphics[width=0.25\linewidth]{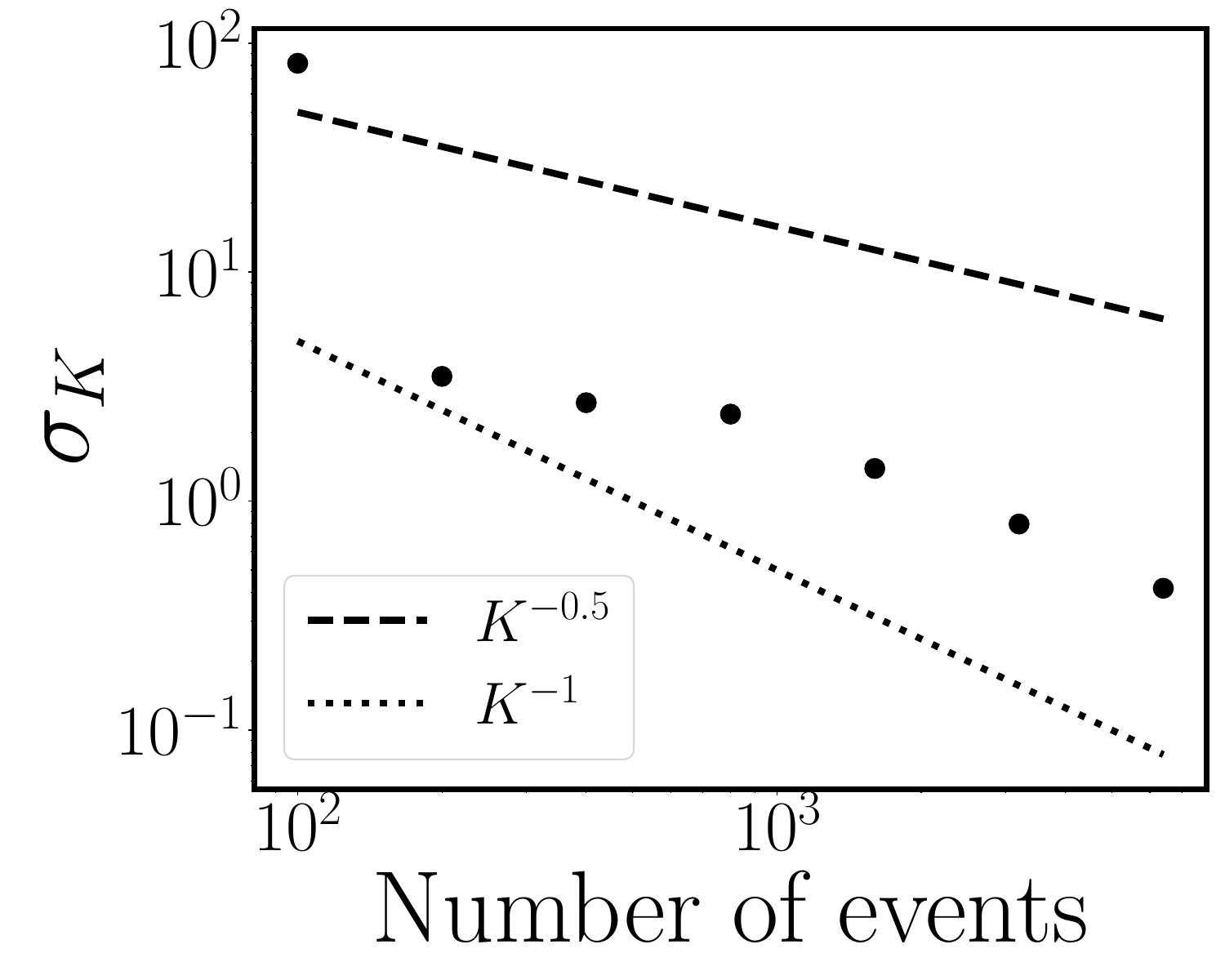}%
    }
    \subfloat[$\beta_2$]{%
        \includegraphics[width=0.25\linewidth]{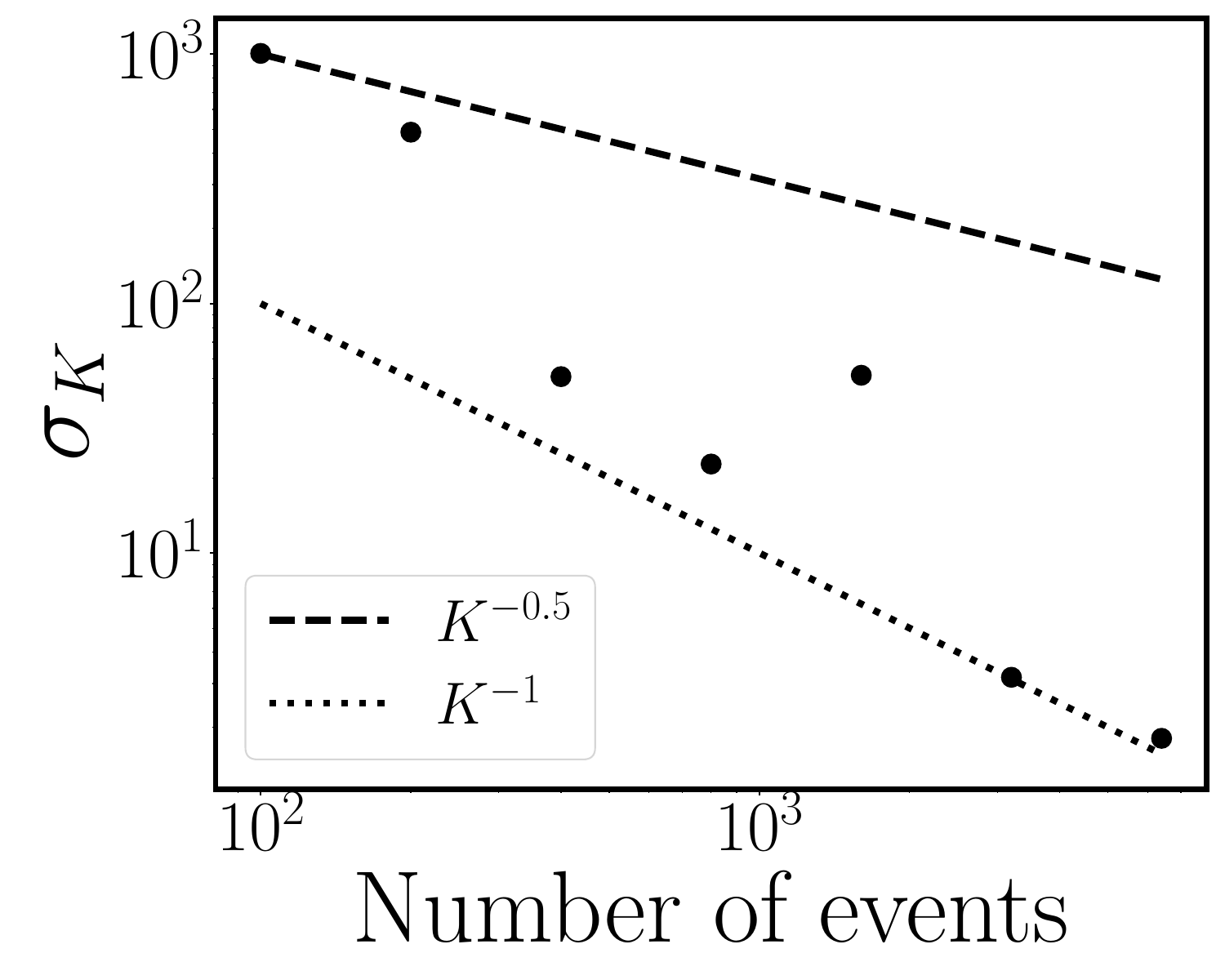}%
    }
    \centering
    \subfloat[$q_1$]{%
        \includegraphics[width=0.25\linewidth]{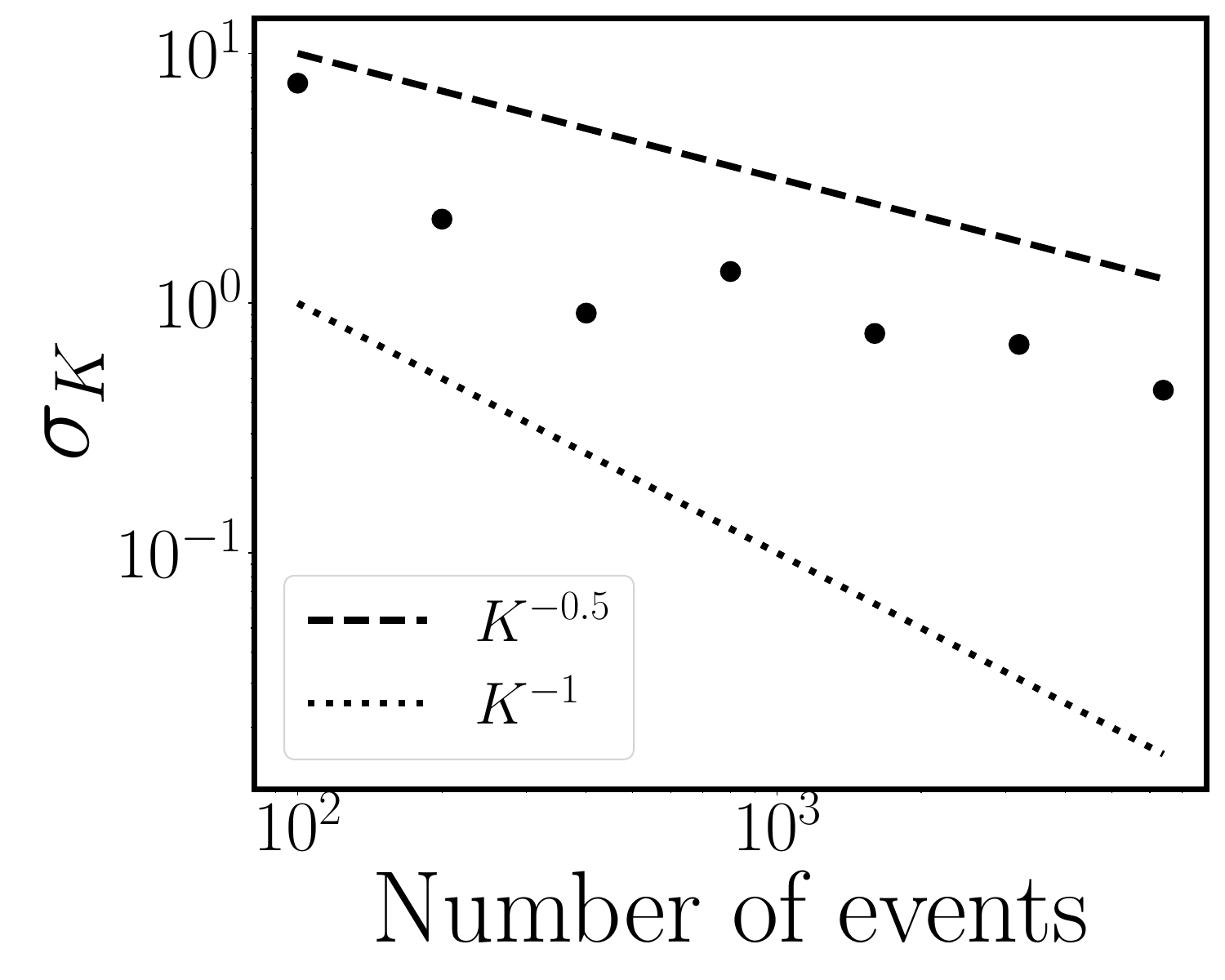}%
    }
    \subfloat[$q_2$]{%
        \includegraphics[width=0.25\linewidth]{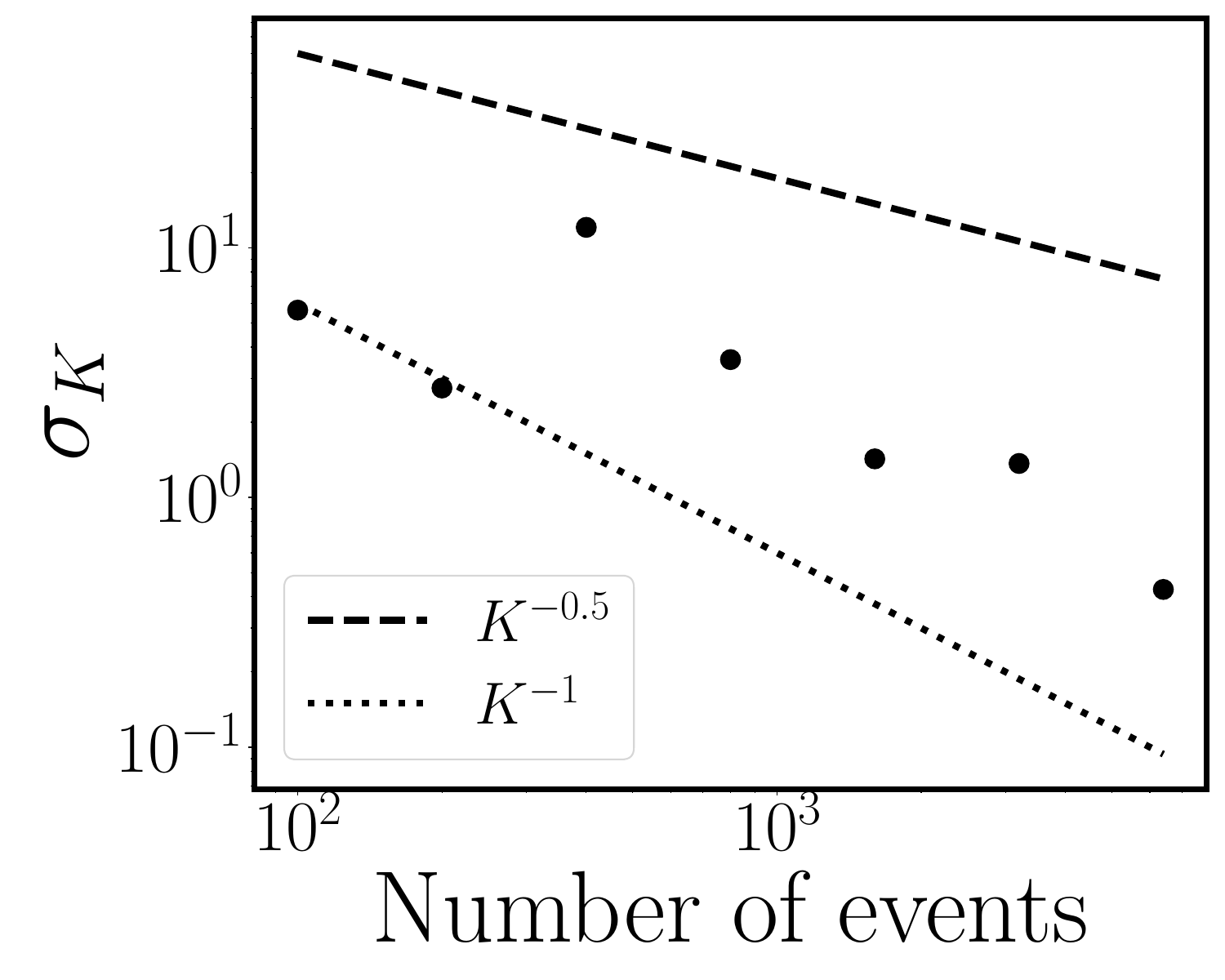}%
    }
    \caption{\textit{Convergence with respect to the number of events} --- Standard deviation of the estimated parameters with respect to the number of events.}
    \label{fig:box_plot_parameters_rate_nevents}
\end{figure}

\paragraph{Quality of the estimation with respect to $\delta$} We simulate 100 samples of 6,400 events of the MMHP-$\delta$ with $\delta\in\{0.01, 0.1, 1\}$ . The boxplots of the EM estimators as functions of $\delta$ are displayed in appendix, in Figure \ref{fig:box_plot_parameters_convergence_delta}. We observe that the algorithm is able to retrieve the true parameters in all cases.

\paragraph{Convergence towards the MMHP with continuous kernel} As $\delta\rightarrow0$, the kernel with $\delta$-piecewise constant kernel should tend to its continuous version. We now validate this assertion on simulations. We simulate 100 samples of the MMHP with continuous decay and estimate its $\delta$-stepwise decay version with varying $\delta$. After each estimation, we compute the transformed inter-arrival times given by Equation \eqref{eq:gof_compensator} and compute the p-value of the Kolmogorov-Smirnov test with respect to the exponential distribution. The results are displayed in Figure \ref{fig:qq_box_plots_convergence_continuous_kernel}. We show the box plots of the p-values with respect to the number of events and the average QQ-plot with respect to the exponential distribution. We observe that as $\delta\rightarrow0$, we are not able to reject the null hypothesis of an exponentially distributed transformed inter-arrival time, meaning that the model approximates its continuous counterpart for small $\delta$. The QQ plot provides more evidence as it becomes closer to the line $y=x$ as $\delta$ gets smaller. Additional material is displayed in Figure \ref{fig:box_plot_parameters_convergence_continuous_kernel} of the Appendix concerning the convergence of the parameters towards the true values of the continuous MMHP.

\begin{figure}
    \centering
    \subfloat[KS test, $\delta = 1$]{%
        \includegraphics[width=0.25\linewidth]{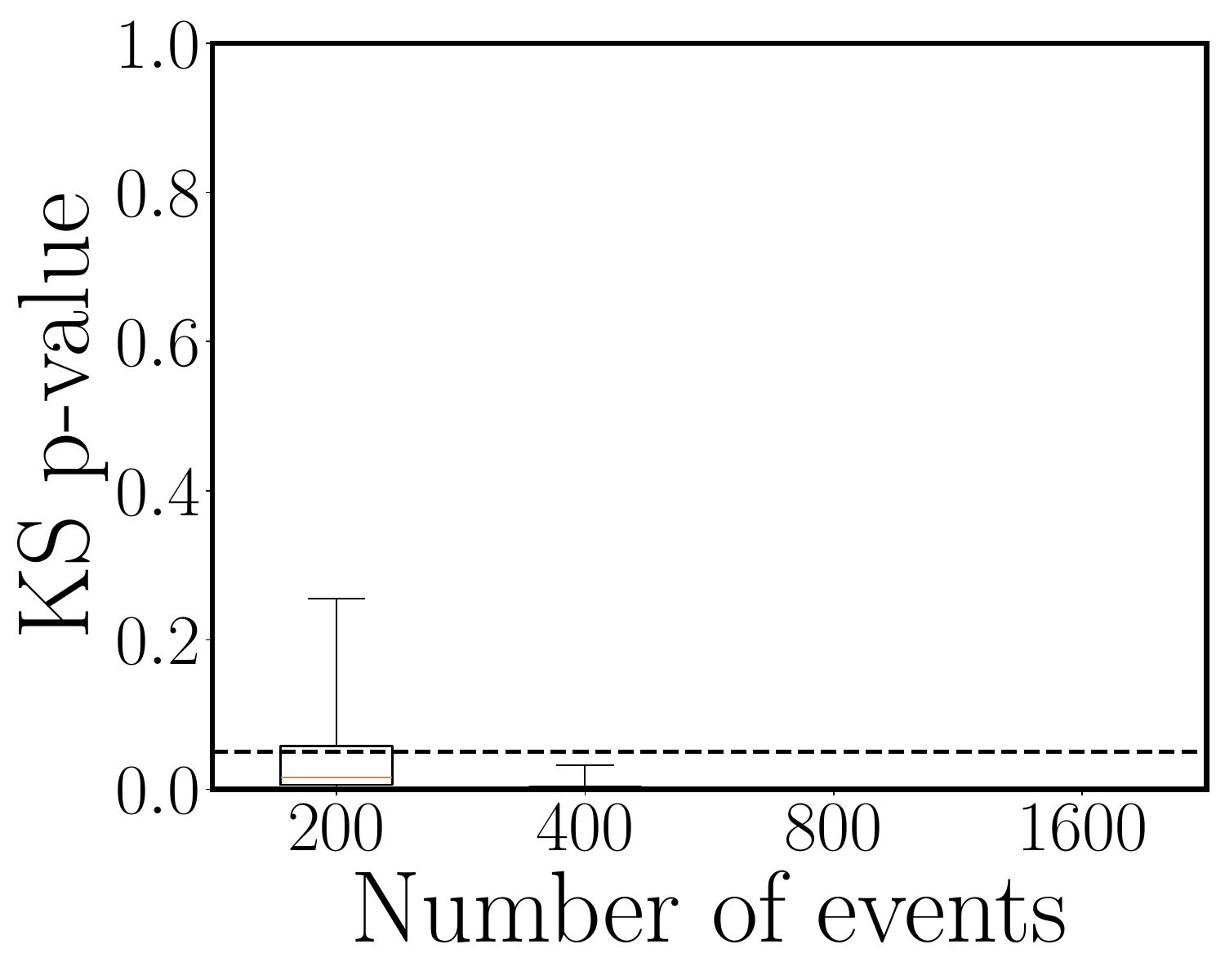}%
    }
    \subfloat[KS test, $\delta = 0.1$]{%
        \includegraphics[width=0.25\linewidth]{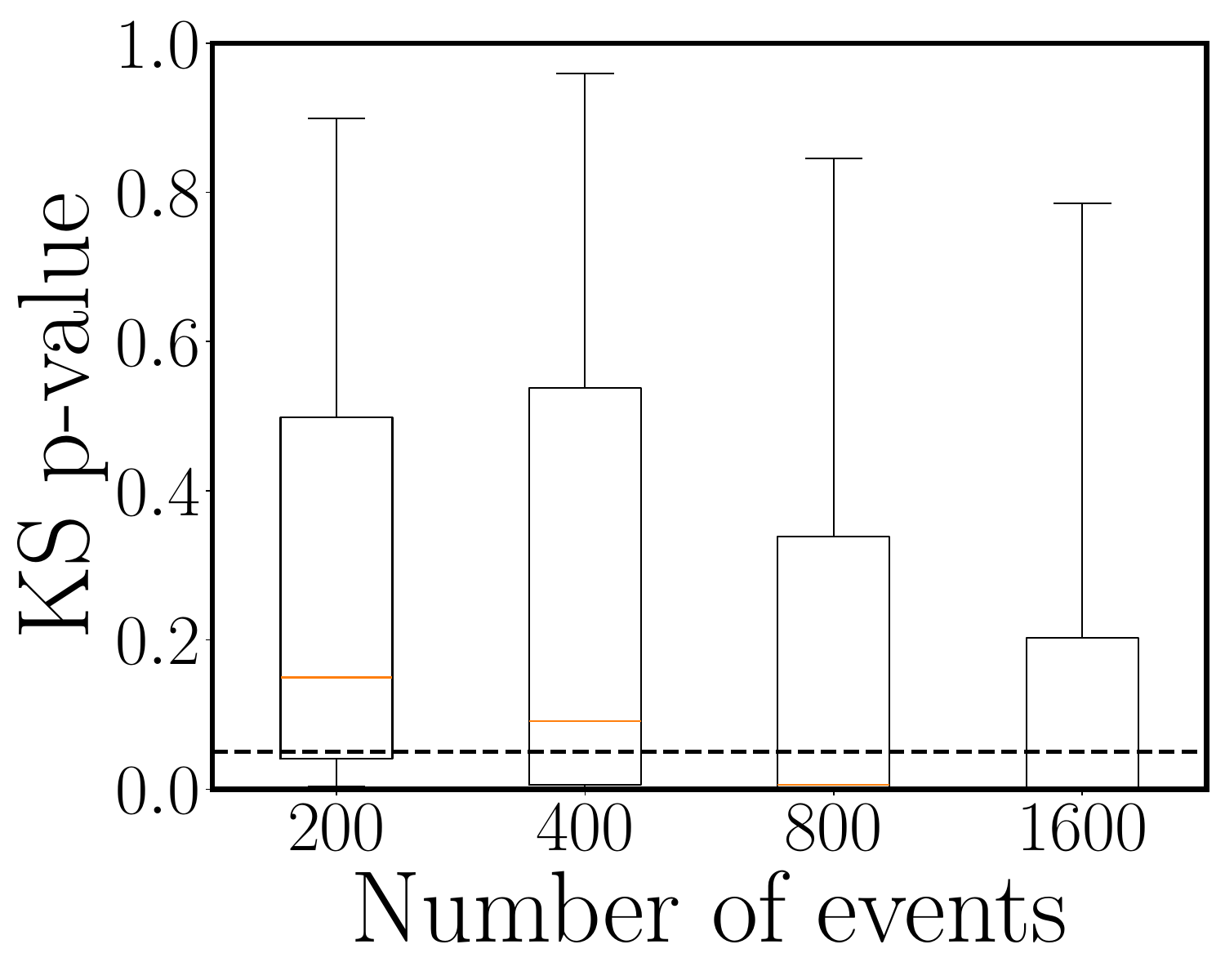}%
    }
    \subfloat[KS test, $\delta = 0.01$]{%
        \includegraphics[width=0.25\linewidth]{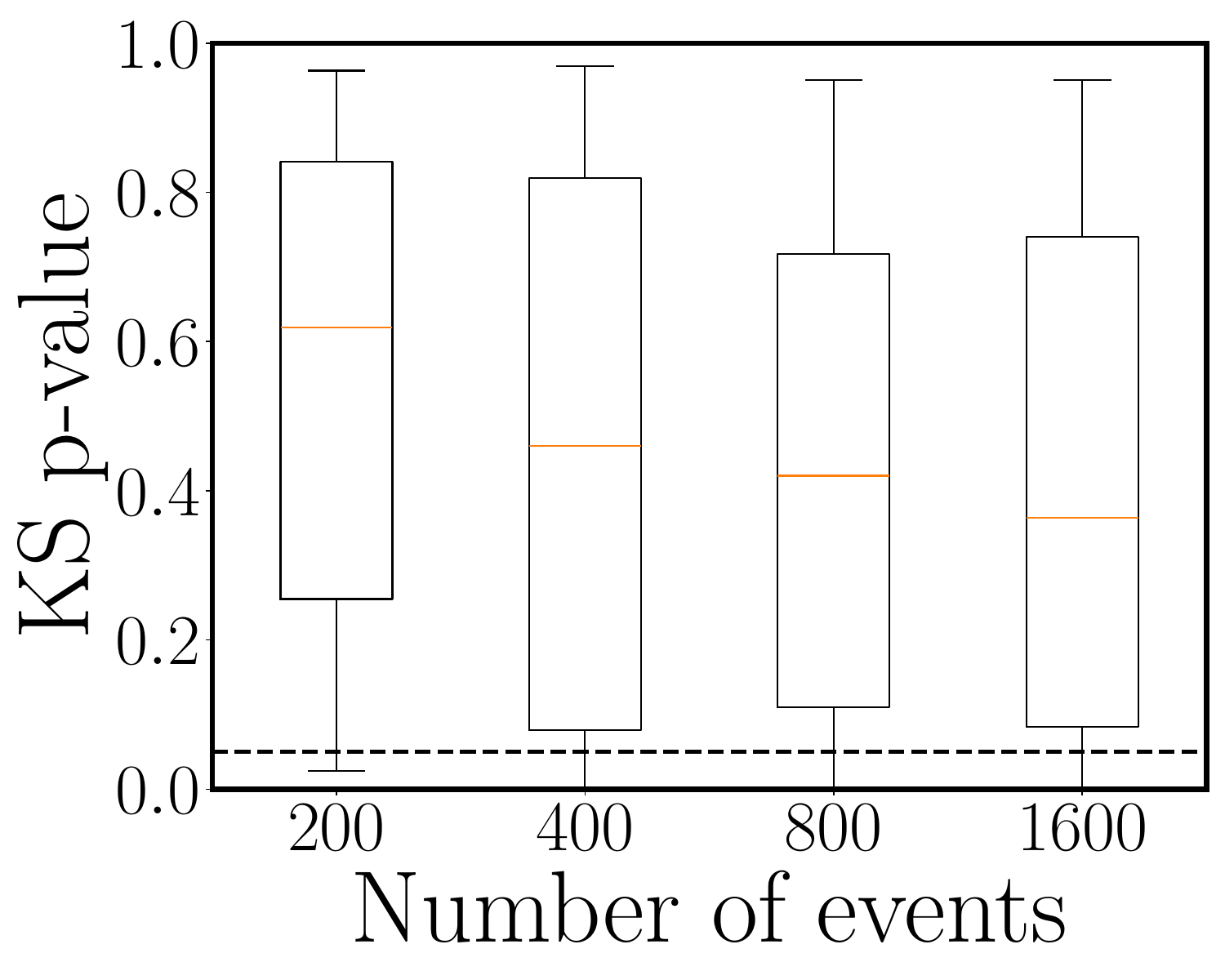}%
    }
    \subfloat[Average QQ plots and 95\% confidence intervals]{%
        \includegraphics[width=0.25\linewidth]{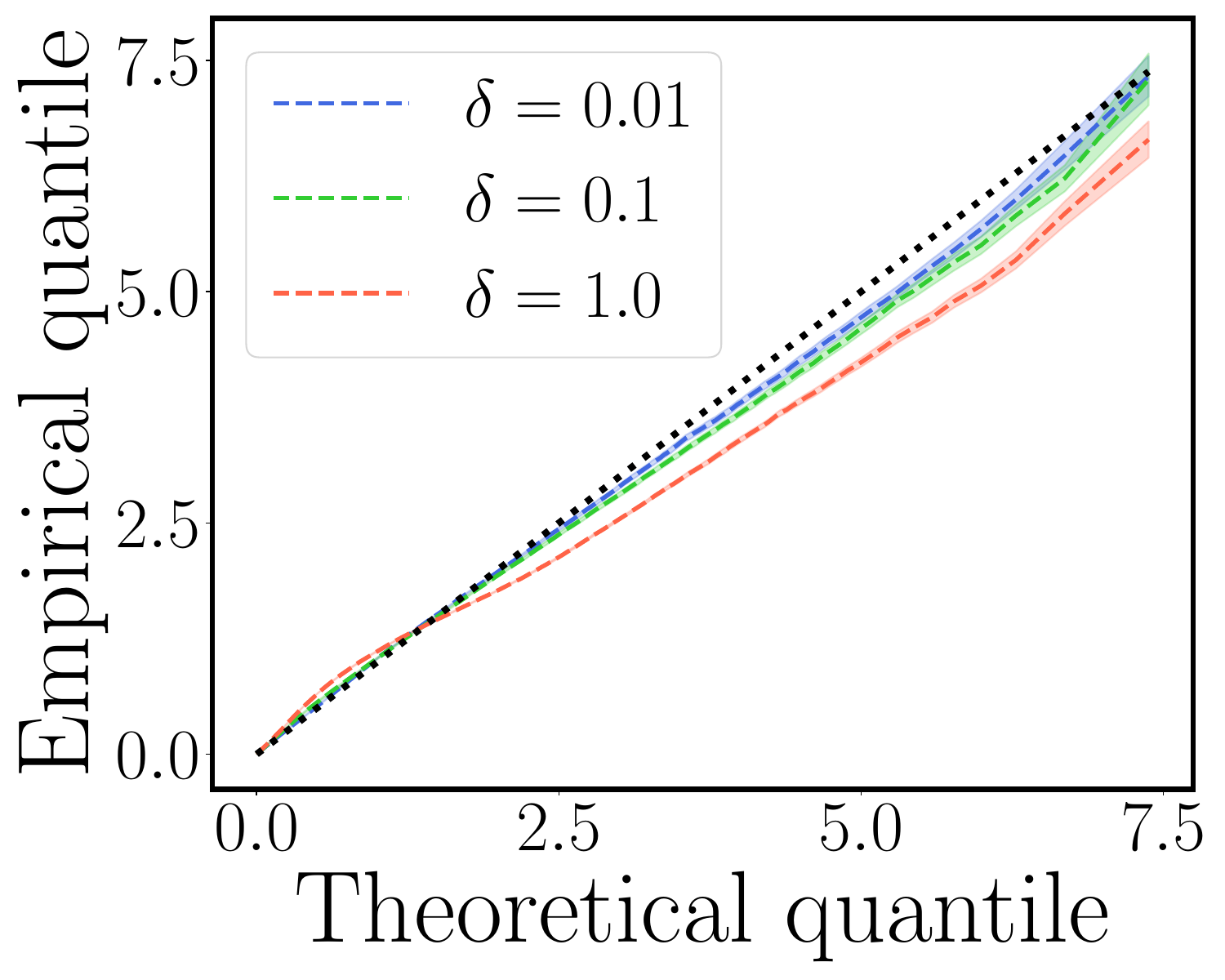}%
    }
    \caption{\textit{Convergence towards the continuous MMHP} --- Box plots of the p-values of the Komogorov-Smirnov test for three values of $\delta$, and average QQ-plots of the simulation experiments with respect to the parameter $\delta$. For the box plots, the dashed line represents the 5\% threshold and whiskers represent 5\% and 95\% quantiles. For the QQ plot, dashed lines represent the average empirical quantiles, and shaded areas represent its 95\% confidence intervals.}
    \label{fig:qq_box_plots_convergence_continuous_kernel}
\end{figure}

\section{Application}
\label{section:application}

In this section, we apply the MMHP-$\delta$ model to cryptocurrency trade data on a centralized exchange (CEX) and show its ability to identify suspicious trading patterns. We compare the model with the Markov-modulated Poisson process (MMPP).

\subsection{Data and intensity seasonality}

The dataset is composed of trades executed for SEI-USD between December 1st, 2023 and March 31st, 2024 on Coinbase. The SEI token was issued on CEXs in 2023 which makes it a great candidate for suspicious trading detection. With every trade comes a timestamp with microsecond precision. This timestamp corresponds to the time at which the matching engine of the market processed the transaction. Transactions are aggregated according to their order identification number allowing to recover the full marketable orders. 

It is well known that trade data exhibits non-stationary properties such as intraday seasonality, \textit{i.e.} the intensity of trades vary significantly throughout the day, and a part of these variations are predictable (traditional market opening and closing times). Applying the MMHP model on such non-stationary data could lead to inaccurate estimations of the parameters (even though previous empirical results of Hawkes processes suggest that fits are rather robust with respect to seasonality in financial timeseries \citep{rambaldi2017role,fabre2024neural}). Hence, we apply the following procedure. We compute the trading intensity at time $t$, denoted by $\bar{\Lambda}(t)$, over 5-minutes bins each day. Then, we compute the average of these intensities per bin over the time period, with a 95\% confidence interval. We also compute the daily trading intensity each day, and its average over the time period per weekday, with a 95\% confidence interval. The results are displayed in Figure \ref{fig:intensity_seasonality}. We observe an intraday seasonality, and a weekday seasonality with less trading intensity during the weekend. The intensity appears to be reasonably stable from 05:00 to 12:00 UTC, exhibiting smaller variance and no apparent drift. We therefore restrict our sample by focusing on the time interval 5:00-12:00 UTC every day. Furthermore, in order to account for weekday seasonality, when we estimate the model on a given weekday, we test the model on the corresponding day of the following week. For example, the first model is estimated on December 1st, 2023, and then tested on December 8th, 2023.

\begin{figure}[!h]
    \centering
    \subfloat[Weekday seasonality]{%
        \includegraphics[width=0.45\linewidth]{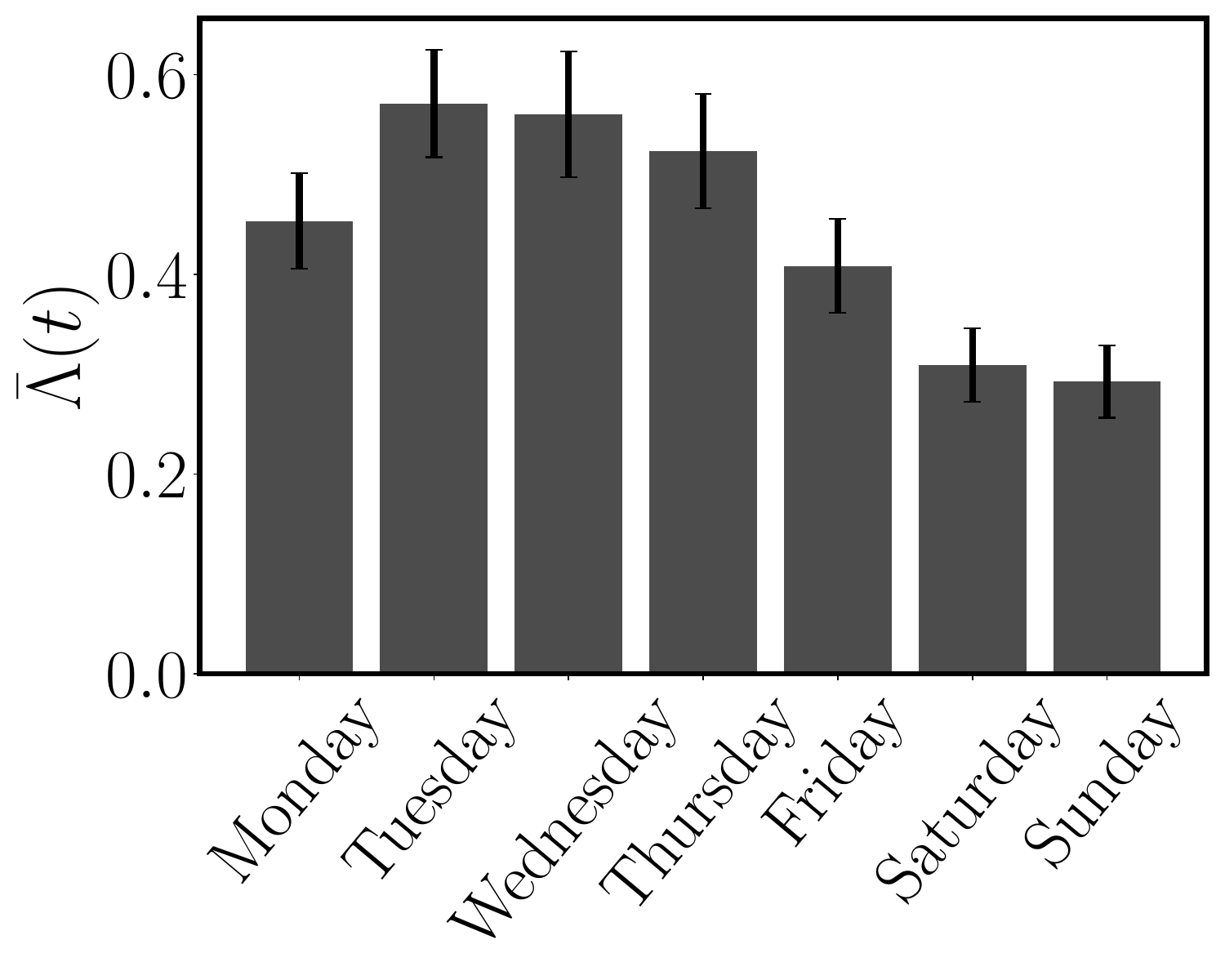}%
    }~~~
    \subfloat[Intraday seasonality]{%
        \includegraphics[width=0.45\linewidth]{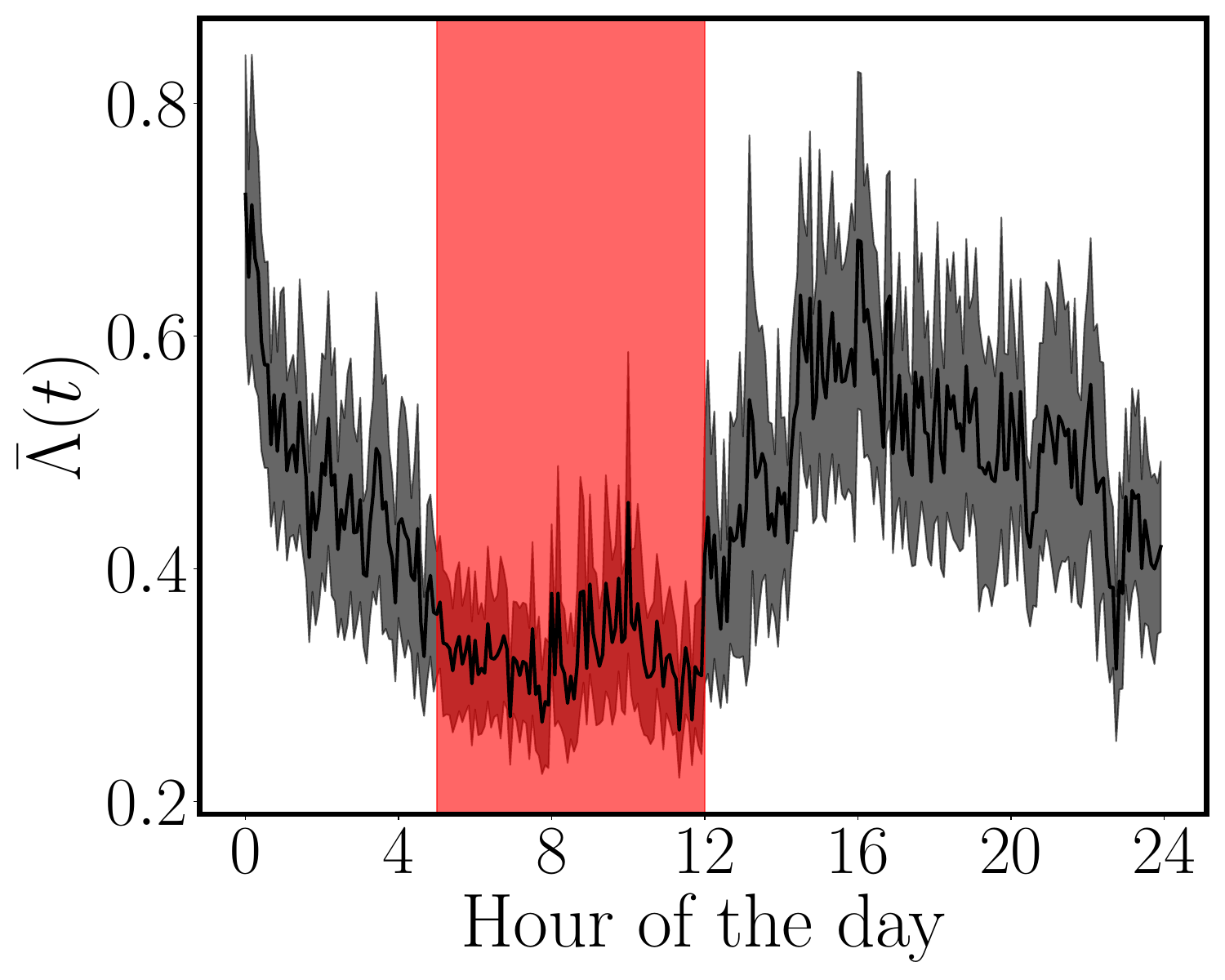}%
    }
    \caption{\textit{Intraday seasonality} --- Intensity of marketable orders expressed in number of events per second, computed over the day for the weekday seasonality, and over bins of 5 minutes for the intraday seasonality. The intensity is averaged over the 4 months. The error bar of the weekday seasonality plot and the dark shade of the intraday seasonality plot both represent the 95\% confidence interval. The red shaded aread represents the time period from 5:00 to 12:00 that is used in the empirical experiment.}
    \label{fig:intensity_seasonality}
\end{figure}

\subsection{Detection of suspicious traded volume}

\subsubsection{Characterization of high-frequency wash trading}

The Committee of European Securities Regulators (CESR) defines wash trading as ``the practice of entering into arrangements for the sale or purchase of
an financial instrument where there is no change in beneficial interests or market risk or where the transfer of beneficial interest or market risk is only between parties who are acting in concert or collusion''\footnote{\url{https://www.esma.europa.eu/esmas-activities/markets-and-infrastructure/market-integrity}}. From this definition, wash trades can either be executed by a single trading account or by multiple trading accounts. In practice, most of the large centralized exchanges have implemented a self-trade prevention (STP) mechanism in their matching engine. By enforcing STP, venues prevent agents to trade against themselves, mitigating the risk of fake volume generation. Nevertheless, a trader can still create many accounts and execute wash trades without being suspected.

The desired effect of wash trading activity is a substantial increase of traded volume, appearing as information to other agents, thus triggering \textit{bona fide} executions. Note that centralized exchanges can also profit from this market abuse since a venue's ranking greatly depends on its reported traded volume. 

From a limit order book perspective, wash trading may either be materialized by the execution of one large trade against one large limit order, or by the execution of multiple trades against one or multiple limit orders. The former tactic generates trading patterns that are not difficult to spot: one may for example monitor the trade sizes and investigate large trades that consume entirely the first limit but do not trade-through. However, the latter wash trading method splits one large trade into many smaller ones with different sizes, making it much more trickier to identify. Furthermore, the second tactic not only manipulates the traded volume, but also the arrival rate of transactions. Since this rate is often a key feature of price formation models used by trading engines, the second method may therefore trigger manipulated order flows. We therefore focus on this tactic in the remainder of the section.

\subsubsection{Suspicious trading detection with the MMHP-$\delta$ model}

Let us denote by $K^b$ and $K^a$ the number of bid and ask trades and let $(P_n^{b, a})_{1\leq n\leq K^{b,a}}$ be the sequence of prices of trades hitting the bid side --- indexed by b --- or the ask side --- indexed by a --- and let $(t_n^{b, a})_{1\leq n\leq K^{b,a}}$ be the corresponding sequence of arrival times. When the execution of an order leads to many trades, \textit{e.g.} a trade-through, the trade price is the one associated to the first transaction.

We apply the MMHP-$\delta$ model to the detection of extreme bursts of transactions that occur on the same side of the limit order book. We apply a filter to the trades that are used in the estimation. Since the transactions that compose a wash trading tactic are executed against the same price limit, we restrict the analysis to trades with zero price return, \textit{i.e.} we consider the sub-sample of arrival times $\{\tau_k^{b,a}\}:=\{t_n^{b, a}:\,P_n^{b, a}-P_{n-1}^{b, a}=0\}$.

We estimate the parameters of a MMHP-$\delta$ with exponential kernel on the observations $\{\tau_k\}$ on each trading day from December 1st, 2023 to March 24th, 2024. A goodness-of-fit procedure is applied to each model on the same weekday of the following week, from December 8th, 2023 to March 31st, 2024. We use 3 different values of $\delta$, which are $\{100, 10, 1\}$ seconds. Note that a MMHP-$\delta$ with $\delta = 100$s could be considered as a proxy of the MMHPSD process of \cite{wang2010statistical}. For each parameter $\delta$, two models are estimated separately, one for trades on the bid side (sell trades), and one for trades hitting the ask side (buy trades). 

Once the time step $\delta$ is set, one needs to specify a number of regimes $M$. A common practice is to estimate the model with different values of $M$, and then select the model that minimizes an information criterion such as the Akaike information criterion (AIC) or the Bayesian information criterion (BIC), see \citet{psaradakis2003determination, krolzig2013markov}. Hence, we estimate the models for $M\in\{2, 3, 4\}$ on each day, compute the corresponding AICs, and rank the models from 1 to 12 (best to worse) on each day. We report the median rank of each model in Table \ref{table:median_aic}. Firstly, the MMHP-$\delta$ demonstrates a superior fit compared to the MMPP. Furthermore, it is evident that increasing the number of regimes enhances the quality of fit as measured by the information criterion. Additionally, for each regime configuration, a reduction in the discretization step $\delta$ appears to improve the criterion. These findings suggest evidence that the MMHP-$\delta$ model with stepwise decreasing kernel provides a better representation of the data than the MMHPSD model proposed by \citet{wang2010statistical}. Although the results indicate that the MMHP-$\delta$ models with four regimes and small $\delta$ minimize the AIC criterion, we stress that going further in terms of the number of regimes was not possible due to limited computational resources. As an example, for the estimation of the MMHP-$\delta$ over a single day, the computation time is approximately 15 minutes with $\delta=100$s and around ten hours with $\delta=1$s. Indeed, with exponential kernels, the computation of the likelihood has a time complexity that is of the order of $\mathcal{O}(N+T\delta^{-1})$, where $N$ is the sample size, and $T$ is the size of the period in calendar time.

\begin{table}[!h]
   \small
   \centering
   \caption{\textit{Selection of the discretization step $\delta$ and the number of regimes $M$} - Median rank of each model among the 12 models, with respect to the AIC criterion.}
   \begin{tabular}{c|cccc}
   \toprule
   M & $\delta=1$ & $\delta=10$ & $\delta=100$ & MMPP \\ 
   \midrule
   2 & \textbf{9} & 10 & 11 & 12 \\
   3 & \textbf{4} & 5 & 6 & 8 \\
   4 & 3 & \textbf{2} & 4 & 5 \\
   \bottomrule
   \end{tabular}
   \label{table:median_aic}
\end{table}

In the rest of this section, we set the number of states to $M=3$, which represents a trade-off between flexibility to capture extreme bursts and computational time. The states are sorted from 1 to 3 with respect to the excitation parameter $\alpha$ of the exponential kernel form $\phi(t)=\alpha e^{-\beta t}$. The larger the hidden state, the larger this parameter. As a benchmark, an MMPP is estimated with $M=3$ regimes. 

We first provide a goodness-of-fit evaluation by computing the transformed residuals of the out-of-sample durations from December 8th, 2023 to March 31st, 2024. In Figure \ref{fig:qq_plot_benchmark_wash_trading}, we show the median QQ plot with its 95\% confidence interval for the MMPP and all MMHP-$\delta$. The median is computed on the daily out-of-sample QQ plots, and the confidence intervals are estimated using bootstrapping. We observe the MMHP-$\delta$ with small $\delta$ (1 or 10) provides better results than the MMPP in terms of out-of-sample goodness-of-fit. Moreover, as expected, setting a large time step $\delta$ seems to worsen the QQ plot. In the rest of the numerical experiments, the results that involve the MMHP-$\delta$ are obtained with $\delta=1$s. 

\begin{figure}
    \centering
    \subfloat[Bid side]{%
        \includegraphics[width=0.45\linewidth]{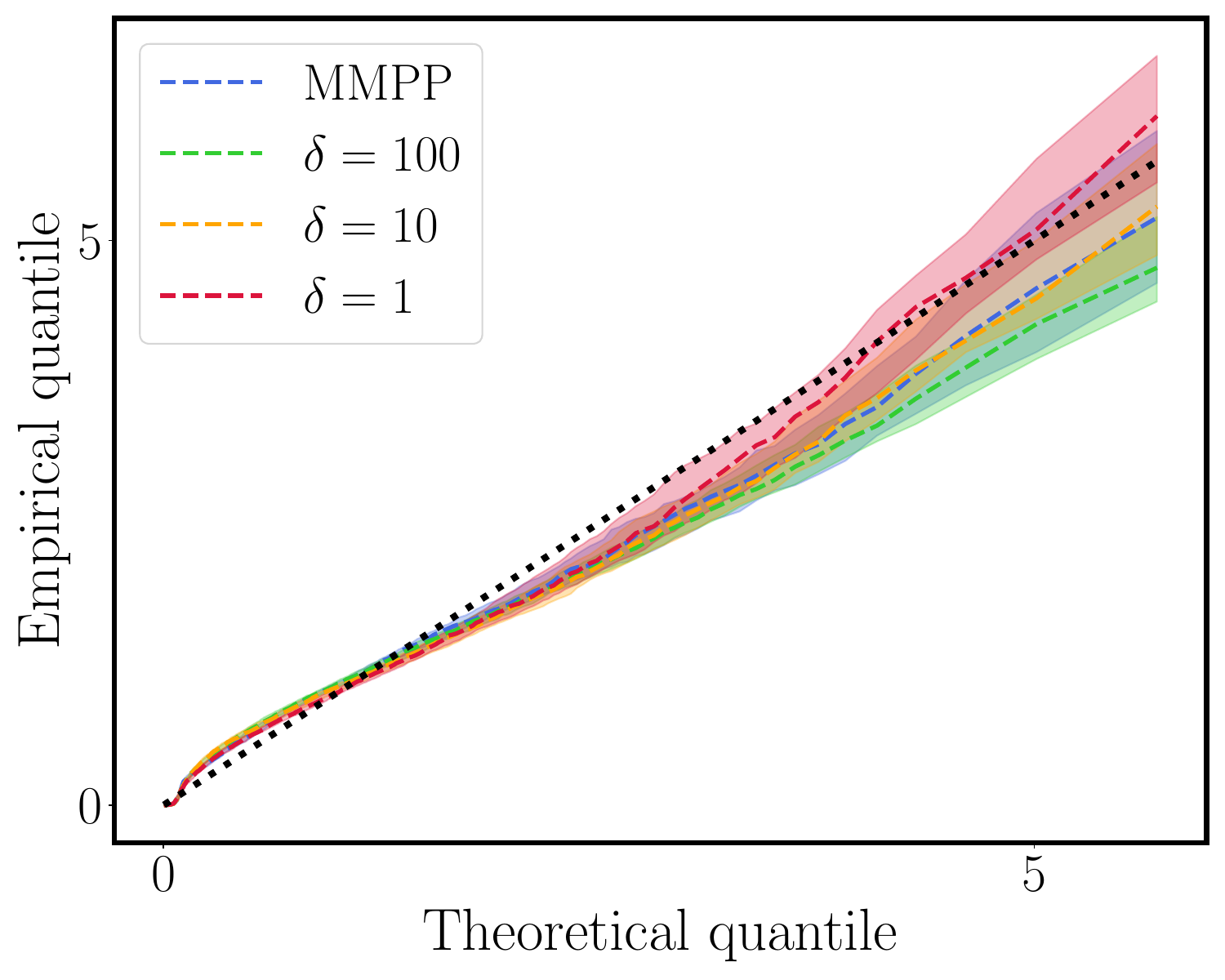}%
    }~~~
    \subfloat[Ask side]{%
        \includegraphics[width=0.45\linewidth]{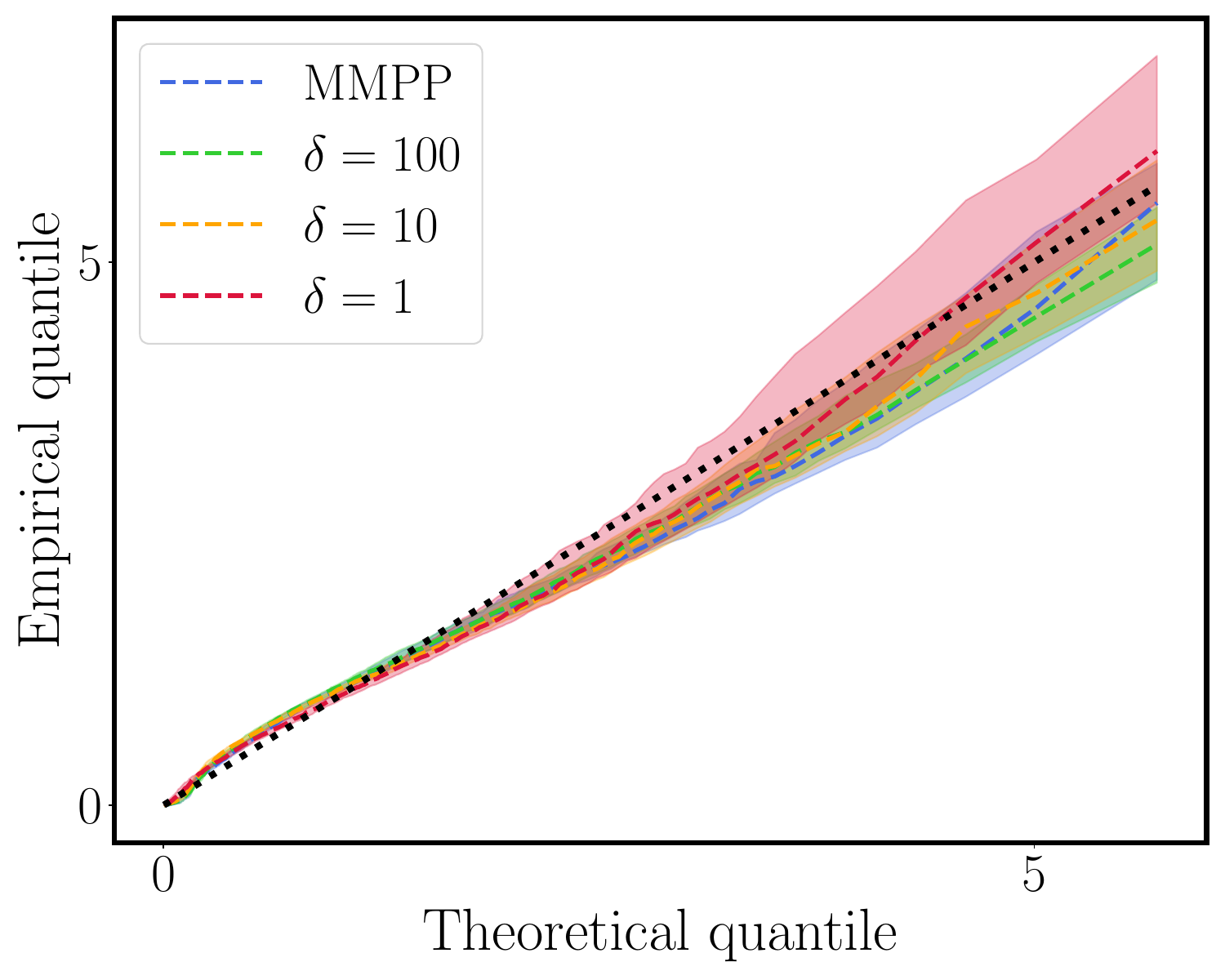}%
    }
    \caption{\textit{Goodness-of-fit} --- Median QQ plots and their 95\% confidence intervals, computed on a daily basis over the out-of-sample data from December 8th, 2023 to March 31st, 2024. Confidence intervals are computed via bootstrapping.}
    \label{fig:qq_plot_benchmark_wash_trading}
\end{figure}

Once the MMHP-$\delta$ model is calibrated, we run the Viterbi algorithm --- described in Algorithm \ref{algo:viterbi} --- to get the sequence of states from December 8th, 2023 to March 31st, 2024. In our setting, trades classified in state 3 are natural candidates for suspicious trading activity. The estimated stationary distributions indicate that the MMHP-$\delta$ model is 90.31\% of the time in state 1 (normal activity), 9.55\% of the time in state 2 (high activity), and 0.14\% of the time in state 3 (extreme bursts) --- here, note that time refers to calendar time. In terms of volume traded, suspicious trading represents 24.20\% of the buy volume and 21.39\% of the sell volume during the period, totaling 216M USD. For comparison, these figures are 25\% of the buy volume and 24.30\% of the sell volume for the MMPP model, totaling 233M USD. Hence, the introduction of a Hawkes component not only fits the data better but also makes the detection more conservative. As such, the MMPP could lead to more false positives than the MMHP-$\delta$ for the same number of regimes.

We show a typical example of suspicious trading activity that is spotted by the MMHP-$\delta$ on the sell trades, and we display an overview of the market at this specific time in Figure \ref{fig:bid_wash_trading_case_26122023}.
\begin{figure}
    \centering
    \subfloat[BBOs' dynamics and trade prices]{%
        \includegraphics[width=0.5\linewidth]{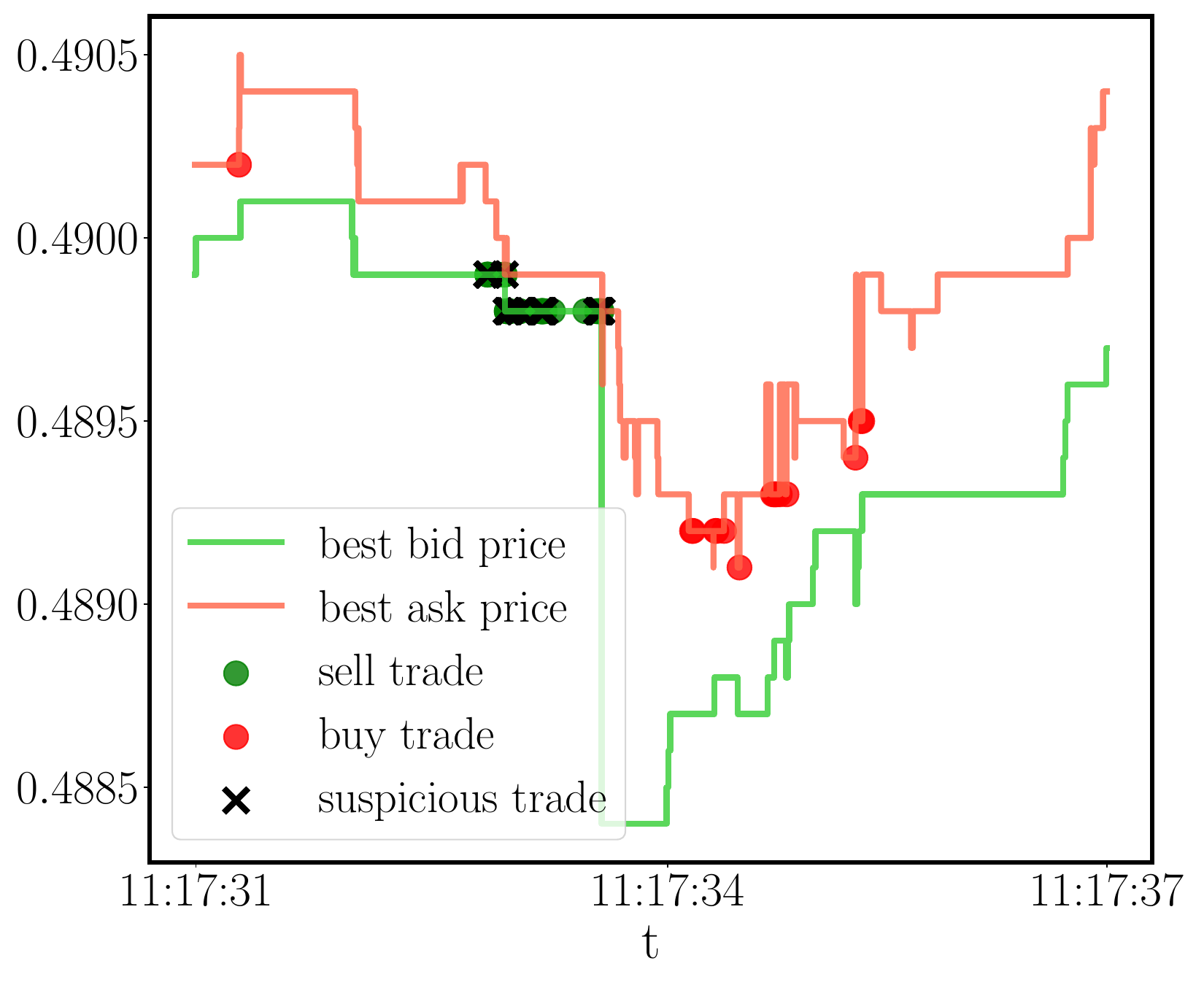}%
    }\hfill
    \subfloat[Trade sizes]{%
        \includegraphics[width=0.5\linewidth]{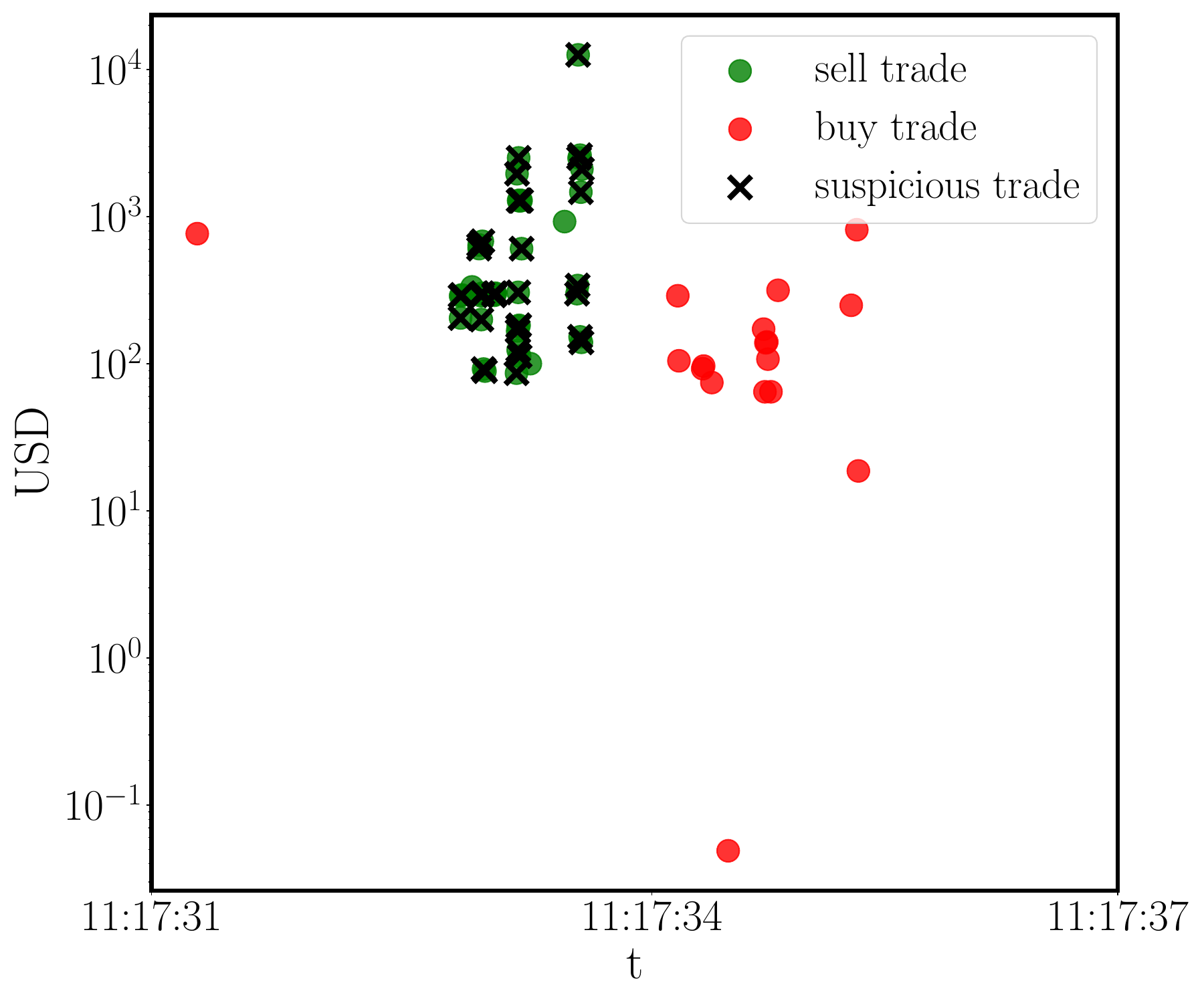}%
    }
    \caption{\textit{Suspicious trading activity} --- An example of ``suspicious'' behavior identified by the bid MMHP-$\delta$ as the extreme burst regime (state 3), anomalies detected on December 26th, 2023 UTC time zone.}
    \label{fig:bid_wash_trading_case_26122023}
\end{figure}
We observe that the model is able to detect sudden bursts of trades. Looking at the trade sizes reveals that tens to hundreds of orders are executed against the same price within a very small time frame, \textit{i.e.} of the order of milliseconds, providing the empirical evidence of a high-frequency activity. Even more interestingly, a significant price bump is observed following the trade burst. Such patterns indicate that this asset could be highly sensitive to fake volume generation. 

We provide more examples in the appendix. Figures \ref{fig:bid_wash_trading_case_06032024}, \ref{fig:ask_wash_trading_case_15012024}, and \ref{fig:ask_wash_trading_case_03032024} show how the joint use of a MMHP-$\delta$ model and the Viterbi algorithm detects suspicious trading activity and potential price manipulation in buy and sell trades.

\subsubsection{Characterization of market states}

Beyond suspicious trading detection, it is very interesting that we are able to characterize the market states identified by the MMHP-$\delta$ model. We now proceed to two experiments to that end. The first experiment analyzes the liquidity imbalance signal in the three market states, while the second one computes a price response function upon regime change.

\paragraph{Liquidity imbalance} We measure the liquidity imbalance at the best queues before a regime change. Denote by $q_k^b$ and by $q_k^a$ the volumes at the best bid queue and at the best ask queue at time $(\tau_k^b)^-$, \textit{i.e.} right before the $k$th event time. We define the liquidity imbalance measure

\begin{equation}
    \mathcal{I}_k:=\frac{q_k^b-q_k^a}{q_k^b+q_k^a}.
\end{equation}

We compute this measure at each state transition, \textit{i.e.} when a trade occurs without moving the price and the Markov chain switches to state $i$, $1\leq i\leq M$, we record the value of the liquidity imbalance $\mathcal{I}$ right before this event. A well-known stylised facts of market microstructure is the dependence of the arrival rate of trades on the best queue imbalance. In practice, one finds that the intensity of sell trades is a decreasing function of the imbalance (and an increasing function for the intensity of buy trades). But this measure does not distinguish between market regimes and heterogeneity of agents as it is computed over all data points.

\paragraph{Price response} In the spirit of \cite{bouchaud2003fluctuations}, we define a price response to the arrival of trades but we adapt the framework to a regime-specific price response. Let $p_n$ be the $n$th mid-price change, $\varepsilon_n$ the sign of the last passive trade that occurred before the $n$th mid-price change (with event time that is an element of $\{\tau_k^{b,a}\}$), and $S_n$ the state of the hidden Markov chain for this trade. We define the mid-price response function as

\begin{equation}
    \mathcal{R}(h,s):=\mathbb{E}\left(\varepsilon_n\left(\frac{p_{n+h}}{p_n}-1\right)\,\bigg|\,S_n=s,S_{n-1}\neq s\right), \hfill 1\leq h\leq H, \hfill 1\leq s\leq M.
\end{equation}

The response $\mathcal{R}(h,s)$ measures the average positive return conditioned on a trade at time 0, $h$ mid-price moves later, and conditioned to a regime transition to state $s$ at time 0.

The results are displayed in Figure \ref{fig:state_characterization_mmhp_wash}. For completeness, we also make the same experiment using the entire trading day, without accounting for intraday seasonality, and report the results in Appendix, in Figure \ref{fig:state_characterization_mmhp_wash_full_period}. In this case, we train the MMHP-$\delta$ on December 1st, 2023 and compute the metrics on the rest of the period.
\begin{figure}
    \centering
    \subfloat[Imbalance distribution before state transition, bid]{%
        \includegraphics[width=0.33\linewidth]{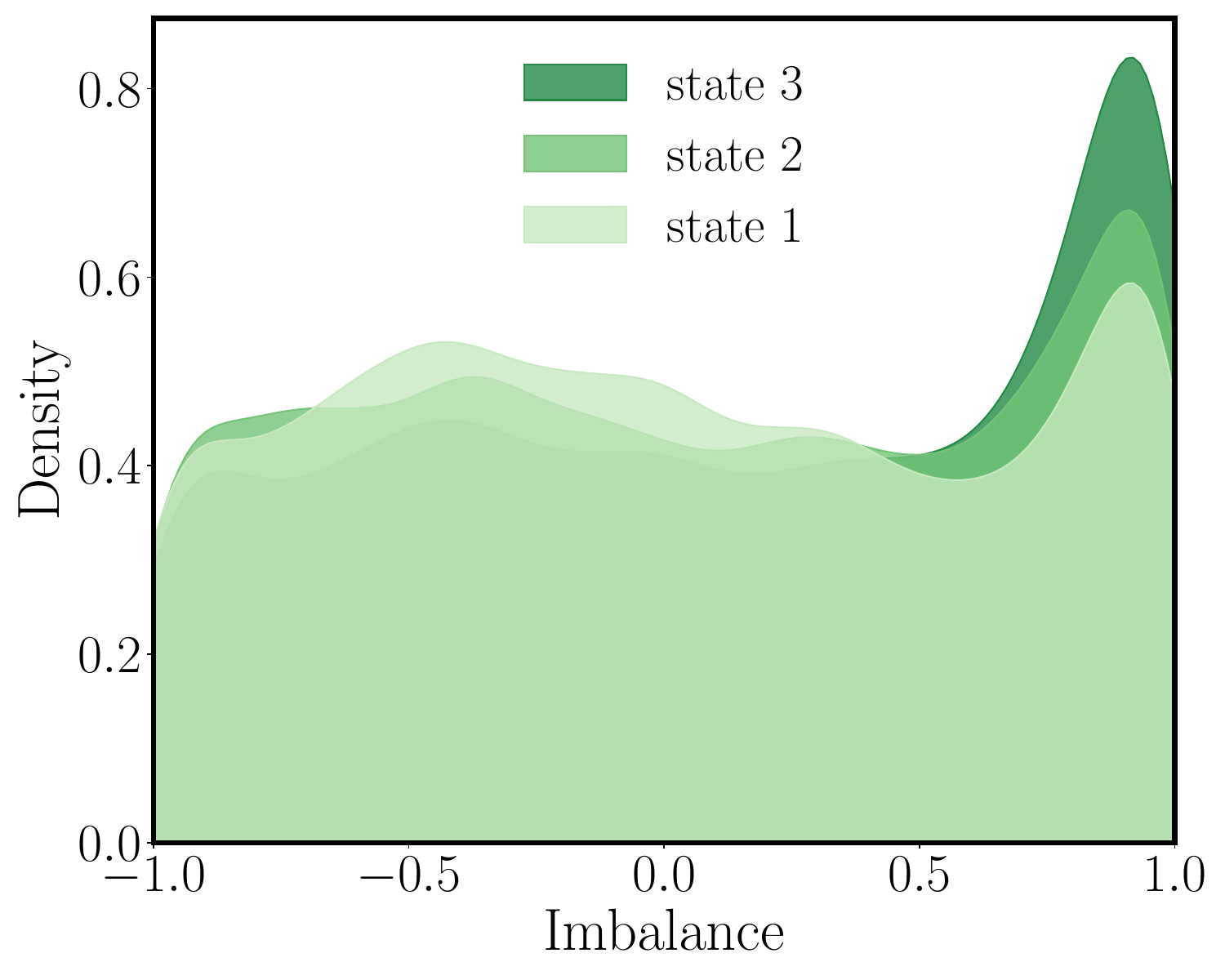}%
    }~~~
    \subfloat[Imbalance distribution before state transition, ask]{%
        \includegraphics[width=0.33\linewidth]{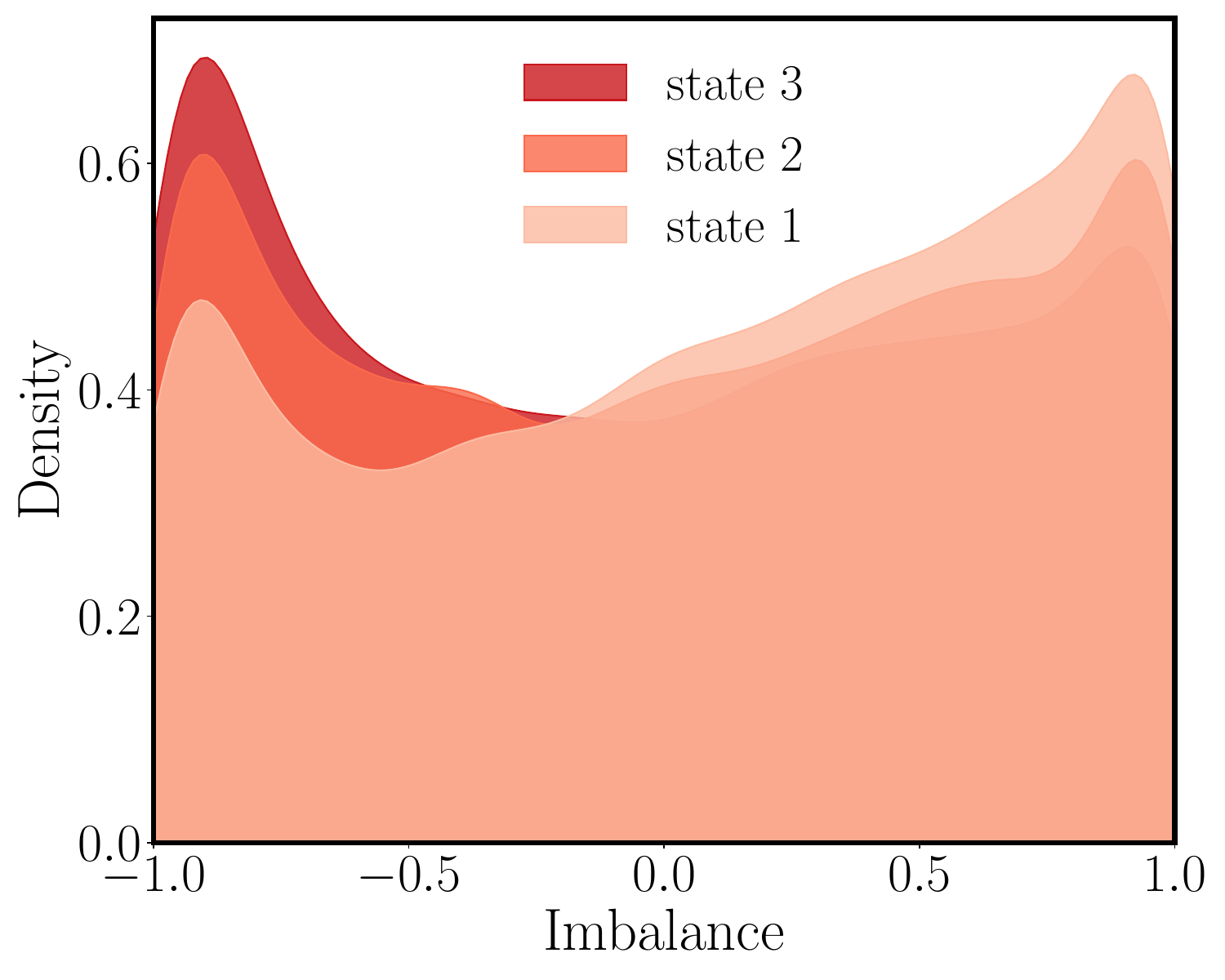}%
    }~~~
    \subfloat[Mid price response function after state transition with 95\% confidence intervals]{%
        \includegraphics[width=0.33\linewidth]{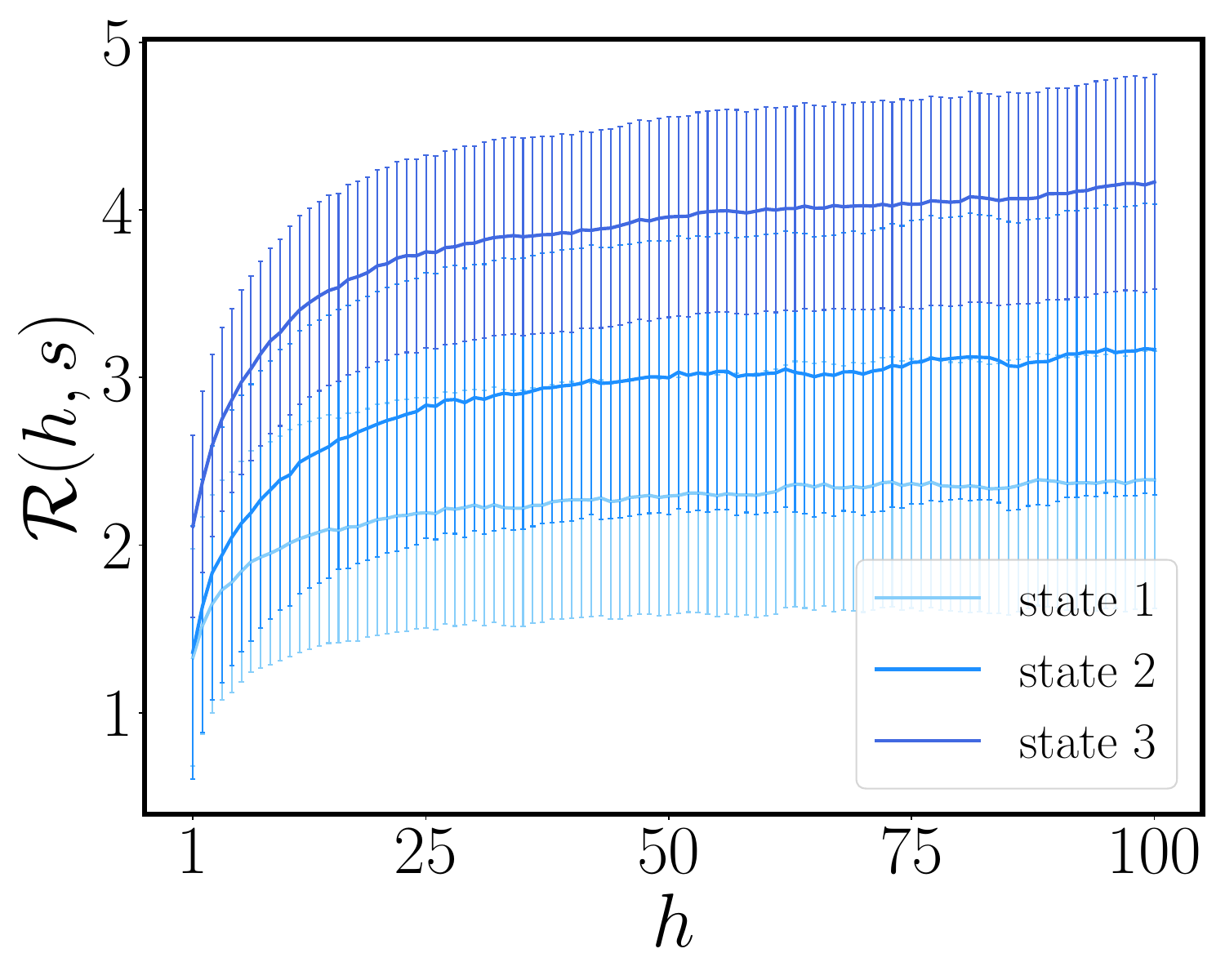}%
    }
    \caption{\textit{Suspicious trading activity} --- Market characterization of the trading activity regimes identified by the MMHP-$\delta$ model. The price response is expressed in basis points and the horizon in number of mid price moves.}
    \label{fig:state_characterization_mmhp_wash}
\end{figure}
Strikingly, we observe that the extreme burst regime occurs with an imbalance distribution that is inversely skewed compared to what we usually measure. For example, for sell trades, the classical result is that the imbalance distribution before the arrival of the trade is left-skewed, which is what we observe for the ``normal'' regime, \textit{i.e.} state 1. As we get closer to the burst regime, the distribution skews on the opposite side, revealing an atypical behavior for such trades. It shows that when the market is in the third state, even though the imbalance is large and an upward price move or the arrival of a buy trade is expected to happen, the opposite happens, causing the well-known imbalance-intensity relationship to break. Concerning the mid-price response to a state transition, we see that the arrival of burst regimes is, on average, followed by significant price moves even though the burst trades do not consume the entire queue. A possible interpretation is that market makers cancel orders that are pending deep in the book, creating large liquidity gaps before the first queue is either filled or canceled. This provides evidence of a significant price impact of the burst regime, demonstrating that a manipulative agent could use this microstructure property to manipulate the price by sending many orders to execute against herself.

\section{Conclusion}

This paper presents a novel approach to market manipulation detection in cryptocurrency markets through the use of regime-switching point processes. We introduced a point process with a Hawkes-like intensity whose parameters' dynamics are driven by a continuous time Markov chain. In this model, ``Hawkes-like'' stands for the piecewise constant kernel that is defined, with a discretization step that is kept constant between events. This specification trick enabled us to develop an estimation procedure, built upon an Expectation-Maximization (EM) algorithm. We demonstrated the convergence of the process to its continuous version as the discretization step becomes small, and we proved the robustness of the method using simulated data. Concerning the model, an important point needs to be discussed and highlighted. By introducing a kernel that is piecewise constant between events, and not over the entire time frame, we diverge from traditional Hawkes processes. This detail of specification makes the theory of Hawkes processes --- \textit{e.g.} the computation of endogeneity ratios or the stationarity condition --- unusable and an exploration of the new process's theoretical properties would be of interest. 

Applying the model to real cryptocurrency trade data, we effectively demonstrated its ability to identify extreme bursts of trades, underscoring its practical utility in real-world market scenarios. Further analysis enabled us to characterize the states of the Markov chain in relation to the limit order book, providing deeper insights into market dynamics. Notably, if these identified bursts are not due to manipulative behavior, it suggests that the imbalance property, often referred to as the "worst-kept secret of high-frequency traders," is highly sensitive to market regimes. Our methodology showed promising results in detecting fake volume and pinging activities, highlighting its potential for broader market surveillance applications. Future research should aim to refine the theoretical foundations of this process and explore additional applications across various trading environments to enhance its robustness and applicability. Given the absence of anonymous trading agent IDs, our detection method falls into the category of indirect detection methodologies and is therefore prone to false positives. To mitigate this issue, complementary detection rules should be incorporated into the detection algorithm.

\bibliography{main}
\bibliographystyle{apalike}

\appendix

\section{Proofs}

\label{section:proofs}

\subsection{Proof of Proposition~\ref{prop:forward_ode}}
\begin{quote}\it
    The forward transition matrix is the solution of the system of ordinary differential equations
    \begin{equation}
        \frac{\mathrm{d}H^{(n)}}{\mathrm{d}u}(u)=H^{(n)}(u)\left(Q-\Lambda_{t_{n-1}+u}\right)\notag,\hspace{0.3cm} u\in]0, x_n],
    \end{equation}
    with initial condition $H^{(n)}(0)=\mathbb{I}_M$.
\end{quote}
\begin{proof}
    Let $1\leq i, j\leq M$, $u,\,\Delta u\,> 0$. With elementary probability calculus and by using the Markov property of $S$, we get
    \begin{align}
        H_{ij}^{(n)}(u+\Delta u)&=\mathbb{P}\left(S_{t_{n-1}+u+\Delta u}=j, N_{t_{n-1}+u+\Delta u}-N_{t_{n-1}}=0\,|\,S_{t_{n-1}}=i,\mathcal{F}_{t_{n-1}}\right)\notag\\
        &\hspace{-2.2cm}=\sum_{k=1}^M\mathbb{P}\left(S_{t_{n-1}+u+\Delta u}=j,S_{t_{n-1}+u}=k,N_{t_{n-1}+u+\Delta u}-N_{t_{n-1}}=0\,|\,S_{t_{n-1}}=i,\mathcal{F}_{t_{n-1}}\right)\notag\\
        &\hspace{-2.2cm}=\sum_{k=1}^M\mathbb{P}\big(S_{t_{n-1}+u+\Delta u}=j,S_{t_{n-1}+u}=k,N_{t_{n-1}+u+\Delta u}-N_{t_{n-1}+u}=0, \notag
        \\ & N_{t_{n-1}+u}-N_{t_{n-1}}=0\,|\, S_{t_{n-1}}=i,\mathcal{F}_{t_{n-1}}\big)\notag\\
        &\hspace{-2.2cm}=\sum_{k=1}^M\mathbb{P}\left(S_{t_{n-1}+u}=k,N_{t_{n-1}+u}-N_{t_{n-1}}=0\,|\,S_{t_{n-1}}=i,\mathcal{F}_{t_{n-1}}\right)\times\,\mathbb{P}\big(S_{t_{n-1}+u+\Delta u}=j,\notag\\
        & N_{t_{n-1}+u+\Delta u}-N_{t_{n-1}+u}=0\,|\,S_{t_{n-1}+u}=k,S_{t_{n-1}}=i,N_{t_{n-1}+u}-N_{t_{n-1}}=0,\mathcal{F}_{t_{n-1}}\big)\notag\\
        &\hspace{-2.2cm}=\sum_{k=1}^MH_{ik}^{(n)}(u)\mathbb{P}\big(S_{t_{n-1}+u+\Delta u}=j,N_{t_{n-1}+u+\Delta u}-N_{t_{n-1}+u}=0\notag\\
        &\hspace{-2.2cm}\hspace{3.5cm}|\,S_{t_{n-1}+u}=k,S_{t_{n-1}}=i,N_{t_{n-1}+u}-N_{t_{n-1}}=0,\mathcal{F}_{t_{n-1}}\big)\notag\\
        &\hspace{-2.2cm}=\sum_{k=1}^MH_{ik}^{(n)}(u)\mathbb{P}\big(S_{t_{n-1}+u+\Delta u}=j,N_{t_{n-1}+u+\Delta u}-N_{t_{n-1}+u}=0\,|\,S_{t_{n-1}+u}=k,\notag
        \\ & N_{t_{n-1}+u}-N_{t_{n-1}}=0,\mathcal{F}_{t_{n-1}}\big)\notag\\
        &\hspace{-2.2cm}=\sum_{k=1}^MH_{ik}^{(n)}(u)\mathbb{P}\big(S_{t_{n-1}+u+\Delta u}=j\,|\,S_{t_{n-1}+u}=k,N_{t_{n-1}+u}-N_{t_{n-1}}=0,\mathcal{F}_{t_{n-1}}\big)\notag\\
        &\hspace{-2.2cm}\hspace{1cm}\times\mathbb{P}\big(N_{t_{n-1}+u+\Delta u}-N_{t_{n-1}+u}=0\,|\,S_{t_{n-1}+u+\Delta u}=j,S_{t_{n-1}+u}=k,N_{t_{n-1}+u}-N_{t_{n-1}}=0,\mathcal{F}_{t_{n-1}}\big)\notag\\
        &\hspace{-2.2cm}=\sum_{k=1}^MH_{ik}^{(n)}(u)\mathbb{P}\big(S_{t_{n-1}+u+\Delta u}=j\,|\, S_{t_{n-1}+u}=k\big)\notag\\
        &\hspace{-2.2cm}\hspace{1cm}\times\mathbb{P}\big(N_{t_{n-1}+u+\Delta u}-N_{t_{n-1}+u}=0\,|\,S_{t_{n-1}+u+\Delta u}=j,S_{t_{n-1}+u}=k,N_{t_{n-1}+u}-N_{t_{n-1}}=0,\mathcal{F}_{t_{n-1}}\big)\notag.
    \end{align}
    By applying a Taylor series expansion with $\Delta u\to0$, we get
    \begin{align}
        \hspace{-2cm}H_{ij}^{(n)}(u+\Delta u)&=\sum_{k=1}^MH_{ik}^{(n)}(u)\left(\delta_{kj}+q_{kj}\Delta u+o(\Delta u)\right)\left(1-\lambda_{t_{n-1}+u}^k\Delta u+o(\Delta u)\right)\notag\\
        &=\sum_{k=1}^MH_{ik}^{(n)}(u)\left\{\delta_{kj}+\left(q_{kj}-\delta_{kj}\lambda_{t_{n-1}+u}^k\right)\Delta u+o(\Delta u)\right\},\label{eq:recurrence_H}
    \end{align}
    where
    \begin{equation}
        \lambda_{t_{n-1}+u}^k=\mu^k+\sum_{t_l\leq t_{n-1}}\phi^k(t_{n-1}-t_l+u)\notag.
    \end{equation}
    Using Equation \eqref{eq:recurrence_H},
    \begin{align}
        \frac{1}{\Delta u}\left(H_{ij}^{(n)}(u+\Delta u)-H_{ij}^{(n)}(u)\right)&=\frac{1}{\Delta u}\sum_{k\neq j}H_{ik}^{(n)}(u)\big(q_{kj}\Delta u+o(\Delta u)\big)\notag\\
        &\hspace{0.5cm}-\frac{1}{\Delta u}H_{ij}^{(n)}(u)\big((q_j+\lambda_{t_{n-1}+u}^j)\Delta u+o(\Delta u)\big)\notag,
    \end{align}
    which leads to
    \begin{equation}
        \lim_{\Delta u \to 0^+}\frac{1}{\Delta u}\left(H_{ij}^{(n)}(u+\Delta u)-H_{ij}^{(n)}(u)\right)\,=\,-H_{ij}^{(n)}(u)\left(q_j+\lambda_{t_{n-1}+u}^j\right)+\sum_{k\neq j}H_{ik}^{(n)}(u)q_{kj}\notag.
    \end{equation}
    Since $H_{ij}^{(n)}(0)=\delta_{ij}$, $1\leq i,j \leq M$ we obtain the desired result.
\end{proof}

\subsection{Proof of Lemma~\ref{lemma:forward_backward_equation}}
\begin{quote}\it
    The forward and backward transition matrices follow the equality
    \begin{equation}
        H^{(n)}(t_n-t_{n-1})=H^{(n)}(t-t_{n-1})G^{(n)}(t_n-t)\notag, \hspace{0.3cm}t_{n-1}\leq t\leq t_n.
    \end{equation}
\end{quote}
\begin{proof}
    Let $1\leq i, j\leq M$, $t_{n-1}\leq t\leq t_n$. By using the Markov property of $S$, we have
    \begin{align}
        H_{ij}^{(n)}(t_n^--t_{n-1})&=\mathbb{P}\left(S_{t_n}=j,N_{t_n^-}-N_{t_{n-1}}=0\,|\,S_{t_{n-1}}=i,\mathcal{F}_{t_{n-1}}\right)\notag\\
        &=\mathbb{P}\left(S_{t_n}=j,N_{t_n^-}-N_t=0,N_t-N_{t_{n-1}}=0\,|\,S_{t_{n-1}}=i,\mathcal{F}_{t_{n-1}}\right)\notag\\
        &=\sum_{k=1}^M\mathbb{P}\left(S_{t_n}=j,S_t=k,N_{t_n^-}-N_t=0,N_t-N_{t_{n-1}}=0\,|\,S_{t_{n-1}}=i,\mathcal{F}_{t_{n-1}}\right)\notag\\
        &=\sum_{k=1}^M\mathbb{P}\left(S_t=k,N_t-N_{t_{n-1}}=0\,|\,S_{t_{n-1}}=i,\mathcal{F}_{t_{n-1}}\right)\notag\\
        &\hspace{1.5cm}\times\mathbb{P}\left(S_{t_n}=j,N_{t_n^-}-N_t=0\,|\,S_t=k,S_{t_{n-1}}=i,N_t-N_{t_{n-1}}=0,\mathcal{F}_{t_{n-1}}\right)\notag\\
        &=\sum_{k=1}^MH_{ik}^{(n)}(t-t_{n-1})\mathbb{P}\left(S_{t_n}=j,N_{t_n^-}-N_t=0\,|\,S_t=k,N_t-N_{t_{n-1}}=0,\mathcal{F}_{t_{n-1}}\right)\notag\\
        &=\sum_{k=1}^MH_{ik}^{(n)}(t-t_{n-1})\mathbb{P}\left(S_{t_n}=j,N_{t_n^-}-N_t=0\,|\,S_t=k,\mathcal{F}_t\right)\notag\\
        &=\sum_{k=1}^MH_{ik}^{(n)}(t-t_{n-1})G_{kj}^{(n)}(t_n-t)\notag.
    \end{align}
\end{proof}

\subsection{Proof of Proposition~\ref{prop:backward_ode}}
\begin{quote}\it
    The backward transition matrix is the solution of the system of ordinary differential equations
    \begin{equation}
        \frac{\mathrm{d}G^{(n)}}{\mathrm{d}u}(u)=\left(Q-\Lambda_{t_n-u}\right)G^{(n)}(u)\notag,\hspace{0.3cm}u\in]0, x_n],
    \end{equation}
    with initial condition $G^{(n)}(0)=\mathbb{I}_M$.
\end{quote}
\begin{proof}
    Let $0\leq u\leq x_n$ and write the equation of Lemma \ref{lemma:forward_backward_equation} with the change of variable $u:=t_n-t$ 
    \begin{equation}\label{eq:forward_backward_with_u}
        H^{(n)}(t_n-t_{n-1})=H^{(n)}(t_n-t_{n-1}-u)G^{(n)}(u)\notag.
    \end{equation}
    Differentiating both sides of Equation \eqref{eq:forward_backward_with_u} with respect to $u$ leads to
    \begin{equation}
        H^{(n)}(t_n-t_{n-1}-u)\frac{\mathrm{d}G^{(n)}}{\mathrm{d}u}(u)=\frac{\mathrm{d}H^{(n)}}{\mathrm{d}u}(t_n-t_{n-1}-u)G^{(n)}(u)\notag.
    \end{equation}
    We then use Proposition \ref{prop:forward_ode} to obtain
    \begin{align}
        \frac{\mathrm{d}G^{(n)}}{\mathrm{d}u}(u)&=\left[H^{(n)}(t_n-t_{n-1}-u)\right]^{-1}\frac{\mathrm{d}H^{(n)}}{\mathrm{d}u}(t_n-t_{n-1}-u)G^{(n)}(u)\notag\\
        &=\left(Q-\Lambda_{t_n-u}\right)G^{(n)}(u)\notag.
    \end{align}
\end{proof}

\subsection{Proof of Proposition~\ref{prop:forward_backward_formulas}}
\begin{quote}\it
    If the intensity is given by Equation \eqref{eq:piecewise_constant_intensity}, the forward transition matrix is
    \begin{equation}
        H^{(n)}(u)=\Xi^{(n)}_{\ell(t_{n-1}+u)}e^{(Q-\Lambda_{t_{n-1}+\ell(t_{n-1}+u)\delta})(u-\ell(t_{n-1}+u)\delta)},\hspace{0.3cm} 0\leq u\leq x_n,
    \end{equation}
    for which we define the following quantity
    \begin{equation}
        \Xi^{(n)}_k:=\prod_{r=1}^{k}e^{((Q-\Lambda_{t_{n-1}+(r-1)\delta})\delta)}, \hspace{0.3cm} 1\leq n\leq K,\hspace{0.1cm} 1\leq k \leq \ell_n.
    \end{equation}
    
    Similarly, the backward transition matrix is written as, for $u\in[0, x_n]$
    \begin{equation}
        \begin{cases} G^{(n)}(u)=e^{(Q-\Lambda_{t_{n-1}+\ell_n\delta})u}\hspace{0.2cm}\text{if}\;u\leq\Delta_n,\\\\  G^{(n)}(u)=e^{(Q-\Lambda_{t_{n-1}+\ell(t_n-u)\delta})(u-((\ell_n-\ell(t_n-u)-1)\delta+\Delta_n))}\Psi^{(n)}_{\ell_n-\ell(t_n-u)-1}e^{(Q-\Lambda_{t_{n-1}+\ell_n\delta})\Delta_n}\hspace{0.2cm}\text{else},
    \end{cases}
    \end{equation}
    for which we define the following quantity
    \begin{equation}
        \Psi^{(n)}_k:=\prod_{r=1}^{k}e^{(Q-\Lambda_{t_{n-1}+(\ell_n-(k-r)-1)\delta})\delta}, \hspace{0.3cm} 1\leq n\leq K,\hspace{0.1cm} 1\leq k \leq \ell_n.
    \end{equation}
    Here, the notation $e^A$ stands for the matrix exponential of $A$ and we use the convention $\prod_{r=1}^{0}A_r=\mathbb{I}_M$.
\end{quote}
\begin{proof}
    Under the $\delta$-piecewise constant intensity of Equation \eqref{eq:piecewise_constant_intensity}, Equation \eqref{eq:ode_forward} is
    \begin{align}
        \frac{\mathrm{d}H^{(n)}}{\mathrm{d}u}(u)&=H^{(n)}(u)\left(Q-\Lambda_{t_{n-1}+\ell(t_{n-1}+u)\delta}\right),\hspace{0.3cm}0< u\leq x_n,\notag\\
        H^{(n)}(0)&=\mathbb{I}_M.
    \end{align}
    Note that for all $u\in[0, x_n]$, $\Lambda_{t_{n-1}+u}=\Lambda_{t_{n-1}+\ell(t_{n-1}+u)\delta}$ and that $\Lambda_{t_{n-1}+\ell(t_{n-1}+u)\delta}$ is the same whether $\Lambda$ is the MMHP-$\delta$ intensity or its continuous counterpart. We apply a discretization to the domain $]0,x_n]$ with intervals of the form $]k\delta,(k+1)\delta]$, and we get the following set of ODEs
    \begin{equation}
        \frac{\mathrm{d}H^{(n)}}{\mathrm{d}u}(u)=H^{(n)}(u)\left(Q-\Lambda_{t_{n-1}+k\delta}\right),\hspace{0.3cm}u\in]k\delta, (k+1)\delta],\hspace{0.1cm}0\leq k\leq \ell_n.
    \end{equation}
    This system is easily solved iteratively, using the initial condition, for $k=0$, $H^{(n)}(0)=\mathbb{I}_M$ and for $1\leq k\leq \ell_n$,
    \begin{equation}
        H^{(n)}(k\delta)=\prod_{r=1}^ke^{(Q-\Lambda_{t_{n-1}+(r-1)\delta})\delta},
    \end{equation}
    
    leading to the desired expression of the forward transition matrix.

    The same idea is applied to the backward transition matrix. Indeed, Equation \eqref{eq:ode_backward} is
    \begin{align}
        \frac{\mathrm{d}G^{(n)}}{\mathrm{d}u}(u)&=\left(Q-\Lambda_{t_{n-1}+\ell(t_n-u)\delta}\right)G^{(n)}(u),\hspace{0.3cm}0< u\leq x_n,\notag\\
        G^{(n)}(0)&=\mathbb{I}_M.
    \end{align}
    By operating a discretization to the domain $]0,x_n]$ with intervals of the form $]k\delta+\Delta_n,(k+1)\delta+\Delta_n]$ and $]0,\Delta_n]$, we first get the ODE over the residual time to $t_n$
    \begin{align}
        \frac{\mathrm{d}G^{(n)}}{\mathrm{d}u}(u)&=\left(Q-\Lambda_{t_{n-1}+\ell_n\delta}\right)G^{(n)}(u),\hspace{0.3cm}0<u<\Delta_n,\notag\\
        G^{(n)}(0)&=\mathbb{I}_M,
    \end{align}
    and then the following ODEs over the $\delta$-intervals
    \begin{equation}
        \frac{\mathrm{d}G^{(n)}}{\mathrm{d}u}(u)=\left(Q-\Lambda_{t_{n-1}+(\ell_n-k-1)\delta}\right)G^{(n)}(u),\hspace{0.3cm}k\delta+\Delta_n < u \leq (k+1)\delta+\Delta_n,\hspace{0.1cm}0\leq k\leq \ell_n-1.
    \end{equation}
    Again, this system is iteratively solved at each step $k$ and we obtain the initial conditions
    \begin{equation}
        G^{(n)}(k\delta+\Delta_n)=\prod_{r=1}^ke^{(Q-\Lambda_{t_{n-1}+(\ell_n-(k-r)-1)\delta})\delta}e^{(Q-\Lambda_{t_{n-1}+\ell_n\delta})\Delta_n},\hspace{0.1cm}0\leq k\leq \ell_n-1.
    \end{equation}
    We finally get the desired expression of the backward transition matrix.
\end{proof}

\subsection{Proof of Proportion~\ref{prop:complete_likelihood}}

\begin{quote}\it
    The log-likelihood of an MMHP with parameters $\Theta:=\Theta^M\cup\Theta^H$ over the observations $\mathcal{T}$ can be written as
    \begin{equation}
        \mathcal{L}(\Theta,\mathcal{T})=\log \mathcal{L}^M(\Theta^M,\mathcal{T})+\log \mathcal{L}^H(\Theta^H,\mathcal{T})\notag,
    \end{equation}
    where
    \begin{equation}
        \log \mathcal{L}^M(\Theta^M,\mathcal{T})=\sum_{i=1}^M\left\{\mathds{1}_{\{S_0=i\}}\log \xi_0^i-D_iq_i+\sum_{j=1,j\neq i}^Mw_{ij}\log q_{ij}\right\}\notag,
    \end{equation}
    and
    \begin{equation}
        \log \mathcal{L}^H(\Theta^H,\mathcal{T})=\sum_{i=1}^M\left\{\sum_{n=1}^K\mathds{1}_{\{S_{t_n}=i\}}\log \lambda_{t_n}^i(\Theta^H)\,-\,\int_0^T\mathds{1}_{\{S_t=i\}}\lambda_t^i(\Theta^H)\,\mathrm{d}t\right\}\notag,
    \end{equation}
    for which we define the time $S$ spends in state $i$ over $[0,T]$
    \begin{equation}
        D_i:=\int_0^T\mathds{1}_{\{S_t=i\}}\mathrm{d}t=\sum_{p=1}^{m+1} \mathds{1}_{\{s_p=i\}}(u_p-u_{p-1})\notag,
    \end{equation}
    and the number of times $S$ jumps from state $i$ to state $j$ over $[0,T]$
    \begin{equation}
        w_{ij}:=\int_0^T\underset{\Delta\to0^+}{\lim}\frac{1}{\Delta}\mathds{1}_{\{S_{t-\Delta}=i,\, S_t=j\}}\mathrm{d}t=\sum_{p=1}^m \mathds{1}_{\{s_p=i, s_{p+1}=j\}}\notag.
    \end{equation}
\end{quote}

\begin{proof}
    Denote by $(u_p)_{1\leq p\leq m}$ the ordered sequence of unobserved transition times of the Markov chain $S$ over the time interval $[0,T]$, with the conventions $u_0=0$, $u_{m+1}=T$. Let $(s_p)_{1\leq p \leq m}$ be the sequence of states such that $S_t=s_p$, for any $t\in[u_{p-1},u_p)$. Let $z_p:=N_{u_p}-N_{u_{p-1}}$, $1\leq p\leq m+1$, be the number of events between the two state transition times $u_{p-1}$ and $u_p$, with the convention $z_0=0$. Define
    \begin{equation}
        \rho_k:=\sum_{p=0}^kz_p\notag,\hspace{0.3cm} 0\leq k\leq m+1,
    \end{equation}
    the number of events in the interval $[0,u_k]$. The likelihood of the model is
    \begin{align}
        \mathcal{L}(\Theta,\mathcal{T})&=\xi_0^{s_1}\left(\prod_{p=1}^m e^{-q_{s_p}(u_p-u_{p-1})}q_{s_ps_{p+1}}\right)e^{-q_{s_{m+1}}(T-u_m)}\prod_{p=1}^{m+1}\prod_{k=1}^{z_p}\lambda_{t_{\rho_{p-1}+k}}^{s_p}e^{-\int_{t_{\rho_{p-1}+k-1}}^{t_{\rho_{p-1}+k}}\lambda_t^{s_p}\mathrm{d}t}\notag\\
        &=\left(\sum_{i=1}^M\mathds{1}_{\{s_0=i\}}\xi_0^i\right)\exp\left(-\sum_{p=1}^m \sum_{i=1}^M \mathds{1}_{\{s_p=i\}}q_i(u_p-u_{p-1})\right)\left(\prod_{p=1}^m \sum_{1\leq i,j \leq M,i\neq j} \mathds{1}_{\{s_p=i, s_{p+1}=j\}}q_{ij}\right)\notag\\
        &\hspace{0.5cm}\times \exp\left(-\sum_{i=1}^M\mathds{1}_{\{s_{m+1}=i\}}q_i(T-u_m)\right)\prod_{p=1}^{m+1}\left(\sum_{i=1}^M \mathds{1}_{\{s_p=i\}}\exp\left(-\int_{t_{\rho_{p-1}}}^{t_{\rho_p}}\lambda_t^i\mathrm{d}t\right)\prod_{k=1}^{z_p}\lambda_{t_{\rho_{p-1}+k}}^i\right)\notag.
    \end{align}
    By taking the logarithm of both sides of the equation, we obtain
    \begin{align}
        \log\mathcal{L}(\Theta,\mathcal{T})&=\sum_{i=1}^M\mathds{1}_{\{s_0=i\}}\log\xi_0^i - \sum_{i=1}^M \sum_{p=1}^{m+1} \mathds{1}_{\{s_p=i\}}(u_p-u_{p-1})q_i + \sum_{1\leq i,j \leq M,i\neq j} \sum_{p=1}^m \mathds{1}_{\{s_p=i, s_{p+1}=j\}}\log q_{ij}\notag\\
        &\hspace{0.5cm}+ \sum_{i=1}^M \sum_{p=1}^{m+1} \mathds{1}_{\{s_p=i\}}\left(\sum_{k=1}^{z_p}\log\lambda_{t_{\rho_{p-1}+k}}^i-\int_{t_{\rho_{p-1}}}^{t_{\rho_p}}\lambda_t^i\mathrm{d}t\right)\notag\\
        &=\sum_{i=1}^M\left(\mathds{1}_{\{s_0=i\}}\log\xi_0^i - \sum_{p=1}^{m+1} \mathds{1}_{\{s_p=i\}}(u_p-u_{p-1})q_i + \sum_{j=1,j\neq i}^M \sum_{p=1}^m \mathds{1}_{\{s_p=i, s_{p+1}=j\}}\log q_{ij}\right)\notag\\
        &\hspace{0.5cm}+ \sum_{i=1}^M \sum_{n=1}^{K} \left(\mathds{1}_{\{s_{t_n}=i\}}\log\lambda_{t_n}^i-\int_{t_{n-1}}^{t_n}\mathds{1}_{\{s_t=i\}}\lambda_t^i\mathrm{d}t\right)\notag.
    \end{align}
    Use the definitions of $D_i$ and $w_{ij}$ in Equations \eqref{eq:D_i} and \eqref{eq:w_ij} to complete the proof.
\end{proof}

\subsection{Proof of Proposition~\ref{prop:e_step_likelihood}}

\begin{quote}\it
    Given an estimated set of parameters $\widehat{\Theta}$, the expected log-likelihood of the E-step can be written as
    \begin{equation}
        \mathbb{E}\left(\log \mathcal{L}(\Theta,\mathcal{T})|\mathcal{F}_T,\widehat{\Theta}\right)=\mathbb{E}\left(\log \mathcal{L}^M(\Theta,\mathcal{T})\Big|\mathcal{F}_T,\widehat{\Theta}\right)+\mathbb{E}\left(\log \mathcal{L}^H(\Theta,\mathcal{T})\Big|\mathcal{F}_T,\widehat{\Theta}\right)\notag,
    \end{equation}
    with
    \begin{equation}
        \mathbb{E}\left(\log \mathcal{L}^M(\Theta,\mathcal{T})\Big|\mathcal{F}_T,\widehat{\Theta}\right)=\sum_{i=1}^M\left\{\xi_{0|T}^i(\widehat{\Theta})\log \xi_0^i-\mathbb{E}(D_i|\mathcal{F}_T,\widehat{\Theta})q_i+\sum_{j=1,j\neq i}^M\mathbb{E}(w_{ij}|\mathcal{F}_T,\widehat{\Theta})\log q_{ij}\right\}\notag,
    \end{equation}
    and
    \begin{equation}
        \mathbb{E}\left(\log \mathcal{L}^H(\Theta,\mathcal{T})\Big|\mathcal{F}_T,\widehat{\Theta}\right)=\sum_{i=1}^M\left\{\sum_{n=1}^K\xi_{t_n|T}^i(\widehat{\Theta})\log \lambda_{t_n}^i(\Theta)-\int_0^T\xi_{t|T}^i(\widehat{\Theta})\lambda_t^i(\Theta)\mathrm{d}t\right\}\notag.
    \end{equation}

    Furthermore, we have the analytic expressions

    \begin{equation}
        \mathbb{E}\left(D_i|\mathcal{F}_T,\widehat{\Theta}\right)=\sum_{n=1}^K\frac{1}{c_n}\left[\sum_{k=0}^{\ell_n-1}\left[e^{C^{(n)}_k\delta}\right]^{M\boxtimes M}+\left[e^{D^{(n)}\Delta_n}\right]^{M\boxtimes M}\right]_{ii}\notag,
    \end{equation}

    \begin{equation}
        \mathbb{E}\left(w_{ij}|\mathcal{F}_T,\widehat{\Theta}\right)=\widehat{q}_{ij}\sum_{n=1}^K\frac{1}{c_n}\left[\sum_{k=0}^{\ell_n-1}\left[e^{C^{(n)}_k\delta}\right]^{M\boxtimes M}+\left[e^{D^{(n)}\Delta_n}\right]^{M\boxtimes M}\right]_{ji}\notag,
    \end{equation}

    \begin{equation}
        \int_0^T\xi_{t|T}^i(\widehat{\Theta})\lambda_t^i\mathrm{d}t=\sum_{n=1}^K\frac{1}{c_n}\left[\sum_{k=0}^{\ell_n-1}\Lambda_{t_{n-1}+k\delta}\,\odot\,\left[e^{C^{(n)}_k\delta}\right]^{M\boxtimes M}+\Lambda_{t_{n-1}+\ell_n\delta}\,\odot\,\left[e^{D^{(n)}\Delta_n}\right]^{M\boxtimes M}\right]_{ii}\notag,
    \end{equation}

    and the smoothed probabilities are given by Equations \eqref{eq:filtered_prob_vs_forward_backward} and \eqref{eq:smoothed_prob_vs_forward_backward}.
\end{quote}
\begin{proof}
    \begin{equation}
        \mathbb{E}\left(\log \mathcal{L}^M(\Theta,\mathcal{T})\Big|\mathcal{F}_T,\widehat{\Theta}\right)=\sum_{i=1}^M\left\{\xi_{0|T}^i(\widehat{\Theta})\log \xi_0^i-\mathbb{E}(D_i|\mathcal{F}_T,\widehat{\Theta})q_i+\sum_{j=1,j\neq i}^M\mathbb{E}(w_{ij}|\mathcal{F}_T,\widehat{\Theta})\log q_{ij}\right\}\notag,
    \end{equation}
    Using the work of \cite{ryden1994parameter}, we can write
    \begin{align}
        \mathbb{E}\left(D_i|\mathcal{F}_T,\widehat{\Theta}\right)&=\int_0^T\frac{\alpha_t^i\beta_t^i}{\mathcal{L}}\mathrm{d}t\notag,\hspace{0.3cm}1\leq i\leq M,\\
        \mathbb{E}\left(w_{ij}|\mathcal{F}_T,\widehat{\Theta}\right)&=\int_0^Tq_{ij}\frac{\alpha_t^i\beta_t^j}{\mathcal{L}}\mathrm{d}t\notag,\hspace{0.3cm}1\leq i,j\leq M,i\neq j.
    \end{align}
    We now aim at computing the integrals
    \begin{equation}
        \int_0^T\frac{\alpha_t^i\beta_t^j}{\mathcal{L}}\,\mathrm{d}t,\hspace{0.3cm}\int_0^T\frac{\alpha_t^i\beta_t^i}{\mathcal{L}}\lambda_t^i\,\mathrm{d}t\notag,\hspace{0.3cm}1\leq i,j\leq M.
    \end{equation}
    Using the $\delta$-piecewise constant kernel hypothesis, write
    \begin{align}
        \int_0^T\frac{\alpha_t^i\beta_t^i}{\mathcal{L}}\lambda_t^i\mathrm{d}t&=\sum_{n=1}^K\frac{1}{\mathcal{L}}\int_{t_{n-1}}^{t_n}\alpha_t^i\beta_t^i\lambda_t^i\mathrm{d}t\notag\\
        &=\sum_{n=1}^K\frac{1}{\prod_{k=1}^Kc_k}\int_{t_{n-1}}^{t_n}\pi\prod_{k=1}^{n-1}f^{(k)}(x_k)\lambda_t^iH^{(n)}\left(t-t_{n-1}\right)\mathds{1}_i\mathds{1}_i'G^{(n)}\left(t_n-t\right)\Lambda_{t_n^-}\prod_{k=n+1}^Kf^{(k)}(x_k)\mathds{1}\mathrm{d}t\notag\\
        &=\sum_{n=1}^K\frac{1}{c_n}\left[\int_0^{x_n}\Lambda_{t_{n-1}+x}\odot G^{(n)}(x_n-x)P^{(n)}H^{(n)}(x)\mathrm{d}x\right]_{ii}\notag\\
        &=\sum_{n=1}^K\frac{1}{c_n}\bigg[\sum_{k=0}^{\ell_n-1}\Lambda_{t_{n-1}+k\delta}\odot\int_{k\delta}^{(k+1)\delta}G^{(n)}(x_n-x)P^{(n)}H^{(n)}(x)\mathrm{d}x\notag\\
        &\hspace{1cm}+\Lambda_{t_{n-1}+\ell_n\delta}\odot\int_{\ell_n\delta}^{\ell_n\delta+\Delta_n}G^{(n)}(x_n-x)P^{(n)}H^{(n)}(x)\mathrm{d}x\bigg]_{ii}\notag\\
        &=\sum_{n=1}^K\frac{1}{c_n}\bigg[\sum_{k=0}^{\ell_n-1}\Lambda_{t_{n-1}+k\delta}\odot\int_{0}^{\delta}G^{(n)}(x_n-(x+k\delta))P^{(n)}H^{(n)}(x+k\delta)\mathrm{d}x\notag\\
        &\hspace{1cm}+\Lambda_{t_{n-1}+\ell_n\delta}\odot\int_{0}^{\Delta_n}G^{(n)}(x_n-(x+\ell_n\delta))P^{(n)}H^{(n)}(x+\ell_n\delta)\mathrm{d}x\bigg]_{ii}\notag\\
        &=\sum_{n=1}^K\frac{1}{c_n}\bigg[\sum_{k=0}^{\ell_n-1}\Lambda_{t_{n-1}+k\delta}\odot\int_{0}^{\delta}e^{(Q-\Lambda_{t_{n-1}+k\delta})(\delta-x)}\Omega^{(n)}_ke^{(Q-\Lambda_{t_{n-1}+k\delta})x}\mathrm{d}x\notag\\
        &\hspace{1cm}+\Lambda_{t_{n-1}+\ell_n\delta}\odot\int_{0}^{\Delta_n}e^{(Q-\Lambda_{t_{n-1}+\ell_n\delta})(\Delta_n-x)}P^{(n)}\Xi^{(n)}_{\ell_n}e^{(Q-\Lambda_{t_{n-1}+\ell_n\delta})\Delta_n}\mathrm{d}x\bigg]_{ii}.\notag
    \end{align}
    We end up computing integrals of matrix exponentials. Following \cite{van1978computing}, as recommended in \cite{roberts2006ryde}, we get
    \begin{equation}
        \int_{0}^{\delta}e^{(Q-\Lambda_{t_{n-1}+k\delta})(\delta-x)}\Omega^{(n)}_ke^{(Q-\Lambda_{t_{n-1}+k\delta})x}\mathrm{d}x=\left[e^{C^{(n)}_k\delta}\right]^{M\boxtimes M}\notag,
    \end{equation}
    and
    \begin{equation}
        \int_{0}^{\Delta_n}e^{(Q-\Lambda_{t_{n-1}+\ell_n\delta})(\Delta_n-x)}P^{(n)}\Xi^{(n)}_{\ell_n}e^{(Q-\Lambda_{t_{n-1}+\ell_n\delta})\Delta_n}\mathrm{d}x=\left[e^{D^{(n)}\Delta_n}\right]^{M\boxtimes M}\notag.
    \end{equation}
    Finally, we obtain
    \begin{equation}
        \int_0^T\frac{\alpha_t^i\beta_t^i}{\mathcal{L}}\lambda_i(t)\mathrm{d}t=\sum_{n=1}^K\frac{1}{c_n}\left[\sum_{k=0}^{\ell_n-1}\Lambda_{t_{n-1}+k\delta}\odot\left[e^{C^{(n)}_k\delta}\right]^{M\boxtimes M}+\Lambda_{t_{n-1}+\ell_n\delta}\odot \left[e^{D^{(n)}\Delta_n}\right]^{M\boxtimes M}\right]_{ii}\notag.
    \end{equation}
    In particular, we have the following formula for the integrated transition probability
    \begin{equation}
        \int_0^T\frac{\alpha_t^i\beta_t^j}{\mathcal{L}}\mathrm{d}t=\sum_{n=1}^K\frac{1}{c_n}\left[\sum_{k=0}^{\ell_n-1}\left[e^{C^{(n)}_k\delta}\right]^{M\boxtimes M}+\left[e^{D^{(n)}\Delta_n}\right]^{M\boxtimes M}\right]_{ji}\notag,
    \end{equation}
    which completes the proof.
\end{proof}

\section{Supplementary figures}
This appendix provides additional figures that do not appear in the main text in order to enhance readability.  
Figure \ref{fig:box_plot_parameters_convergence_continuous_kernel} displays boxplots of the estimated parameters for the MMHP with $\delta$-piecewise constant intensity with respect to $\delta$, from the simulation experiments of Section \ref{section:estimation_procedure}. It shows the convergence of the parameters towards those of the continuous MMHP.
Finally, Figure  \ref{fig:box_plot_parameters_convergence_delta} displays boxplots of the estimated parameters for three MMHPs with different steps $\delta$ of Section \ref{section:estimation_procedure}. We observe that the stability of the estimated parameters does not appear to depend on $\delta$.

\begin{figure}
    \centering
    \subfloat[$\mu_1$]{%
        \includegraphics[width=0.25\linewidth]{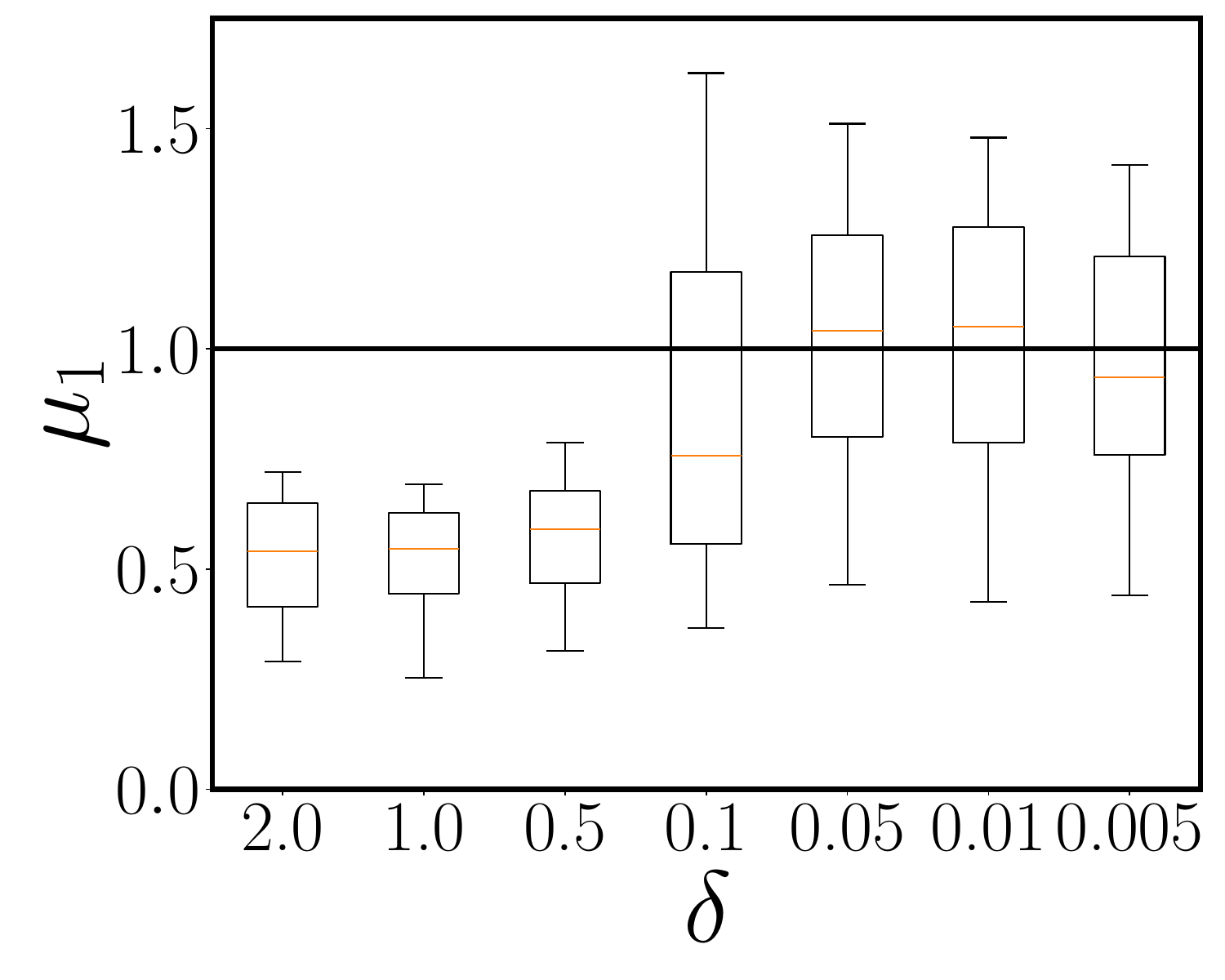}%
    }
    \subfloat[$\mu_2$]{%
        \includegraphics[width=0.25\linewidth]{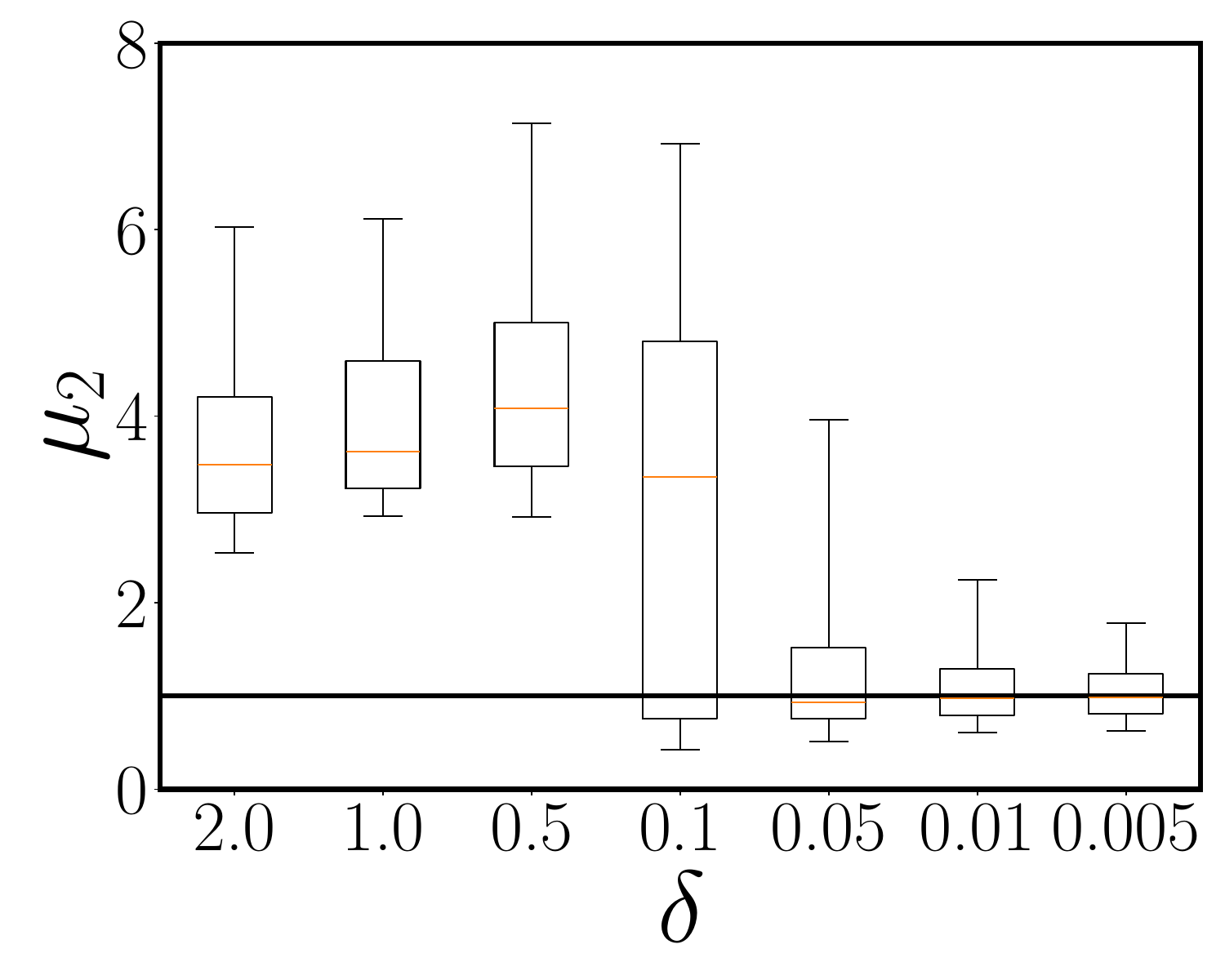}%
    }
    \subfloat[$\alpha_1$]{%
        \includegraphics[width=0.25\linewidth]{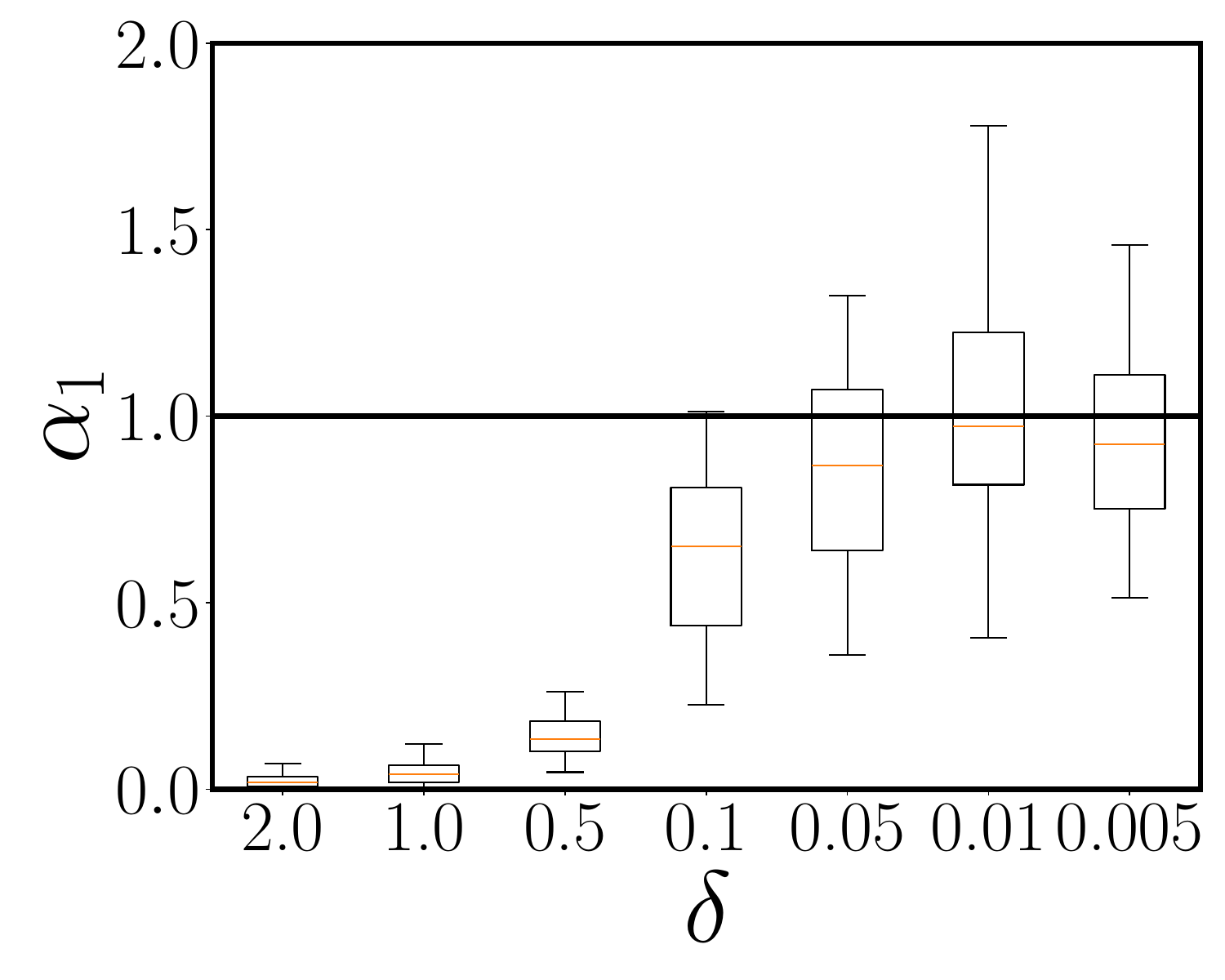}%
    }
    \subfloat[$\alpha_2$]{%
        \includegraphics[width=0.25\linewidth]{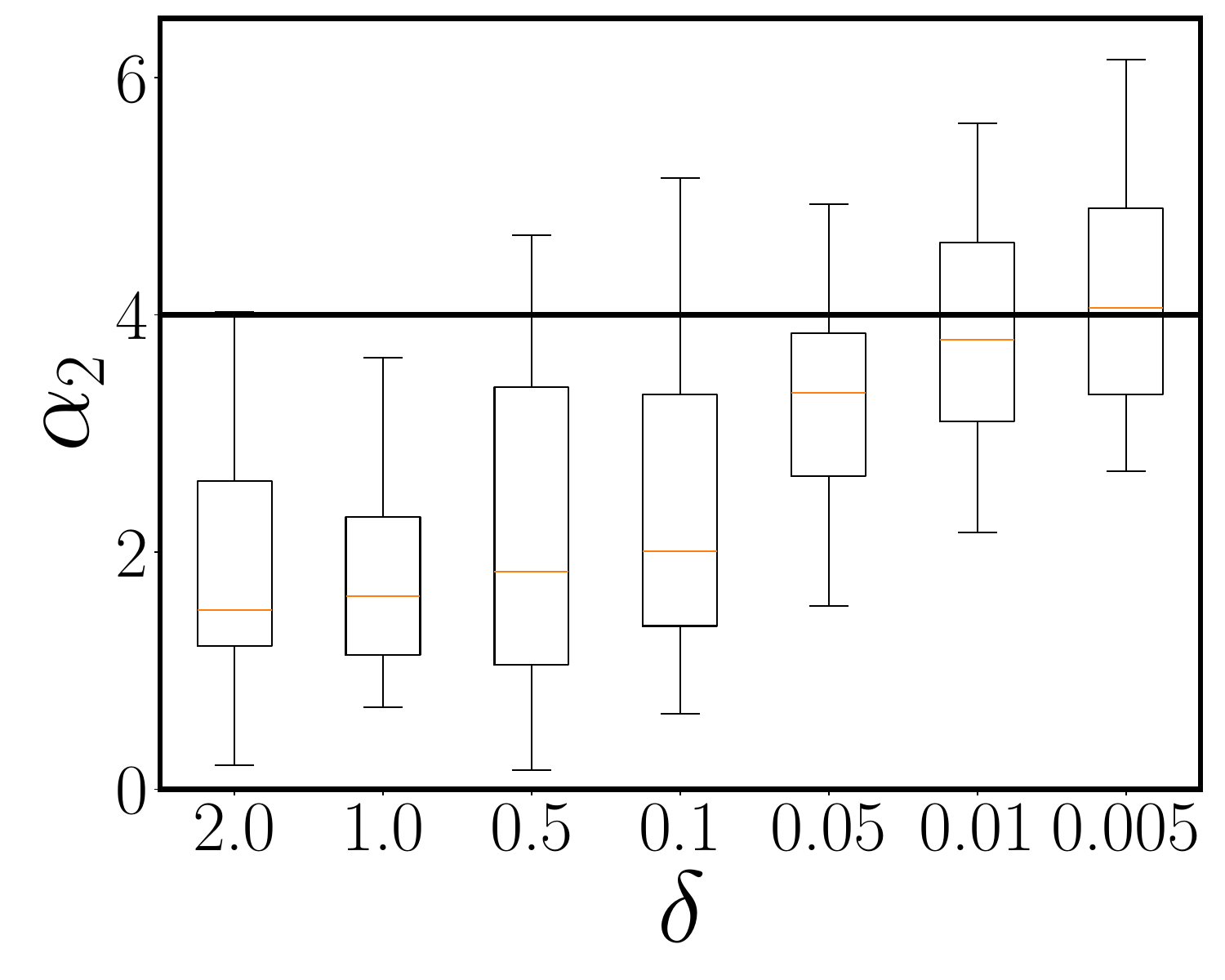}%
    }\\
    \subfloat[$\beta_1$]{%
        \includegraphics[width=0.25\linewidth]{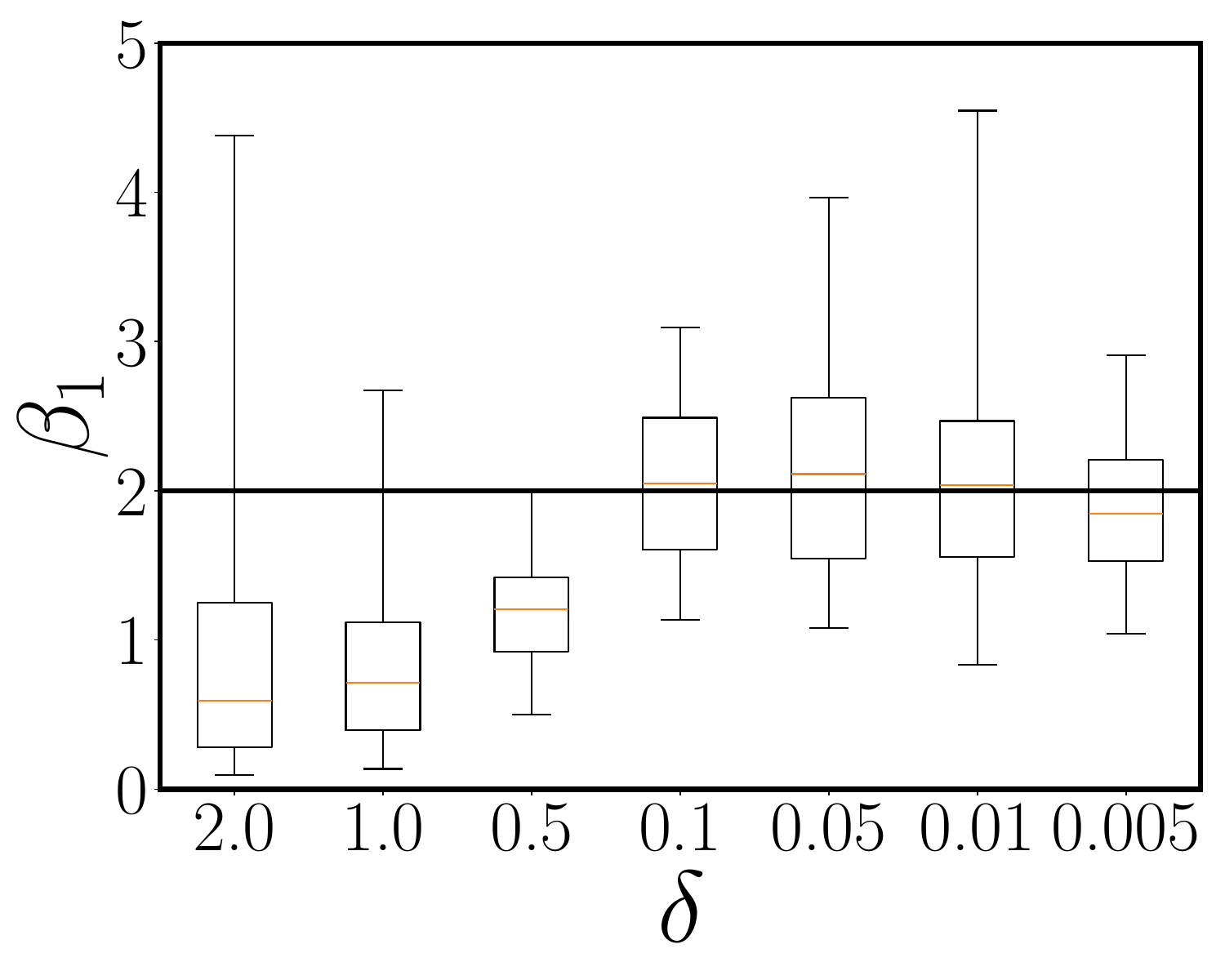}%
    }
    \subfloat[$\beta_2$]{%
        \includegraphics[width=0.25\linewidth]{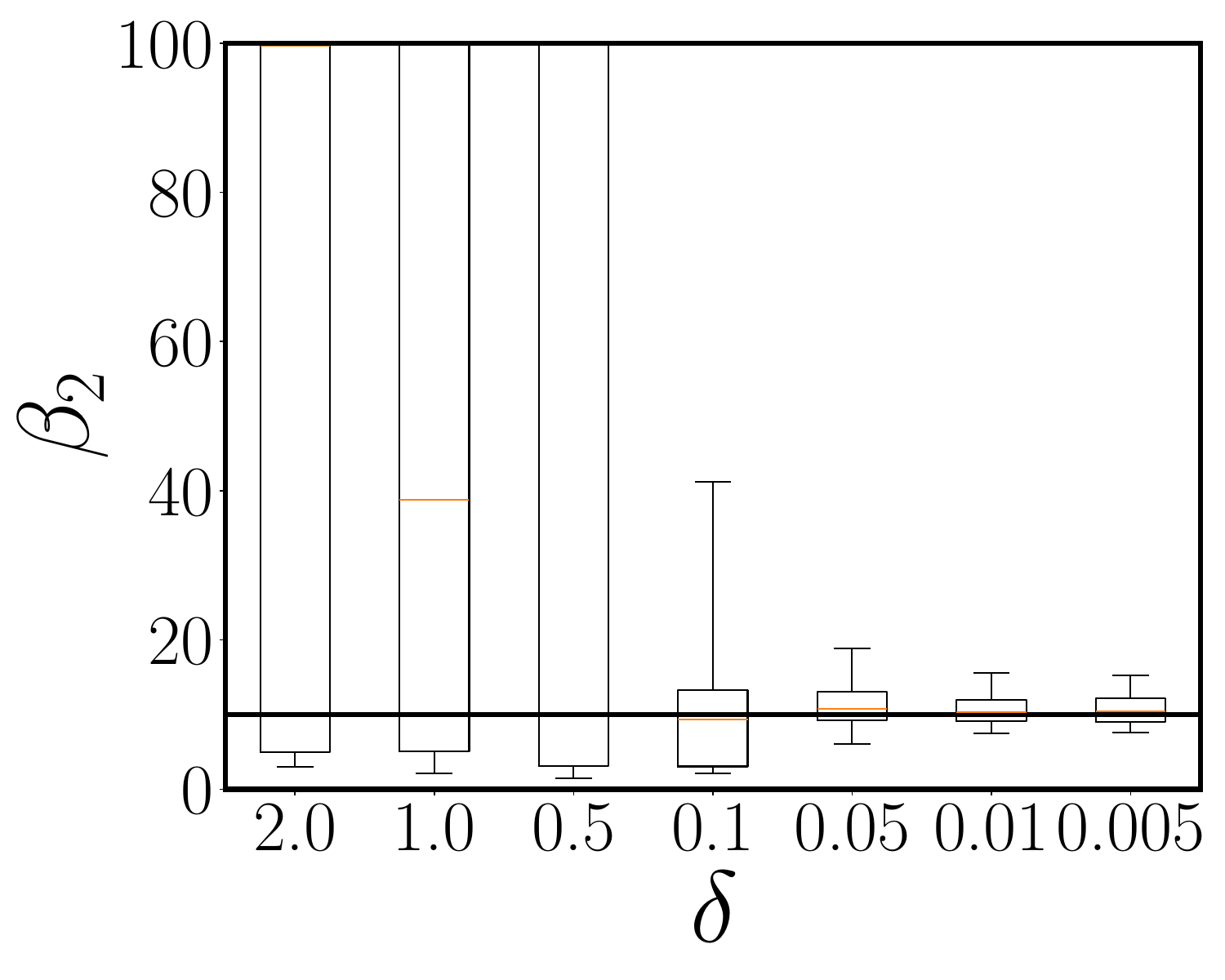}%
    }
    \centering
    \subfloat[$q_1$]{%
        \includegraphics[width=0.25\linewidth]{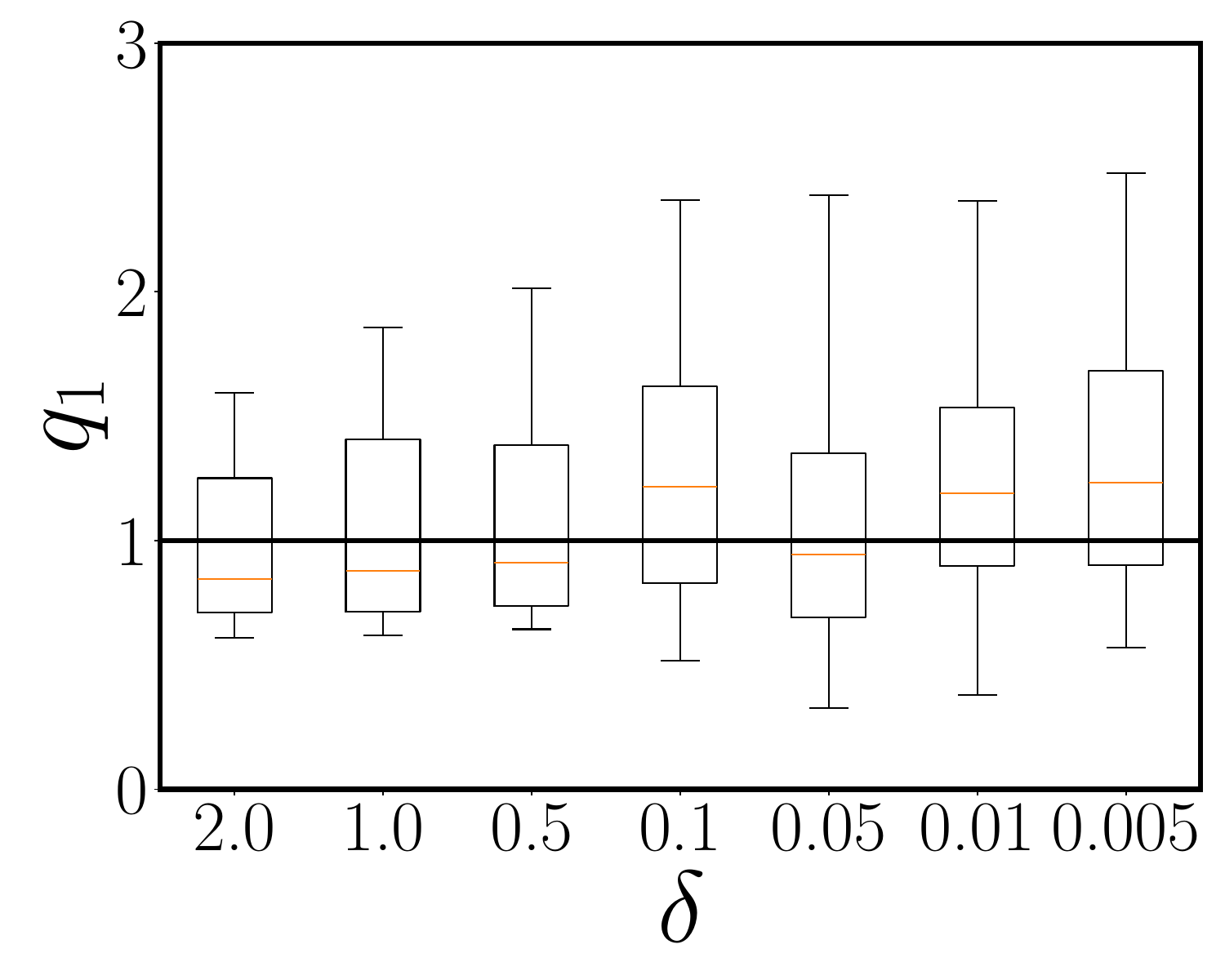}%
    }
    \subfloat[$q_2$]{%
        \includegraphics[width=0.25\linewidth]{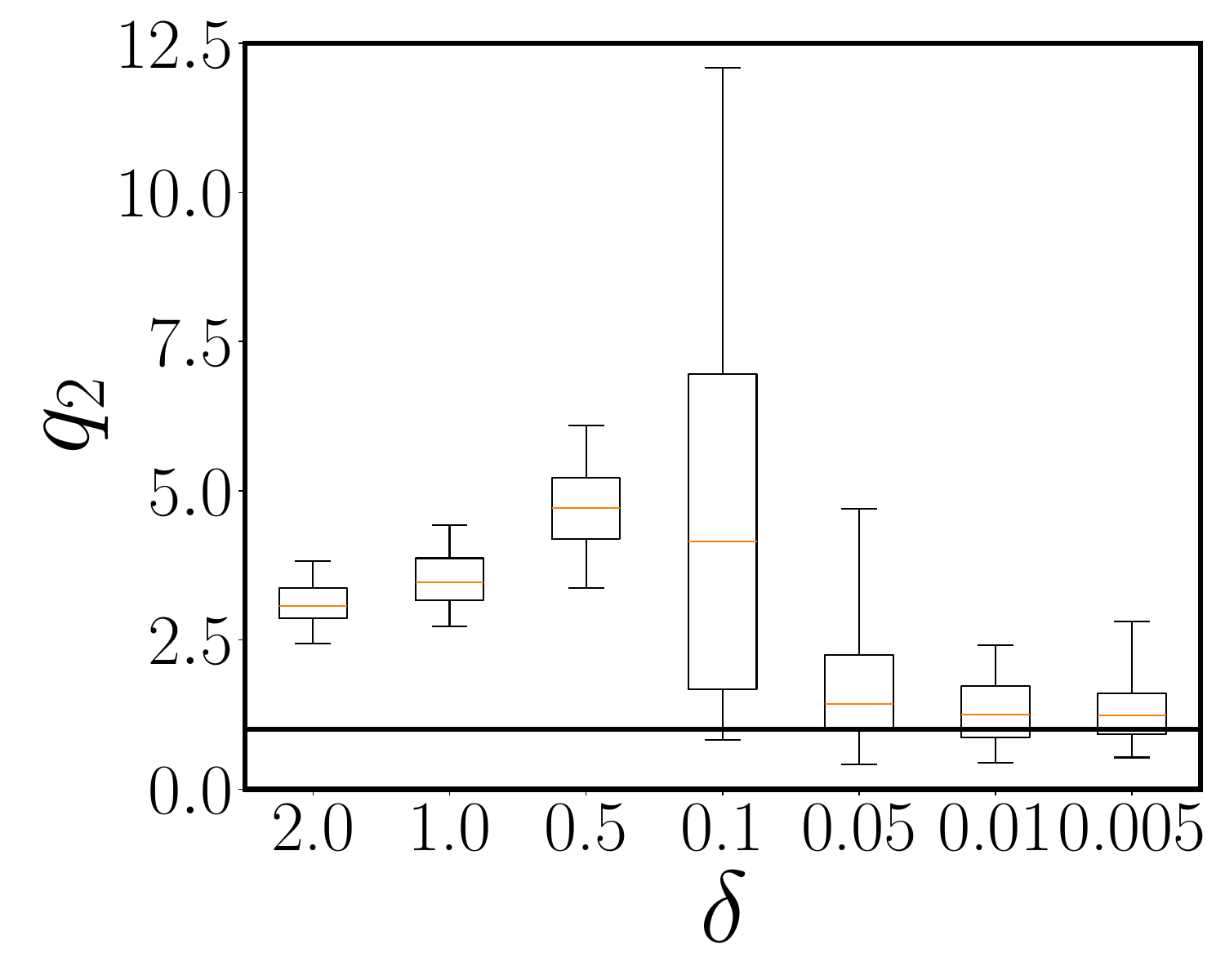}%
    }
    \caption{\textit{Convergence towards the continuous MMHP} --- Box plots of the parameters with respect to the parameter $\delta$. Whiskers represent 5\% and 95\% quantiles. The plain line represents the true parameter used for the simulation of the MMHP with continuous kernel.}
    \label{fig:box_plot_parameters_convergence_continuous_kernel}
\end{figure}

\begin{figure}
    \centering
    \subfloat[$\mu_1$]{%
        \includegraphics[width=0.25\linewidth]{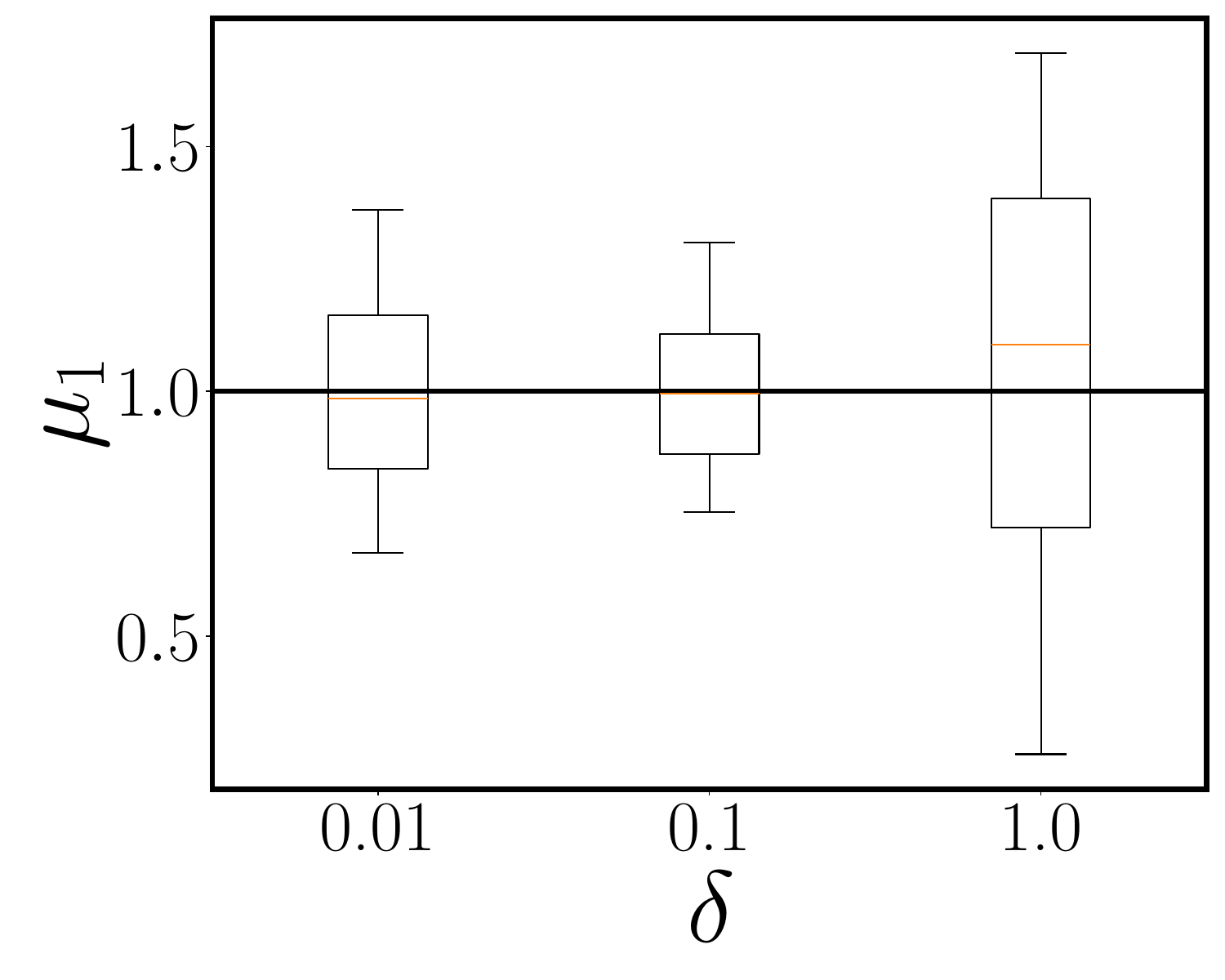}%
    }
    \subfloat[$\mu_2$]{%
        \includegraphics[width=0.25\linewidth]{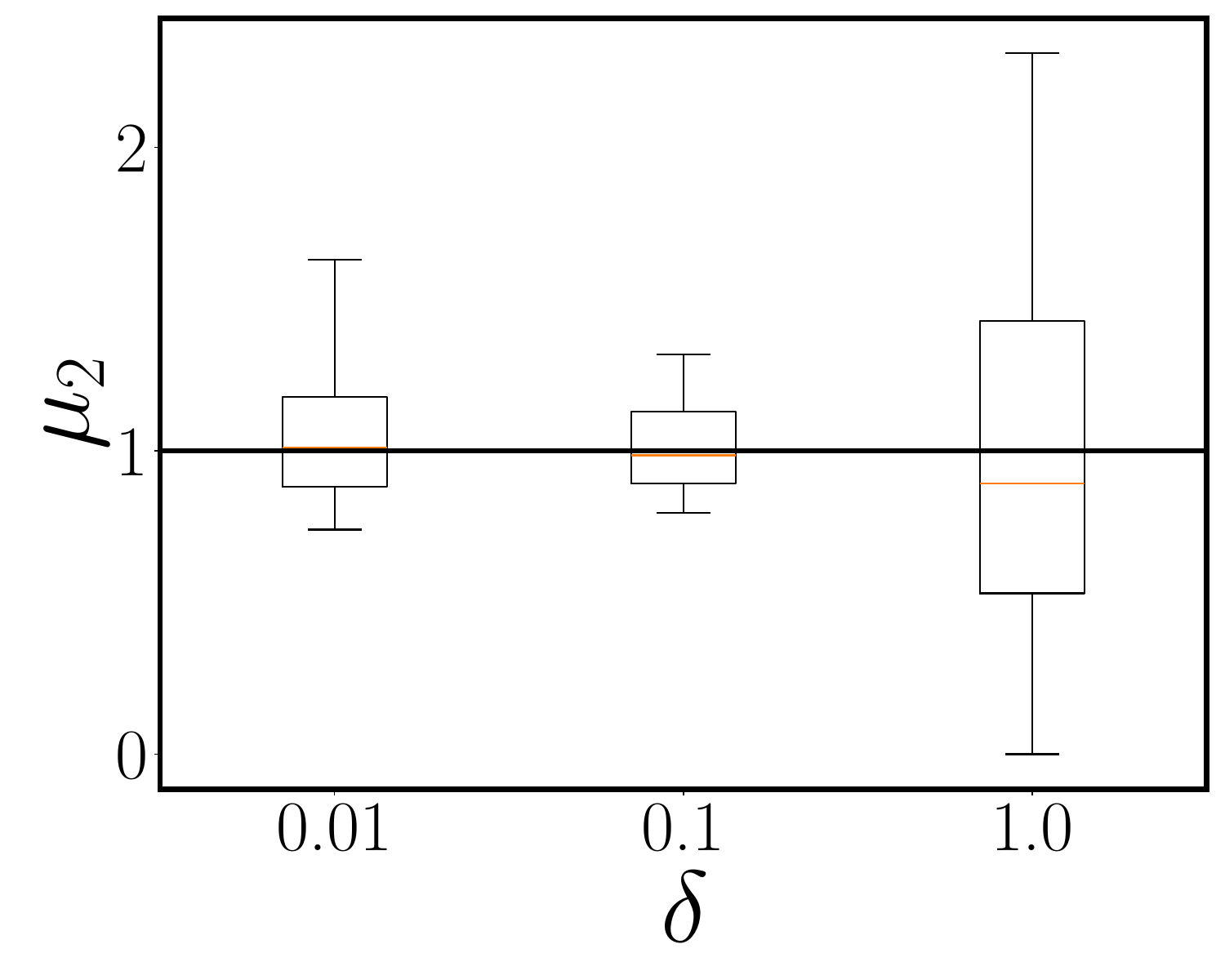}%
    }
    \subfloat[$\alpha_1$]{%
        \includegraphics[width=0.25\linewidth]{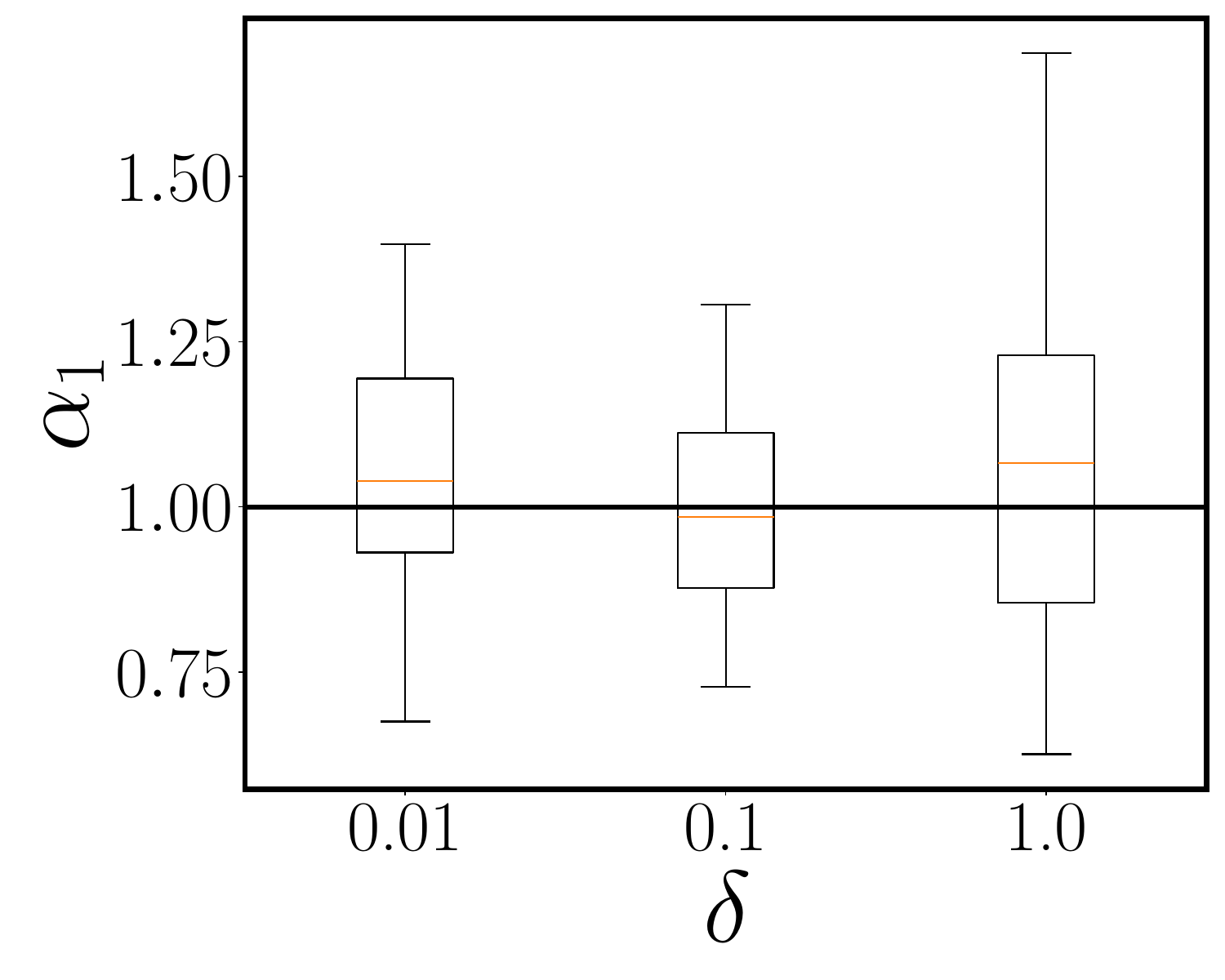}%
    }
    \subfloat[$\alpha_2$]{%
        \includegraphics[width=0.25\linewidth]{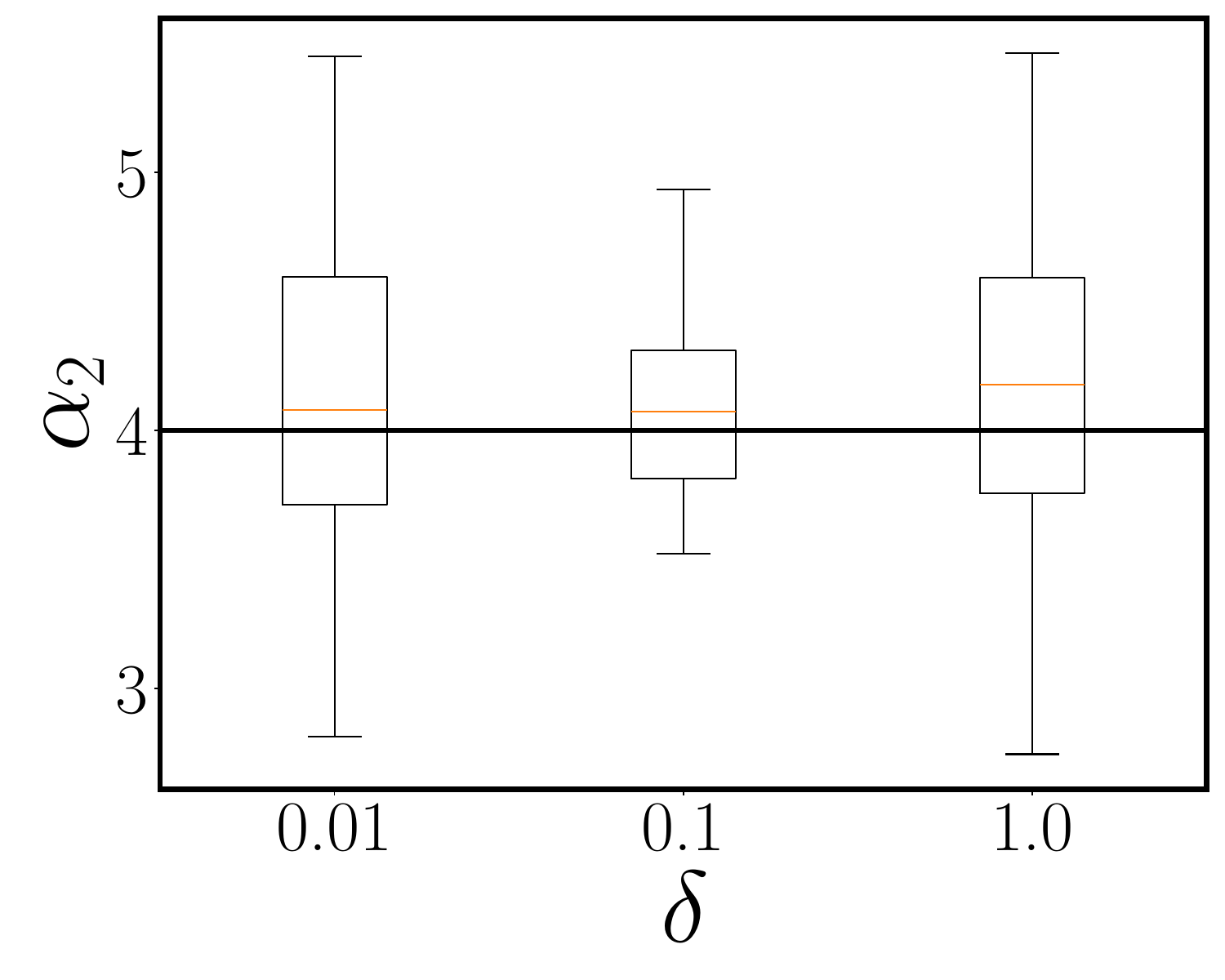}%
    }\\
    \subfloat[$\beta_1$]{%
        \includegraphics[width=0.25\linewidth]{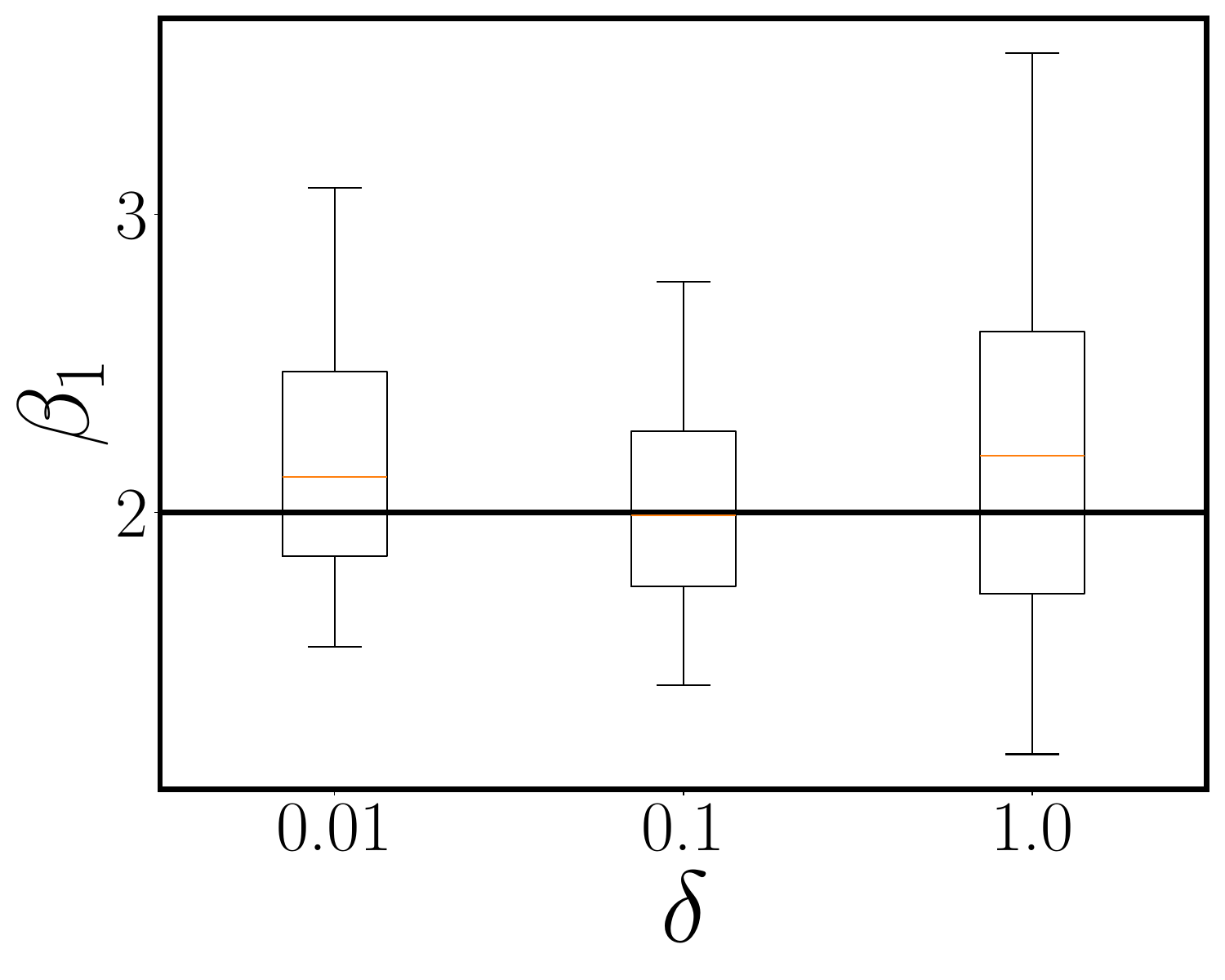}%
    }
    \subfloat[$\beta_2$]{%
        \includegraphics[width=0.25\linewidth]{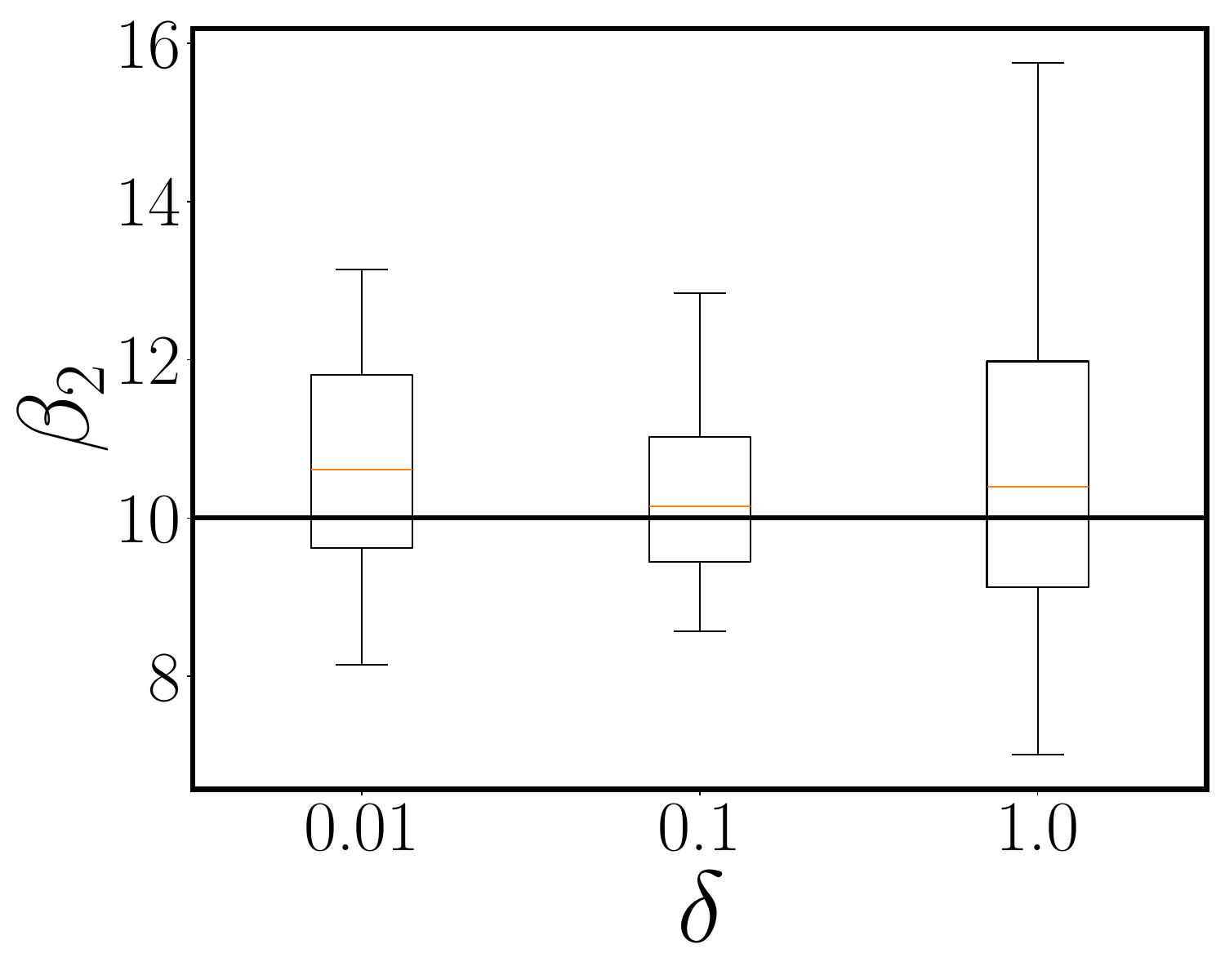}%
    }
    \centering
    \subfloat[$q_1$]{%
        \includegraphics[width=0.25\linewidth]{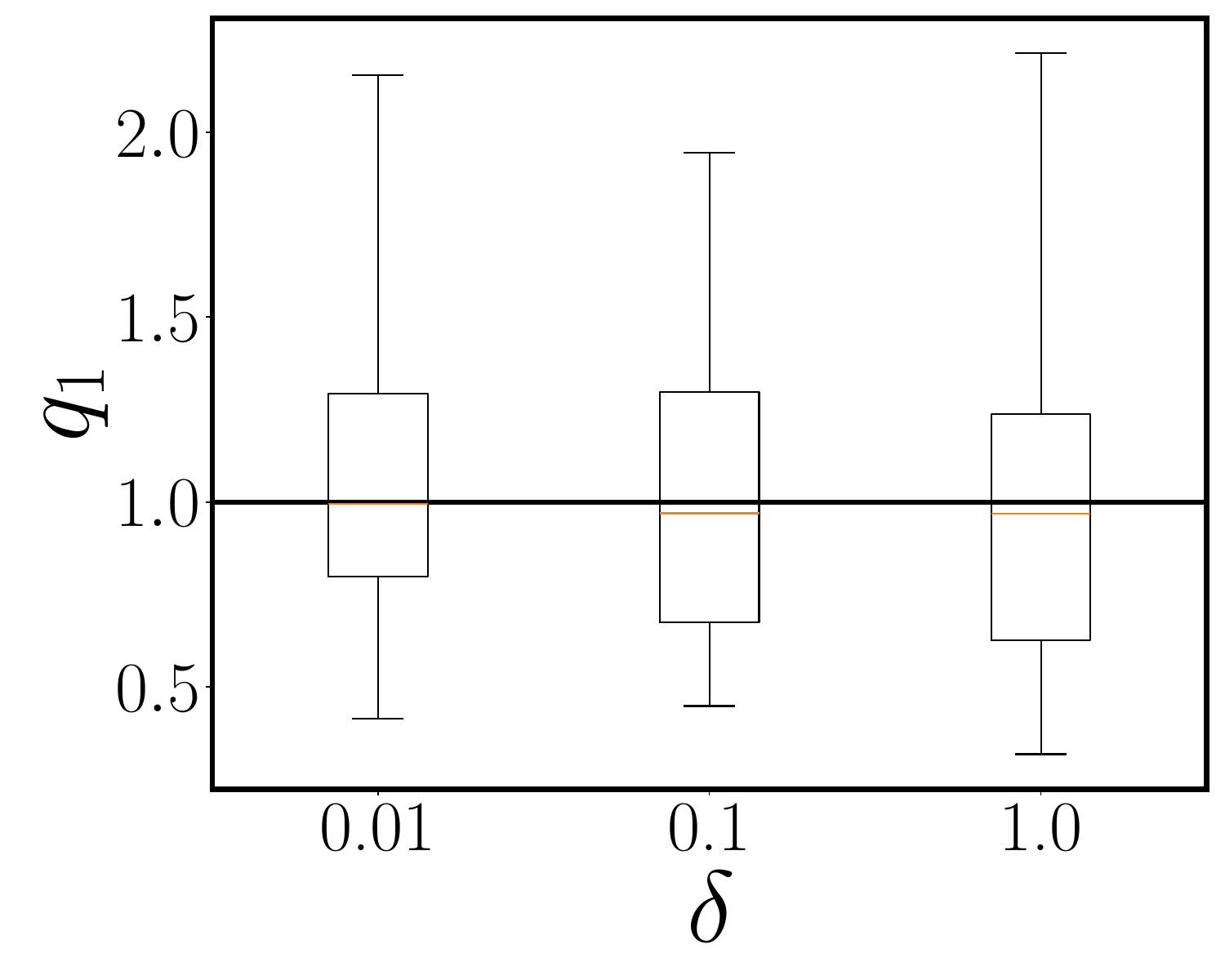}%
    }
    \subfloat[$q_2$]{%
        \includegraphics[width=0.25\linewidth]{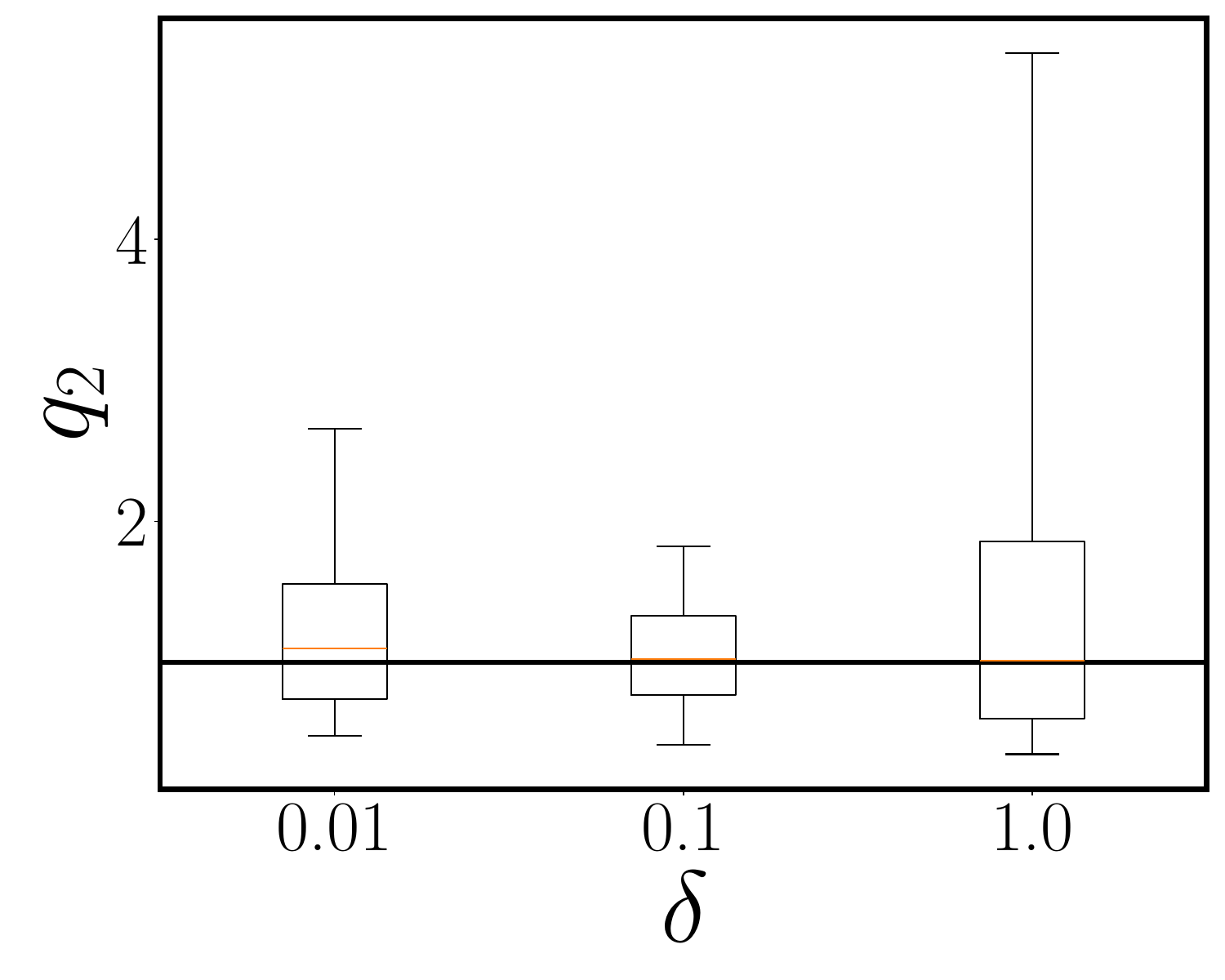}%
    }
    \caption{\textit{Convergence with respect to $\delta$} --- Box plots of the parameters with respect to the parameter $\delta$. Whiskers represent 5\% and 95\% quantiles. The plain line represents the true parameter.}
    \label{fig:box_plot_parameters_convergence_delta}
\end{figure}

\begin{figure}
    \centering
    \subfloat[BBOs' dynamics and trade prices]{%
        \includegraphics[width=0.5\linewidth]{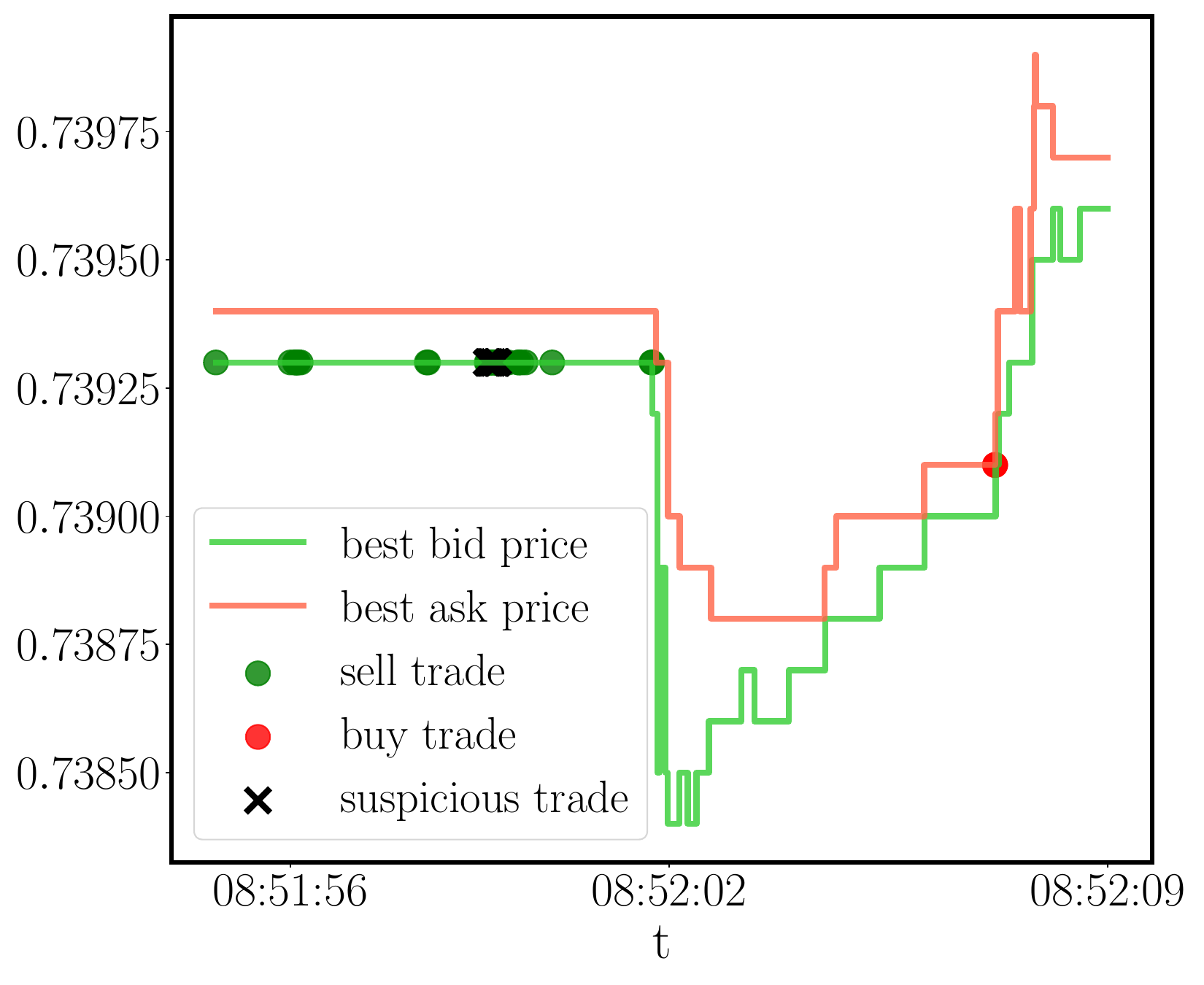}%
    }\hfill
    \subfloat[Trade sizes]{%
        \includegraphics[width=0.5\linewidth]{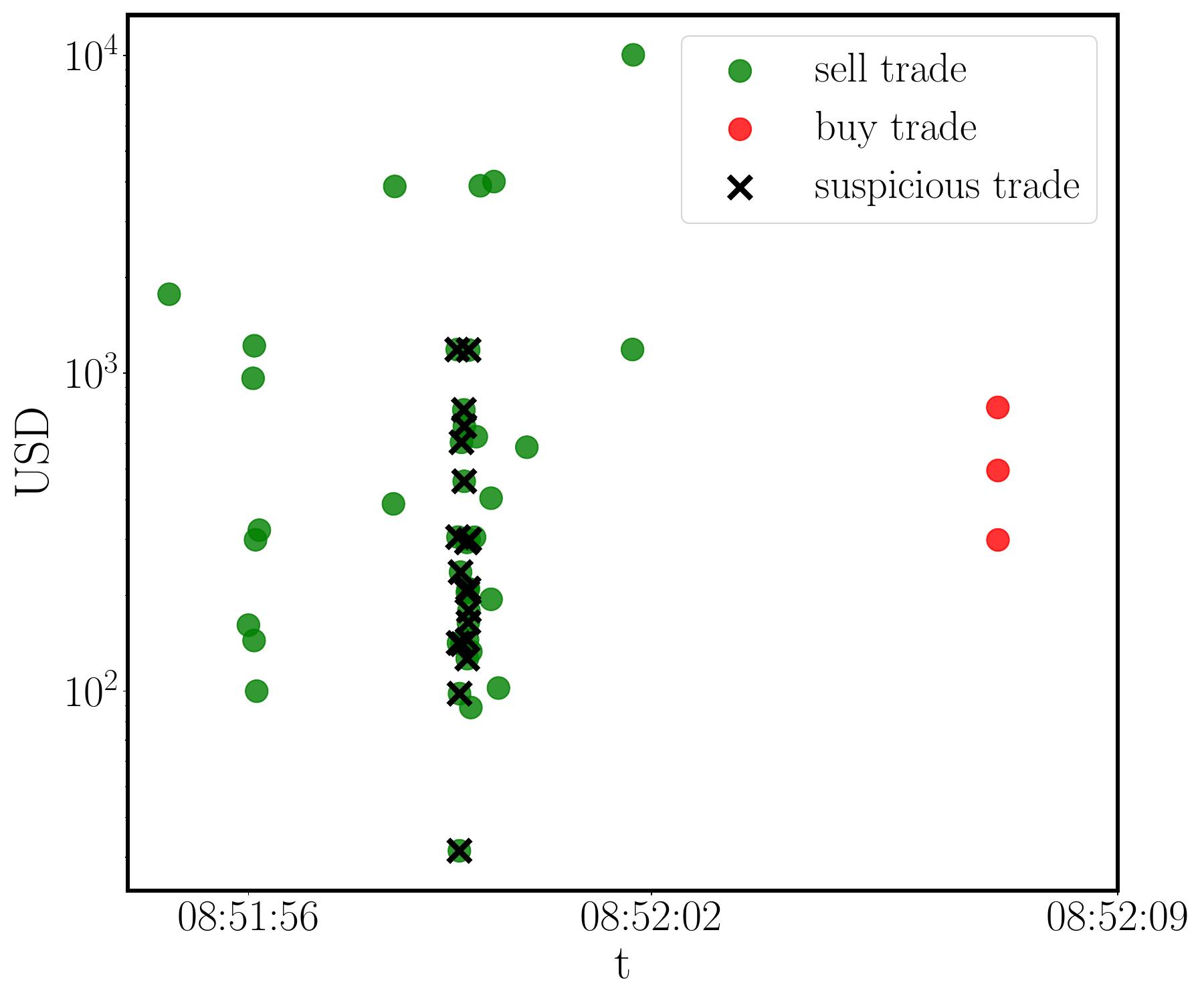}%
    }
    \caption{\textit{Suspicious trading activity} --- An example of ``suspicious'' behavior identified by the bid MMHP-$\delta$ as the extreme burst regime (state 3), anomalies detected on March 6th, 2024 UTC time zone.}
    \label{fig:bid_wash_trading_case_06032024}
\end{figure}

\begin{figure}
    \centering
    \subfloat[BBOs' dynamics and trade prices]{%
        \includegraphics[width=0.5\linewidth]{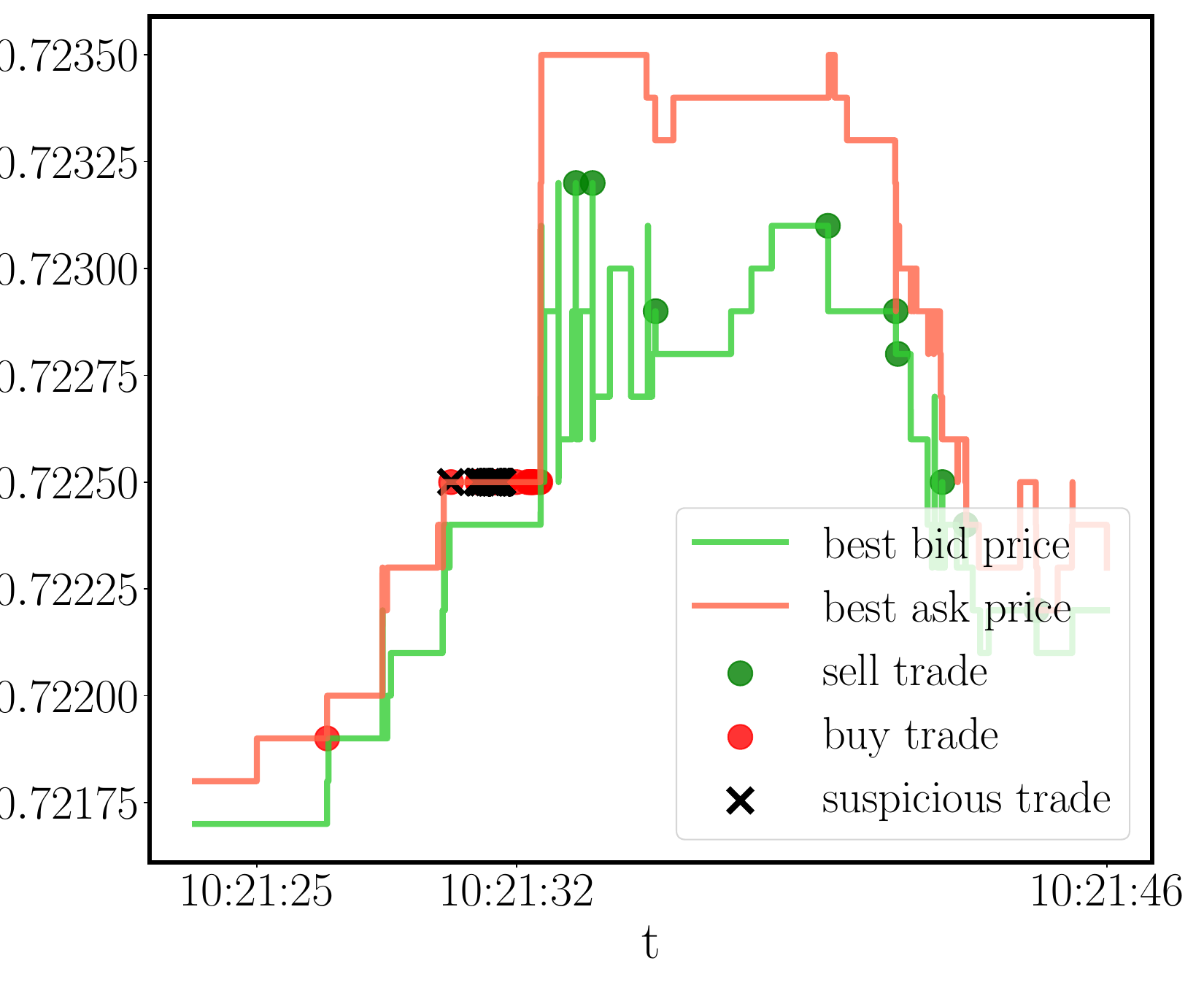}%
    }\hfill
    \subfloat[Trade sizes]{%
        \includegraphics[width=0.5\linewidth]{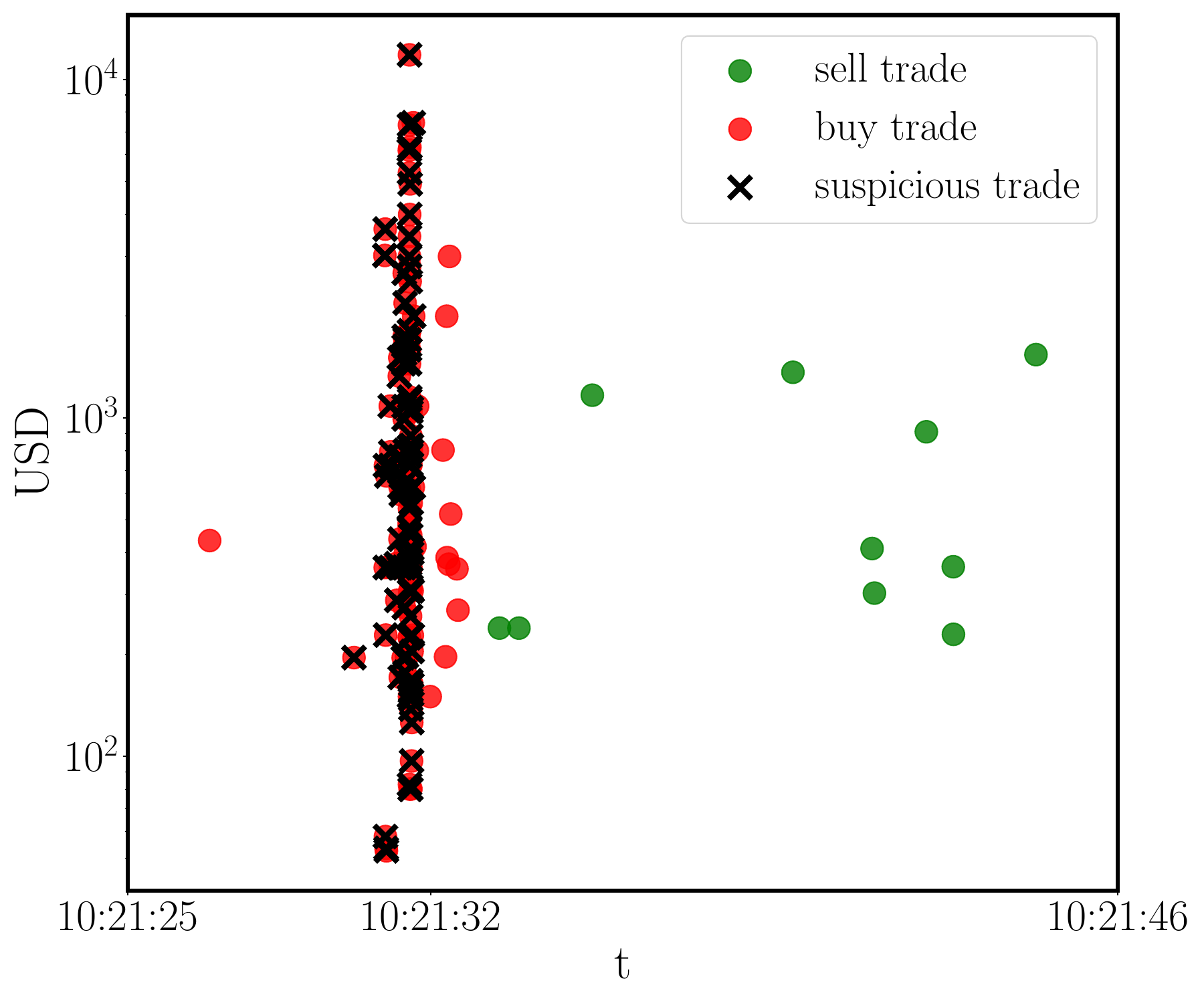}%
    }
    \caption{\textit{Suspicious trading activity} --- An example of ``suspicious'' behavior identified by the ask MMHP-$\delta$ as the extreme burst regime (state 3), anomalies detected on Janurary 15th, 2024 UTC time zone.}
    \label{fig:ask_wash_trading_case_15012024}
\end{figure}

\begin{figure}
    \centering
    \subfloat[BBOs' dynamics and trade prices]{%
        \includegraphics[width=0.5\linewidth]{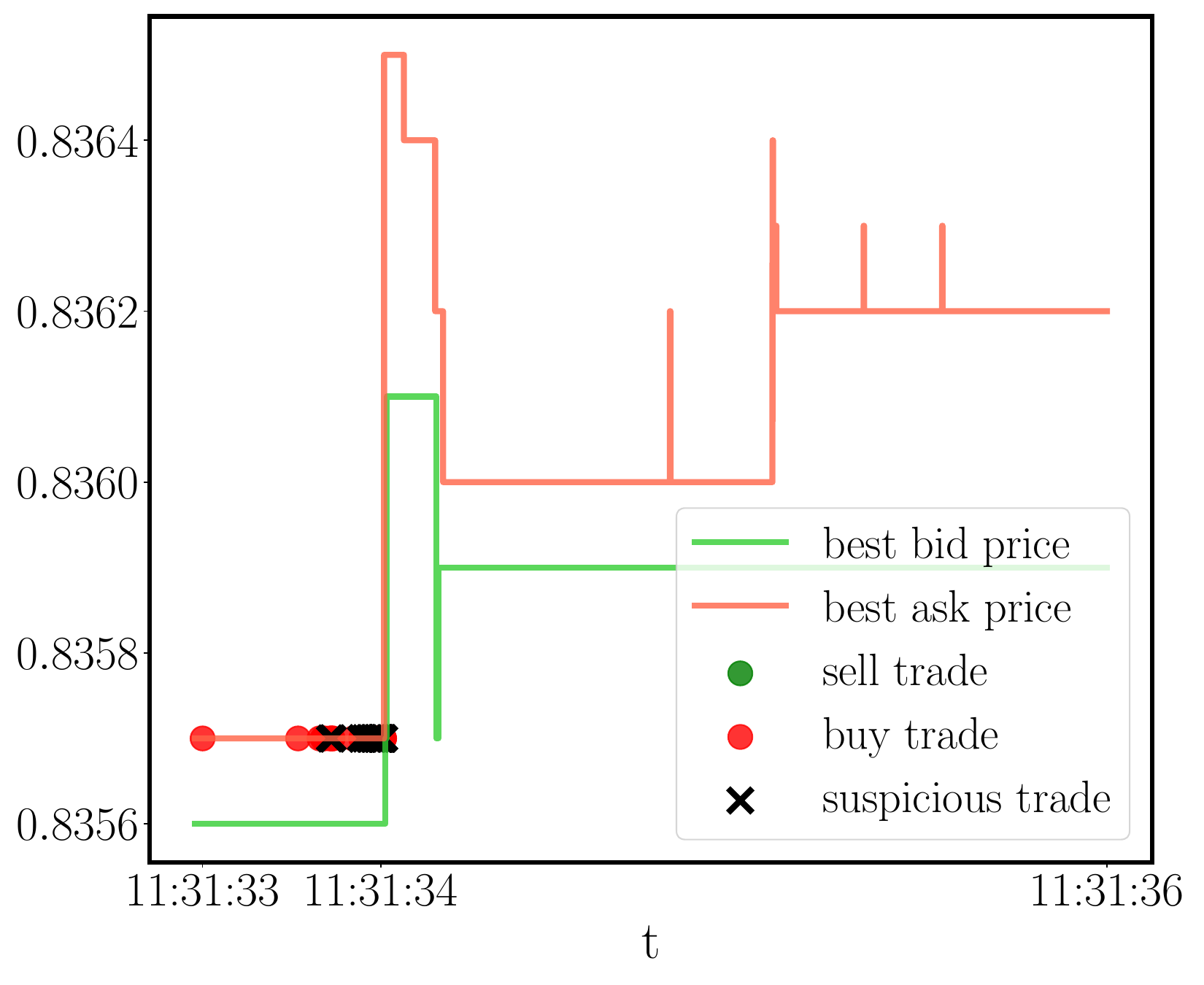}%
    }\hfill
    \subfloat[Trade sizes]{%
        \includegraphics[width=0.5\linewidth]{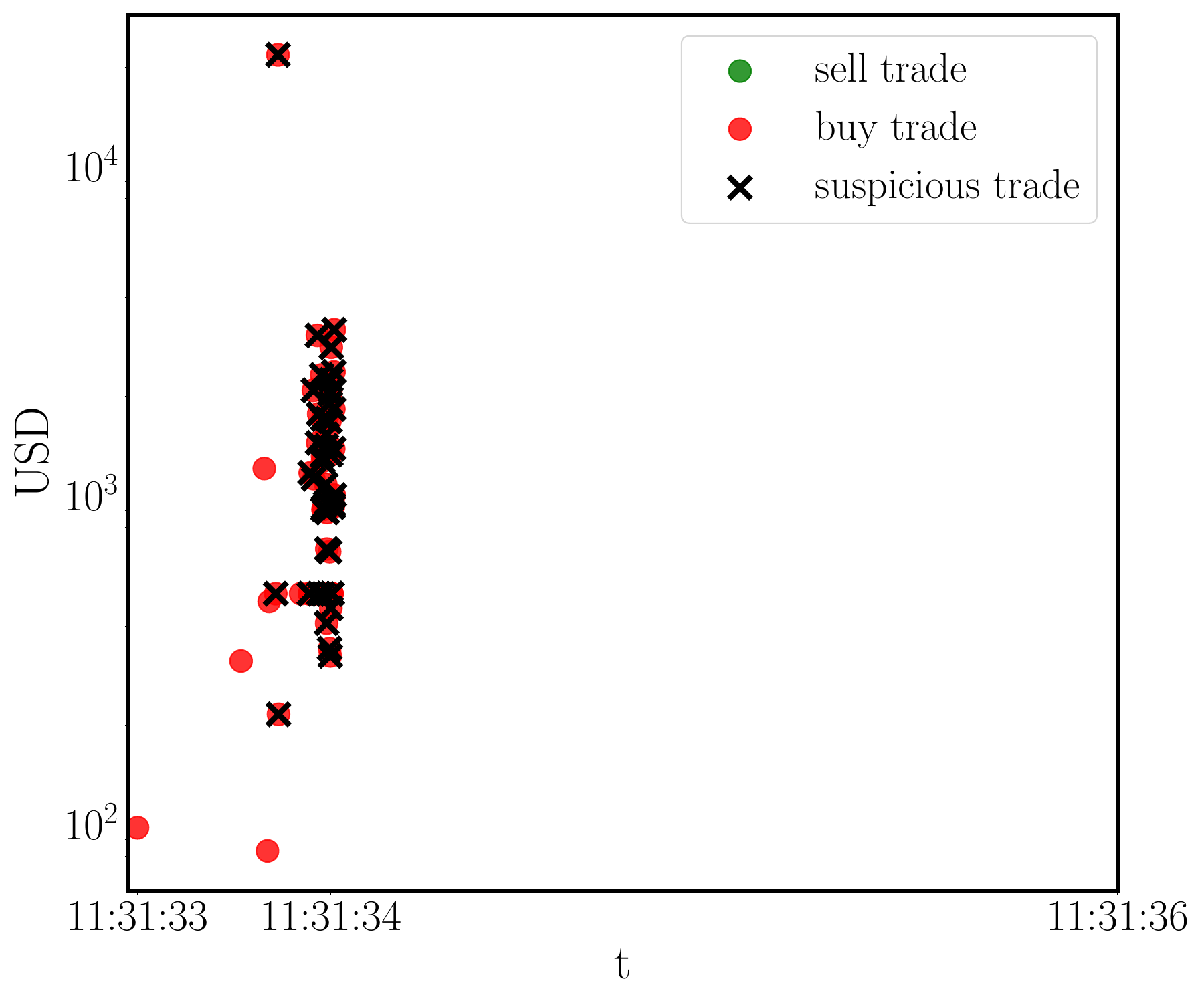}%
    }
    \caption{\textit{Suspicious trading activity} --- An example of ``suspicious'' behavior identified by the ask MMHP-$\delta$ as the extreme burst regime (state 3), anomalies detected on March 3th, 2024 UTC time zone.}
    \label{fig:ask_wash_trading_case_03032024}
\end{figure}

\begin{figure}
    \centering
    \subfloat[Imbalance distribution before state transition, bid]{%
        \includegraphics[width=0.33\linewidth]{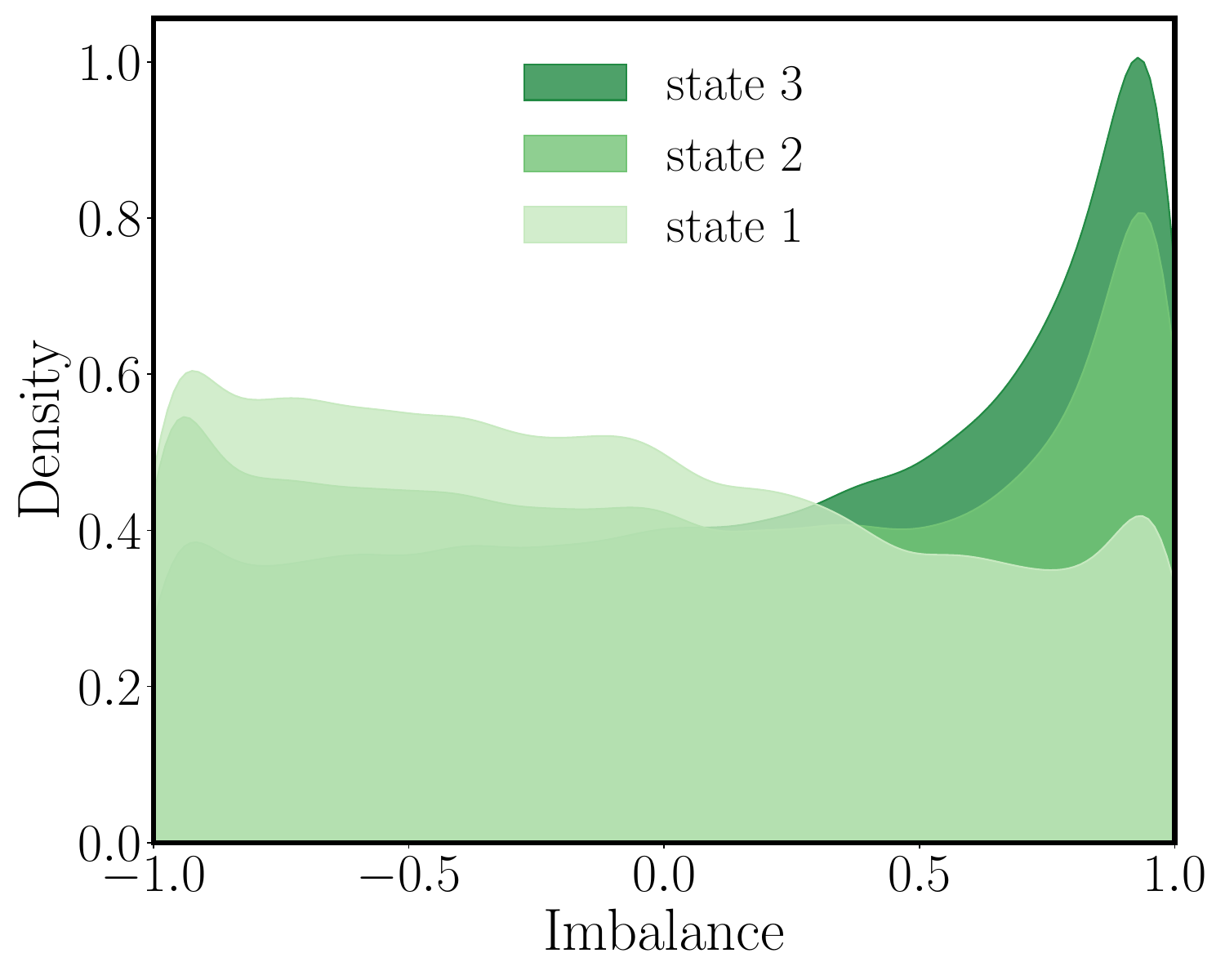}%
    }~~~
    \subfloat[Imbalance distribution before state transition, ask]{%
        \includegraphics[width=0.33\linewidth]{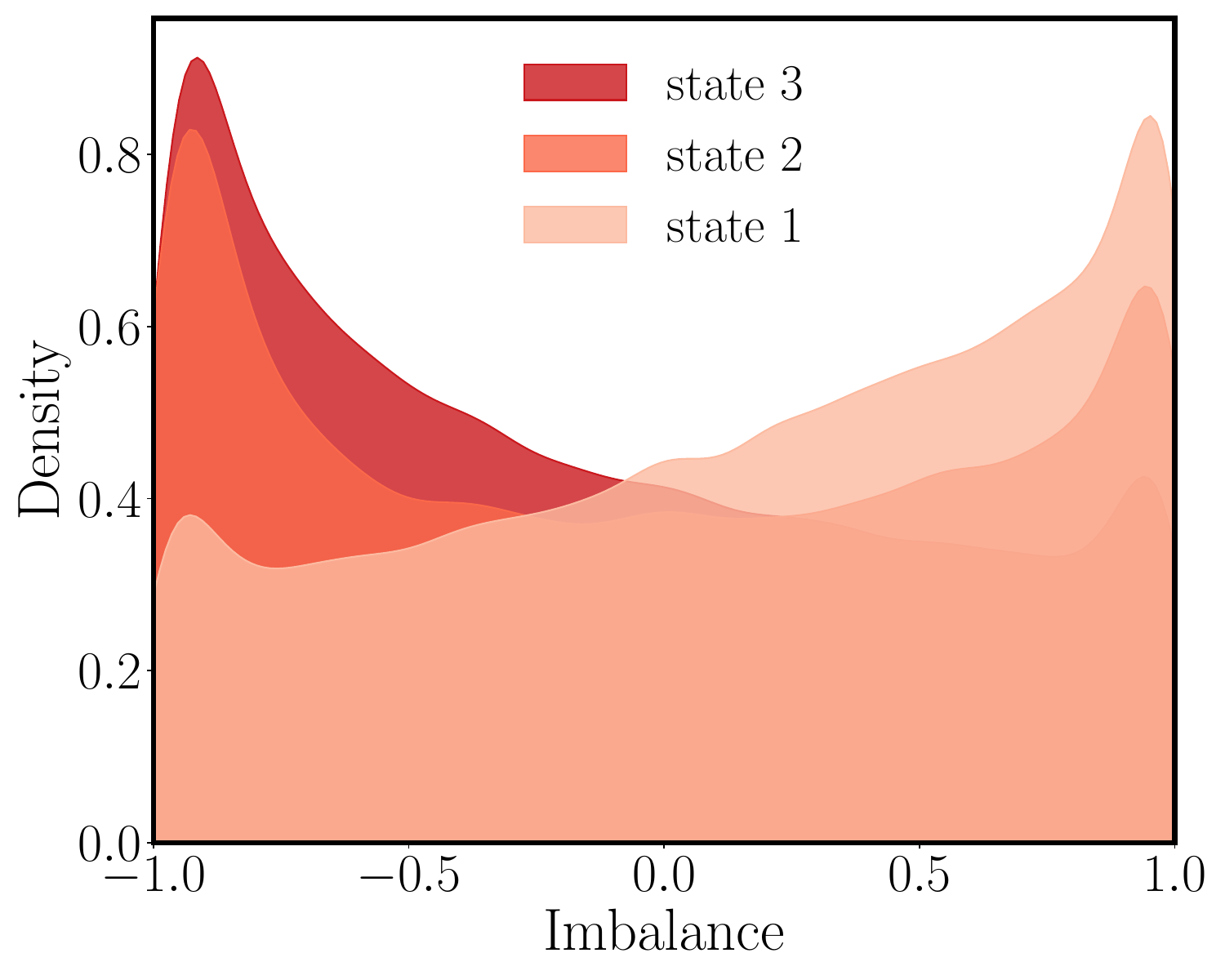}%
    }~~~
    \subfloat[Mid price response function after state transition with 95\% confidence intervals]{%
        \includegraphics[width=0.33\linewidth]{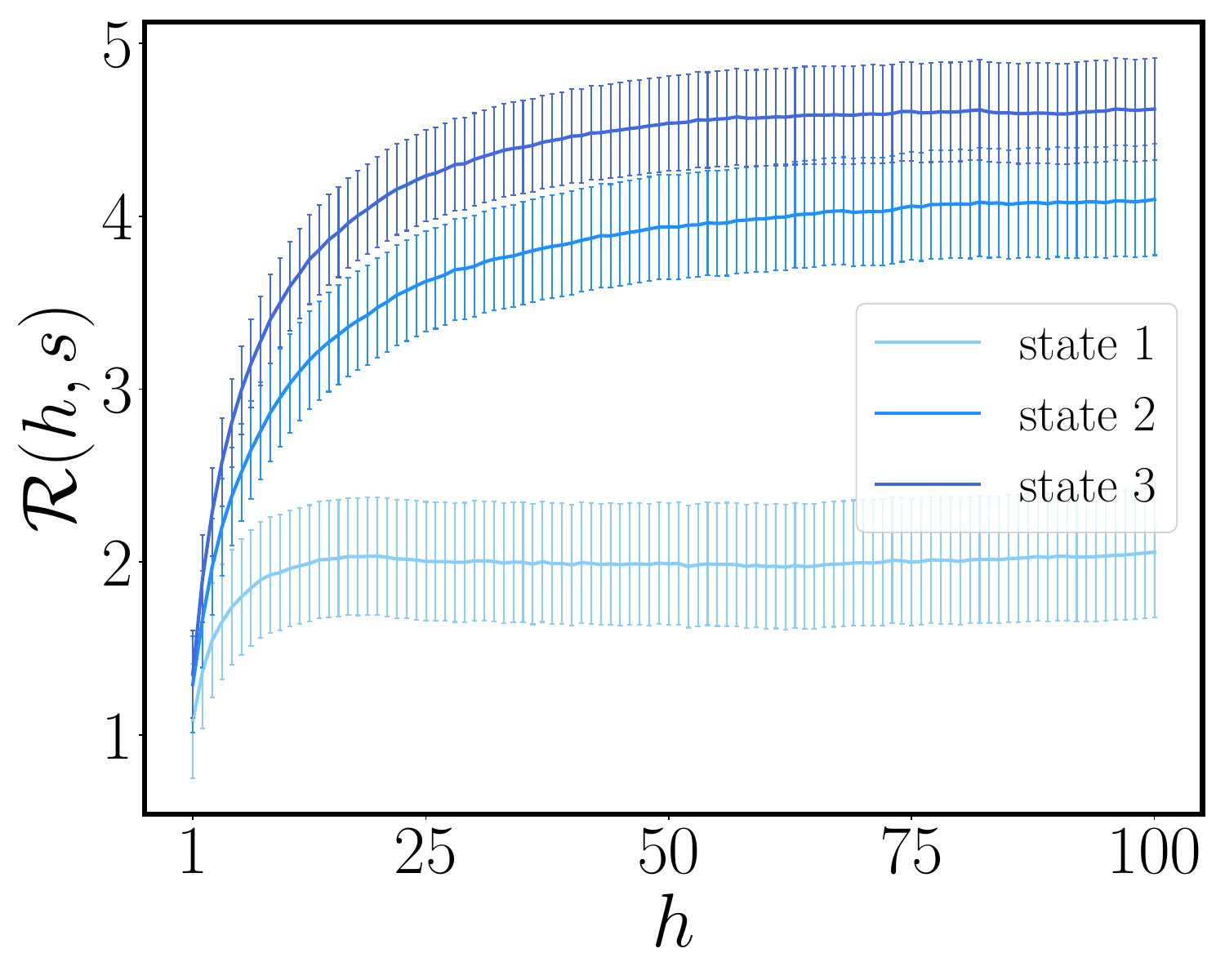}%
    }
    \caption{\textit{Suspicious trading activity} --- Market characterization of the trading activity regimes identified by the MMHP-$\delta$ model, when the model is trained on December 1st, 2023, without taking into account the intraday seasonality.}
    \label{fig:state_characterization_mmhp_wash_full_period}
\end{figure}

\end{document}